\documentclass[aps,pre,reprint,superscriptaddress,longbibliography,nofootinbib]{revtex4-2}

\pdfoutput=1  

\usepackage[utf8]{inputenc} 
\usepackage[T1]{fontenc}    
\usepackage{hyperref}       
\usepackage{url}            
\usepackage{booktabs}       
\usepackage{amsfonts}       
\usepackage{nicefrac}       
\usepackage{microtype}      
\usepackage{xcolor}         

\usepackage{bm}
\usepackage{amsmath,amssymb,mathtools}
\usepackage{graphicx} 
\usepackage{caption} 
\usepackage[subrefformat=parens]{subcaption}
\usepackage{listings}  
\usepackage{comment}

\RequirePackage{algorithm}
\RequirePackage{algorithmic}

\usepackage{amsthm}
\theoremstyle{plain}
\newtheorem{thm}{Theorem}
\newtheorem*{thm*}{Theorem}
\newtheorem{prop}{Proposition}
\newtheorem*{prop*}{Proposition}
\newtheorem{cor}{Corollary}
\newtheorem*{cor*}{Corollary}
\newtheorem{lem}{Lemma}
\newtheorem*{lem*}{Lemma}

\usepackage{amsmath,amsfonts,bm}

\def\eqref#1{equation~\ref{#1}}

\def\1{\bm{1}}

\def\eps{{\epsilon}}

\def\rb{{\textnormal{b}}}

\def\rh{{\textnormal{h}}}

\def\rt{{\textnormal{t}}}

\def\vb{{\bm{b}}}

\def\vh{{\bm{h}}}

\def\vt{{\bm{t}}}

\DeclareMathAlphabet{\mathsfit}{\encodingdefault}{\sfdefault}{m}{sl}
\SetMathAlphabet{\mathsfit}{bold}{\encodingdefault}{\sfdefault}{bx}{n}

\newcommand{\R}{\mathbb{R}}

\newcommand{\softmax}{\mathrm{softmax}}

\DeclareMathOperator*{\argmax}{arg\,max}

\let\ab\allowbreak

\newcount\K 
\def\blah{
    \K=0 \loop\ifnum\K<20 
    {\color[rgb]{0.8, 0.8, 0.8} blah blah blah blah blah blah blah blah blah blah} 
    \advance\K by1\repeat 
}

\usepackage{mleftright}
\mleftright

\newcommand{\game}{\Gamma} \newcommand{\sgame}{\Gamma_\sigma} \newcommand{\zgame}{\Gamma_{-1}}
\newcommand{\agents}{\mathcal{N}}
\newcommand{\actions}{\mathcal{A}}  
\newcommand{\payoffs}{\mathcal{U}}  
\newcommand{\payoffmat}[1]{U^{(#1)}}  
\newcommand{\sspace}{\mathcal{X}}  

\newcommand{\pspace}[1]{C^\infty_\mathrm{SC}(#1,\R)}  
\newcommand{\hreg}[1]{h_{#1}} \newcommand{\hdual}[1]{h_{#1}^\ast}

\newcommand{\xx}[1]{x^{#1}} \newcommand{\yy}[1]{y^{#1}} \newcommand{\XX}[1]{X^{#1}}
\newcommand{\strategies}{x}  
\newcommand{\xs}[1]{x^{#1 \ast}}  

\newcommand{\hftrl}{H_{\mathrm{FTRL}}}
\newcommand{\gsimplex}{G_S}
\newcommand{\gfenchel}{G_F}

\newcommand{\SO}{\mathrm{SO}} \newcommand{\rot}{R} \newcommand{\rott}{\mathcal{R}}
\newcommand{\sof}{\mathfrak{so}} \newcommand{\rotd}{K} \newcommand{\rottd}{\Omega}
\newcommand{\gangular}{G_K} \newcommand{\gangularr}{G_\Omega}
\newcommand{\gcumul}{G_C} \newcommand{\gdilation}{G_D}  

\makeatletter
\newsavebox{\@brx}
\newcommand{\llangle}[1][]{\savebox{\@brx}{\(\m@th{#1\langle}\)}%
  \mathopen{\copy\@brx\mkern2mu\kern-0.9\wd\@brx\usebox{\@brx}}}
\newcommand{\rrangle}[1][]{\savebox{\@brx}{\(\m@th{#1\rangle}\)}%
  \mathclose{\copy\@brx\mkern2mu\kern-0.9\wd\@brx\usebox{\@brx}}}
\makeatother

\newcommand{\hess}{\mathrm{Hess}}
\newcommand{\onevec}{\bm{1}}
\newcommand{\simplectic}{\omega}
\newcommand{\pbracket}[2]{\left\{ #1, #2 \right\}_{\mathrm{P}}}  

\newcommand{\step}{\epsilon}
\newcommand{\xd}[1]{\hat{x}^{#1}} \newcommand{\yd}[1]{\hat{y}^{#1}}  
\newcommand{\coeff}{\alpha}  

\newcommand{\defname}[1]{\emph{#1}}

\newcommand{\githubrepo}{https://github.com/Toshihiro-Ota/dftrl}

\newcommand{\eref}[1]{Eq.~(\ref{#1})}
\newcommand{\fref}[1]{Fig.~\ref{#1}}

\newcommand{\sref}[1]{Sec.~\ref{#1}}
\newcommand{\appenref}[1]{Appendix~\ref{#1}}

\newcommand{\thmref}[1]{Theorem~\ref{#1}}
\newcommand{\propref}[1]{Proposition~\ref{#1}}

\newcommand{\lemref}[1]{Lemma~\ref{#1}}

\newcommand{\dd}[2]{\frac{d #1}{d #2}}
\newcommand{\dpdp}[2]{\frac{\partial #1}{\partial #2}}

\newcommand{\set}[1]{\left\{ #1 \right\}}
\newcommand{\abs}[1]{\left\lvert #1 \right\rvert}
\newcommand{\norm}[1]{\left\lVert #1 \right\rVert}
\newcommand{\inner}[2]{\left\langle #1, #2 \right\rangle}

\newcommand{\Z}{\mathbb{Z}}

\begin{document}

\title{Hamiltonian of polymatrix zero-sum games}

\author{Toshihiro Ota}
    \email{ota\_toshihiro@cyberagent.co.jp}
    \affiliation{AI Lab, CyberAgent, Shibuya, Tokyo 150-0002, Japan}
    \affiliation{iTHEMS, RIKEN, Wako, Saitama 351-0198, Japan}
\author{Yuma Fujimoto}
    \email{fujimoto.yuma1991@gmail.com}
    \affiliation{AI Lab, CyberAgent, Shibuya, Tokyo 150-0002, Japan}
    \affiliation{RCIES, Soken University, Hayama, Kanagawa 240-0193, Japan}


\begin{abstract}

The understanding of a dynamical system's properties can be significantly advanced by establishing it as a Hamiltonian system and then systematically exploring its inherent symmetries.
By formulating agents' strategies and cumulative payoffs as canonically conjugate variables, we identify the Hamiltonian function that generates the dynamics of poly-matrix zero-sum games.
We reveal the symmetries of our Hamiltonian and derive the associated conserved quantities, showing how the conservation of probability and the invariance of the Fenchel coupling are intrinsically encoded within the system.
Furthermore, we propose the \emph{dissipation FTRL} (DFTRL) dynamics by introducing a perturbation that dissipates the Fenchel coupling, proving convergence to the Nash equilibrium and linking DFTRL to last-iterate convergent algorithms.
Our results highlight the potential of Hamiltonian dynamics in uncovering the structural properties of learning dynamics in games, and pave the way for broader applications of Hamiltonian dynamics in game theory and machine learning.

\end{abstract}

\maketitle  

\section{Introduction}
\label{sec:introduction}

Dynamical systems analysis plays a fundamental role in game theory.
In the context of learning in games \cite{fudenberg1998theory}, multiple agents dynamically update their strategies over time in an attempt to maximize their own payoffs.
The utility function of each agent depends not only on their own strategy but also on those of others, and simultaneous optimal strategies across all agents are characterized by the Nash equilibrium \cite{nash1950equilibrium}.
However, when the utilities of agents are in direct conflict (zero-sum games), the resulting learning dynamics often exhibit recurrent behaviors and convergence to a Nash equilibrium may fail to occur \cite{mertikopoulos2018cycles}.

The learning algorithm known as Follow the Regularized Leader (FTRL) provides a systematic formulation of dynamics in online learning \cite{shalev2006convex,mertikopoulos2018cycles}.
In the FTRL algorithm, each agent's strategy evolves according to their cumulative payoff, and it has been shown that in zero-sum games the Fenchel coupling is conserved throughout the dynamics.
This conservation law implies that, under the FTRL dynamics, the distance between the strategies and the Nash equilibrium remains essentially invariant, thereby explaining the failure of convergence to the equilibrium.
Although heuristic modifications to the FTRL algorithm have been proposed to mitigate this issue, a comprehensive understanding remains elusive.

Conserved quantities generally provide crucial insights into the system under consideration, as a sufficient number of conserved quantities can even render the system solvable.
Hamiltonian dynamics offers a sophisticated mathematical framework for systematically analyzing \emph{symmetries} and associated conservation laws.
Attempts to examine learning dynamics in games using a Hamiltonian perspective date back several decades \cite{hofbauer1996evolutionary}, yet most studies have rarely extended beyond two-agent scenarios.
Bailey and Piliouras \cite{bailey2019multi} made a notable observation that, in poly-matrix zero-sum games \cite{daskalakis2009network,cai2011minmax,cai2016zero}, it is promising that agents' strategies and their cumulative payoffs are thought of as canonical conjugates for formulating the FTRL dynamics as a Hamiltonian system.
Nevertheless, their approach resulted only in the reconstruction of existing conserved quantities such as the Fenchel coupling, without identifying a genuine Hamiltonian function for poly-matrix zero-sum games.\footnote{We make comments on the paper by Bailey and Piliouras in \appenref{append:bp}.}
In any case, previous studies have not fully exploited the paradigm of ``symmetry and conservation law'' inherent in Hamiltonian dynamics.

\begin{figure*}[t]
    \begin{minipage}[b]{0.495\linewidth}
        \centering
        \includegraphics[keepaspectratio, scale=0.24]{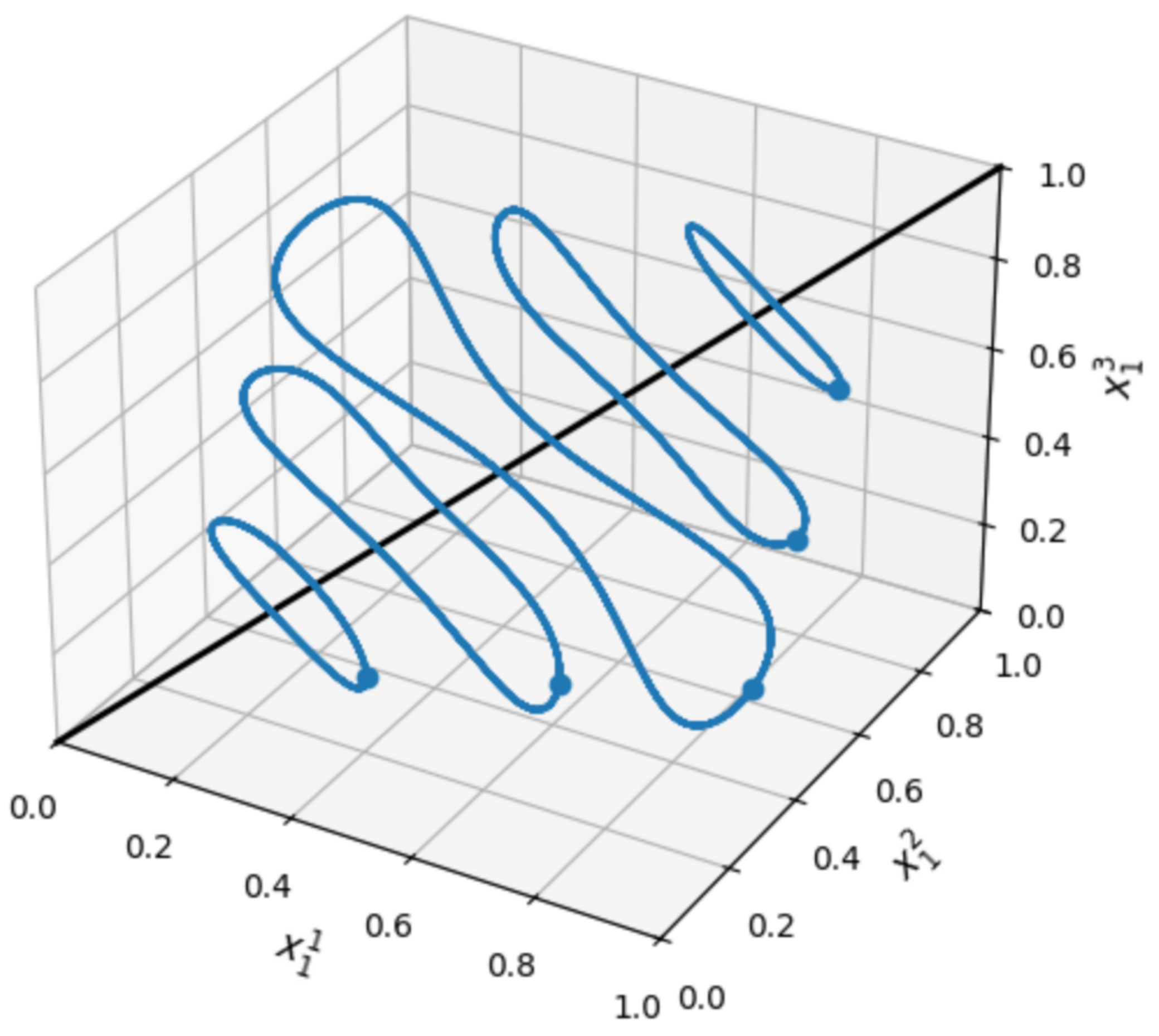}
        \subcaption{FTRL}  
        \label{fig:mp0}
    \end{minipage}
    \begin{minipage}[b]{0.495\linewidth}
        \centering
        \includegraphics[keepaspectratio, scale=0.24]{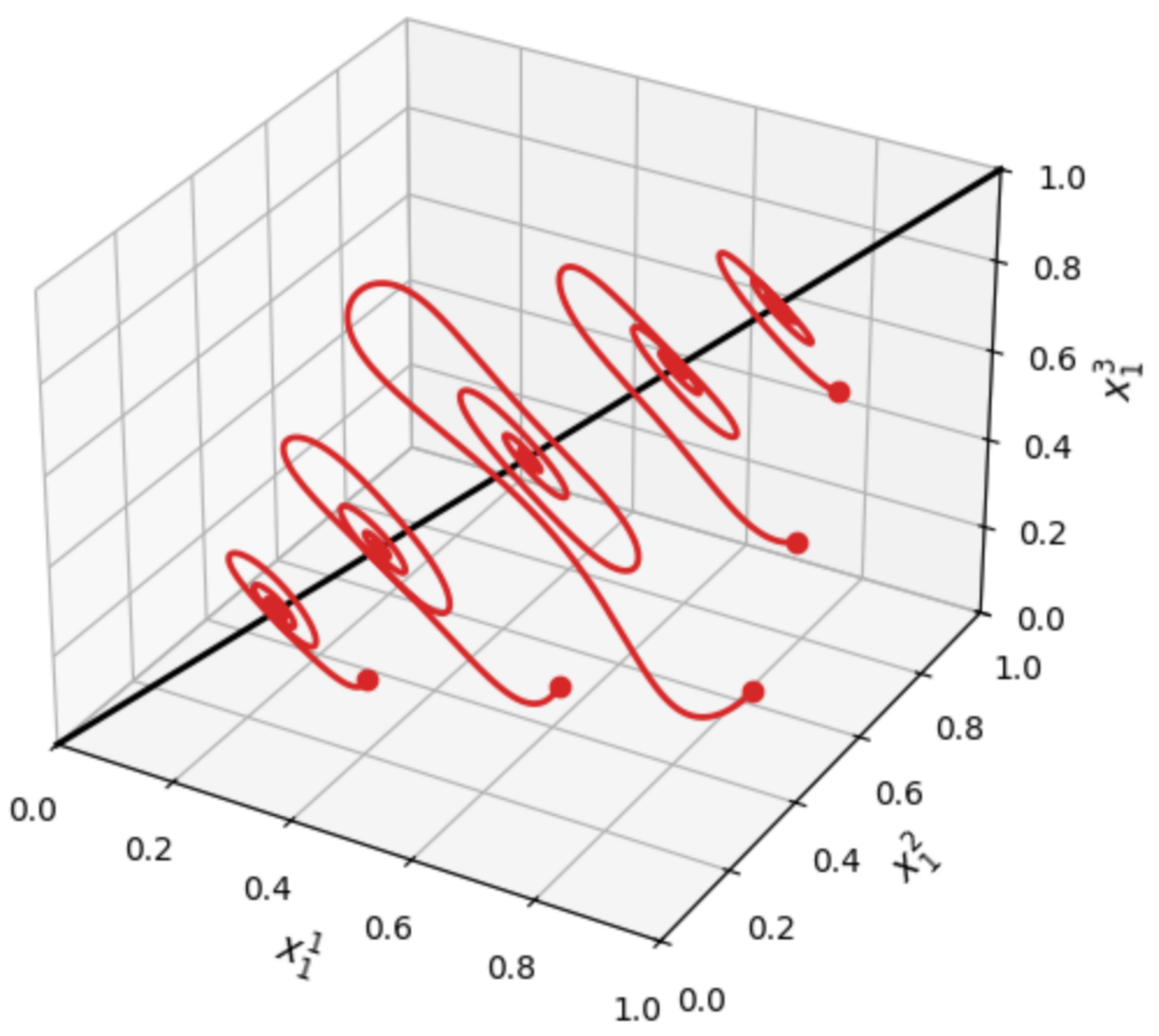}
        \subcaption{DFTRL}  
        \label{fig:mp1}
    \end{minipage}
    \caption{
        Three-player Matching Pennies.
        The diagonal straight line from $(0,0,0)$ to $(1,1,1)$ represents the Nash equilibria.
        The solution trajectories show (a) cyclic and (b) convergent behavior.
        }
    \label{fig:mp}
\end{figure*}

In this paper, we address the aforementioned issues by leveraging an analysis of symmetry, conservation laws, and their breaking within the framework of Hamiltonian dynamics.
Specifically, our contributions are as follows:
\begin{itemize}
    \item We establish the Hamiltonian function $\hftrl$ of the learning dynamics in poly-matrix zero-sum games, with agents' strategies and cumulative payoffs serving as canonically conjugate variables, and demonstrate that the Hamiltonian system generated by $\hftrl$ produces the FTRL dynamics.
    \item We elucidate the symmetries of the Hamiltonian $\hftrl$ and derive the associated conserved quantities, revealing explicitly how the conservation of probability in strategies and the temporal invariance of the Fenchel coupling are intrinsically encoded within $\hftrl$. We also examine angular momentum-like conservation laws, which have not been studied in the literature.
    \item We propose the \emph{dissipation FTRL} (DFTRL) dynamics by introducing a perturbation designed to dissipate (monotonically decrease) the Fenchel coupling along the trajectory. We prove that the DFTRL dynamics converges to the Nash equilibrium (as numerically illustrated in \fref{fig:mp}) and coincides with the continuous optimistic and extra-gradient FTRL algorithms in the limit of infinitesimal perturbations.
\end{itemize}

\subsection{Related work}

The analysis of learning dynamics in zero-sum games from the perspective of Hamiltonian dynamics has a long history \cite{hofbauer1996evolutionary,sato2002chaos,ostrovski2011piecewise,van2011hamiltonian,balduzzi2018mechanics,letcher2019differentiable}.
These studies typically adopt a formulation in which the canonical variables are given by a pair of strategies, thereby restricting their analysis to two-agent settings.
Moreover, the Hamiltonians proposed in these works are tailored to specific game dynamics, such as the replicator, projection, or best-response dynamics, and do not possess any symmetry giving rise to known conserved quantities.
Another line of research has focused on the Poisson structure of evolutionary games \cite{alishah2015hamiltonian}, and the conserved quantities of the game within that framework \cite{najafi2020conservative}.
While their Hamiltonian can also be defined for non-zero-sum games, it needs an equilibrium in its definition and is restricted to replicator dynamics.

To address the non-convergent behavior observed in FTRL dynamics, several modified versions of FTRL have been proposed.
Representative examples include the optimistic \cite{cai2022finitetime,daskalakis2019last,daskalakis2018training,syrgkanis2015fast,rakhlin2013optimization}, extra-gradient \cite{cai2022finitetime,lee2021fast,mertikopoulos2018optimistic,Korpelevich1976TheEM}, and negative-momentum \cite{hemmat2023lead,zhang2021suboptimality,kovachki2021continuous,gidel2019negative,polyak1964methods} FTRL algorithms.
While all these variants of FTRL are known to converge to the Nash equilibrium, their similarities and differences remain an active subject of discussion.
For instance, it has been reported that the extra-gradient FTRL converges even in time-varying games, where the payoff matrix evolves over time, whereas the other two do not \cite{feng2023last,feng2024last}.
In \sref{sec:dftrl}, we will demonstrate that the proposed DFTRL algorithm offers a unified perspective for classifying these existing variants in the continuum limit.

\section{Preliminaries}
\label{sec:preliminaries}

To fix the notation, in this section we describe a brief overview of poly-matrix games and the FTRL learning dynamics.
We basically follow the discussions in \cite{mertikopoulos2018cycles,bailey2019multi}.

We denote a graphical \defname{poly-matrix game} $\game$ by the following three-tuple:
\begin{equation}
    \game = (\agents, \actions, \payoffs),
\end{equation}
where $\agents=\set{1, \dots, n}$ represents the set of agents (players), $\actions=\set{A_1, \dots, A_n}$ are finite sets of actions that an agent $i$ takes,
\begin{equation}
    A_i = \set{a^i_1, \dots, a^i_{N_i +1}},  \quad  N_i \geq 0,
\end{equation}
and $\payoffs=\set{\payoffmat{ij}}_{i,j\in\agents}$ is a set of payoff matrices:
\begin{equation}
    \payoffmat{ij} \in \R^{\abs{A_i}\times\abs{A_j}},
    \quad
    \payoffmat{ii} \equiv 0.
\end{equation}
Each point of the $N_i$-simplex,  
\begin{equation}
    \sspace^{N_i} \coloneq \set{z \in \R^{\abs{A_i}}_{\geq 0} \,\middle|\, \sum_{a\in A_i} z_a = 1},
\end{equation}
is called a \defname{(mixed) strategy}, whose entries correspond to the probabilities of the actions that an agent may select.
In particular, a strategy is called \defname{fully mixed} when it lies in the relative interior of the simplex.
We collectively denote a set of $n$ strategies as
\begin{equation}
    \strategies = (\xx{1}, \dots, \xx{n}) \in \sspace^{N_1} \times \cdots \times \sspace^{N_n} \eqcolon \sspace.
\end{equation}

Given the mixed strategies of agents $i$ and $j$, $\xx{i}\in\sspace^{N_i}$ and $\xx{j}\in\sspace^{N_j}$, agent $i$ receives \defname{utility} from agent $j$: $(\xx{i})^\top \payoffmat{ij} \xx{j}$.
The total utility of agent $i$ is%
\footnote{
    The $j=i$ term can be consistently included in the sum, since by our definition of $\game$ the corresponding payoff matrix is null; $\payoffmat{ii}\equiv 0$, which implies that agents do not engage in self-play (do not have self-interaction).
    The same applies in the subsequent sections generally, although we occasionally remove the $j=i$ term from the sum for clarity.
    }
\begin{equation}
    u_i(\xx{i}; \xx{j(\neq i)}) = \sum_{j(\neq i)} (\xx{i})^\top \payoffmat{ij} \xx{j},
\end{equation}
which results in $n$ optimization problems:
\begin{equation}
    \max_{\xx{i}\in\sspace^{N_i}} u_i(\xx{i}; \xx{j(\neq i)})
        = \max_{\xx{i}\in\sspace^{N_i}} (\xx{i})^\top \sum_{j(\neq i)} \payoffmat{ij} \xx{j},
\end{equation}
for $i=1, \dots, n$.
In $\game$, $\strategies^\ast = \left( \xs{1}, \dots, \xs{n} \right) \in \sspace$ is called a \defname{Nash equilibrium} if
\begin{equation}
    \begin{aligned}
        &{}^{\forall}i\in\agents, ~ {}^{\forall}\xx{i}\in\sspace^{N_i},  \\
        &u_i(\xs{i}; (\xx{j(\neq i)})^\ast) \geq u_i(\xx{i}; (\xx{j(\neq i)})^\ast).
    \end{aligned}
    \label{eq:nash}
\end{equation}
That is, the Nash equilibrium is an $n$-simultaneous solution to the above optimization problem.

A game $\game$ is called \defname{zero-sum} (\defname{coordination}) if
\begin{equation}
    {}^{\forall}i,j\in\agents,
    \qquad
    \payoffmat{ji} = \sigma \big(\payoffmat{ij}\big)^\top,
\end{equation}
where $\sigma=-1$ $(+1)$.
We denote zero-sum or coordination poly-matrix games by $\sgame\coloneq (\agents, \actions, \payoffs, \sigma)$.
A significant nature of $\zgame$ is that the total utility of all the agents always vanishes due to the zero-sum property,
\begin{equation}
    \sum_{i\in\agents} u_i(\xx{i}; \xx{j(\neq i)})
        = \sum_{i,j\in\agents} (\xx{i})^\top \payoffmat{ij} \xx{j}
        = 0,
    \label{eq:total-utility}
\end{equation}
which implies that the total payoffs circulate among agents and are conserved within the game.

In the subsequent sections, we may refer to curves such as
\begin{equation}
    \phi : I \rightarrow \sspace^{N_i} \text{ or } \R^{\abs{A_i}}
\end{equation}
as the \defname{learning dynamics}, where $I$ is an interval, and write their images in the codomains by $\xx{i}(t)$ or $\yy{i}(t)$ for convenience.

\subsection{Follow the Regularized Leader}
\label{sec:ftrl}

Let $\pspace{M}$ ($\subset C^\infty(M,\R)$) be the set of strongly convex and $L$-smooth functions on a convex subset $M$ of Euclidean space.  
Given initial payoff vectors $\yy{i}(0)\in\R^{\abs{A_i}}$, the agents in $\sgame$ update their strategies by%
\footnote{
    In this paper, we focus on games of $\sgame$.
    However, the discussion in this subsection is generically applicable to $\game$, not only to $\sgame$.
    }
\begin{align}
    \yy{i}(t) &= \yy{i}(0) + \int_0^t \sum_{j(\neq i)} \payoffmat{ij} \xx{j}(s) ds,  \label{eq:yftrl}  \\
    \xx{i}(t) &= \argmax_{z\in\sspace^{N_i}} \left( \inner{z}{\yy{i}(t)} - \hreg{i}(z) \right),  \label{eq:xftrl}
\end{align}
where $\inner{\cdot}{\cdot}$ is the Euclidean inner product, and $\hreg{i}\in\pspace{\sspace^{N_i}}$ are regularizer functions.  
This learning rule is known as the Follow-the-Regularized-Leader dynamics \cite{mertikopoulos2018cycles}.
For the function $\hdual{i}:\R^{\abs{A_i}}\to\R$, which is the convex conjugate of $\hreg{i}$,  
\begin{equation}
    \hdual{i}(\yy{i}) = \max_{z\in\sspace^{N_i}} \left( \inner{z}{\yy{i}} - \hreg{i}(z) \right),
    \label{eq:hdual}
\end{equation}
it is known \cite{shalev2012online} that the right-hand side of \eref{eq:xftrl} can be expressed as the gradient of $\hdual{i}$:
\begin{equation}
    \xx{i}(t) = \nabla \hdual{i}(\yy{i})\mid_{\yy{i}=\yy{i}(t)}.
    \label{eq:xftrl-dual}
\end{equation}

To summarize, the FTRL dynamics is a dynamical system consisting of the following three components:
\begin{itemize}
    \item Setup of game: $\sgame = (\agents, \actions, \payoffs, \sigma)$
    \item Regularizers: a set of $\hreg{i}\in\pspace{\sspace^{N_i}}$ or their dual $\hdual{i}$
    \item Learning rule: Equation (\ref{eq:xftrl}) or its dual (\ref{eq:xftrl-dual}), with (\ref{eq:yftrl})
\end{itemize}
In particular, it is noted that this dynamical system is characterized by the choice of payoff matrices $\set{\payoffmat{ij}}_{i,j\in\agents}$ (the ``rule'' of a game)
and the choice of ``potential functions'' $\hreg{i}$ (or $\hdual{i}$).

\section{The Hamiltonian of poly-matrix zero-sum games}
\label{sec:hftrl}

We focus on poly-matrix zero-sum games $\zgame$ and show that the FTRL dynamics in $\zgame$ is understood as a Hamiltonian system.
We investigate the system through the lens of \emph{symmetry} and the conservation laws.
For readers unfamiliar with Hamiltonian dynamics, see \appenref{append:hamiltonian}.

\subsection{FTRL dynamics as Hamiltonian systems}  

As noted in \eref{eq:total-utility}, the total utility in zero-sum games is a constant over time and thus could serve as a Hamiltonian function for the learning dynamics.
With this viewpoint, let us first observe that, by taking time derivatives, the FTRL dynamics in Eqs.~(\ref{eq:yftrl}) and (\ref{eq:xftrl-dual}) can be rewritten as follows:
\begin{align}
    \dd{\yy{i}(t)}{t}
        &= \sum_{j\in\agents} \payoffmat{ij} \xx{j}(t)  \nonumber \\
        &= \sum_{j\in\agents} \payoffmat{ij} \nabla \hdual{j}(\yy{j})\mid_{\yy{j}=\yy{j}(t)}  \label{eq:dyftrl},  \\
    \dd{\xx{i}(t)}{t}
        &= \nabla\nabla \hdual{i}(\yy{i})\mid_{\yy{i}=\yy{i}(t)} \dd{\yy{i}(t)}{t}  \nonumber  \\
        &= \hess\big(\hdual{i}(\yy{i}(t))\big) \sum_{j\in\agents} \payoffmat{ij} \xx{j}(t)  \label{eq:dxftrl},
\end{align}
where we have used \eref{eq:xftrl-dual} and $\hess$ is the Hessian.
A point here is that the right-hand sides of these equations are Lipschitz continuous, owing to the strong convexity of the regularizer functions.
From the standard existence and uniqueness theorem of initial value problems, Eqs.~(\ref{eq:dyftrl}) and (\ref{eq:dxftrl}) on $I$ admit a unique solution that is indeed the FTRL dynamics Eqs.~(\ref{eq:yftrl}) and (\ref{eq:xftrl-dual}).
Hence, these differential equations are equivalent to the FTRL dynamics.

Now, we define the phase space of $\zgame$ by $M_{\zgame}=T^\ast\R^{\sum_i \abs{A_i}}_{\geq 0}\simeq\R^{\sum_i \abs{A_i}}_{\geq 0}\times\R^{\sum_i \abs{A_i}}$ equipped with the symplectic structure $\simplectic\simeq J$ as described in \appenref{append:hamiltonian}.
For a given choice of $\hdual{i}\in\pspace{\R^{\abs{A_i}}}$, we define the \defname{FTRL Hamiltonian} $\hftrl: M_{\zgame}\to\R$ by
\begin{equation}
    \hftrl(\strategies, y) = \sum_{i,j\in\agents} \inner{\nabla \hdual{j}(\yy{j})}{\payoffmat{ji} \xx{i}}.
    \label{eq:ftrl-hamiltonian}
\end{equation}

\begin{thm}
    \label{thm:ftrl-hamiltonian}
    The FTRL dynamics in the form of Eqs.~(\ref{eq:dyftrl}) and (\ref{eq:dxftrl}) is a Hamiltonian system defined by \eref{eq:ftrl-hamiltonian}.
\end{thm}

\begin{proof}
    What we need to show is that the flow equation on $M_{\zgame}$ generated by \eref{eq:ftrl-hamiltonian} accompanied with the symplectic structure reduces to the FTRL dynamics (\ref{eq:dyftrl}) and (\ref{eq:dxftrl}).
    First, the flow equation for $\strategies(t)$ (``position'') induced by \eref{eq:ftrl-hamiltonian} is
    \begin{equation}
        \begin{aligned}[b]
            \dd{\xx{k}(t)}{t}
                &= \nabla_{\yy{k}} \hftrl(\strategies, y)\mid_{x=\strategies(t),\, y=y(t)}  \\
                &= \sum_{i,j\in\agents} \delta^{kj} \left( \nabla\nabla \hdual{j}(\yy{j}(t)) \right)^\top \payoffmat{ji} \xx{i}(t)  \\
                &= \hess\big(\hdual{k}(\yy{k}(t))\big) \sum_{i\in\agents} \payoffmat{ki} \xx{i}(t),
        \end{aligned}
    \end{equation}
    which is equivalent to \eref{eq:dxftrl}.
    Next, by utilizing the symmetricity of the inner product, $\inner{\alpha}{\beta}=\inner{\beta}{\alpha}$, one finds the flow equation for $y(t)$ (``momentum''),
    \begin{equation}
        \begin{aligned}[b]
            \dd{\yy{k}(t)}{t}
                &= -\nabla_{\xx{k}} \hftrl(\strategies, y)\mid_{x=\strategies(t),\, y=y(t)}  \\
                &= -\sum_{i,j\in\agents} \delta^{ki} \big(\payoffmat{ji}\big)^\top \nabla \hdual{j}(\yy{j}(t))  \\
                &= \sum_{j\in\agents} \payoffmat{kj} \nabla \hdual{j}(\yy{j}(t)),
        \end{aligned}
    \end{equation}
    where the last equality follows from the assumption that our game is zero-sum: $\big(\payoffmat{jk}\big)^\top = -\payoffmat{kj}$ for all $j,\, k\in\agents$.
\end{proof}

The aforementioned Lipschitz continuity ensures that the Hamiltonian system of $\hftrl$ describes the FTRL dynamics in poly-matrix zero-sum games.
The energy conservation law in this Hamiltonian system leads to the conservation law of the total utility of the game $\zgame$:
\begin{equation}
    \begin{aligned}[b]
        E_{\text{FTRL}}(t)
            &= \hftrl(\strategies(t), y(t))  \\
            &= \sum_{i,j\in\agents} \inner{\xx{i}(t)}{\payoffmat{ij} \xx{j}(t)}
            \equiv 0,
    \end{aligned}
    \label{eq:energy-conservation}
\end{equation}
which is the equivalent of \eref{eq:total-utility}, as expected.
Now that the learning dynamics in $\zgame$ is characterized by our proposed Hamiltonian (\ref{eq:ftrl-hamiltonian}), we investigate the system in terms of symmetry and the conservation laws in the following subsections.

\subsection{Symmetry and conservation law from the FTRL Hamiltonian}  

The Hamiltonian (\ref{eq:ftrl-hamiltonian}) is defined on the symplectic manifold $M_{\zgame}\simeq\R^{\sum_i \abs{A_i}}_{\geq 0}\times\R^{\sum_i \abs{A_i}}$.
The base manifold (configuration space) $\R^{\sum_i \abs{A_i}}_{\geq 0}$ seems too large to be the space of strategies, since the strategies $\xx{i}$ in games should be thought of as probabilities of actions and have to sum to one: $\sum_a \xx{i}_a = 1$.
Interestingly, this condition is implemented as the \emph{symmetry} of our Hamiltonian and its conservation law, as the following proposition demonstrates.

\begin{prop}
    \label{prop:simplex}
    If $\nabla\hdual{i}(\yy{i})$ are translationally invariant for all $i\in\agents$, then the solutions $\strategies(t)$ to the Hamiltonian system of $\hftrl$ are constrained onto the probability simplices $\sspace$.
\end{prop}

The proof follows straightforwardly from the following lemma:

\begin{lem}
    \label{lem:simplex}
    If a Hamiltonian function on $M_{\zgame}$ has the translation symmetry in $\yy{i}\in\R^{\abs{A_i}}$, then the following sum of the elements of the corresponding conjugate variable is a conserved quantity:
    \begin{equation}
        \gsimplex^i(\strategies, y) = \sum_{a\in A_i} \xx{i}_a.
    \end{equation}
\end{lem}

This lemma is derived from Noether's theorem; see \appenref{append:hamiltonian} and \ref{append:proof-simplex}.

\begin{proof}[Proof of \propref{prop:simplex}]
    Given translationally invariant $\nabla\hdual{i}(\yy{i})$ for all $i\in\agents$, the FTRL Hamiltonian (\ref{eq:ftrl-hamiltonian}) has the translation symmetry in $\yy{i}\in \R^{\abs{A_i}}$.
    From \lemref{lem:simplex}, in this Hamiltonian system we have the conserved quantities
    \begin{equation}
        \gsimplex^i(\strategies(t), y(t)) = \sum_{a\in A_i} \xx{i}_a(t) = \text{const.},
    \end{equation}
    for all $i$.
    Suppose that the initial conditions $\sum_a \xx{i}_a(0)=1$ are imposed, then the solution trajectory of the Hamiltonian system generated by \eref{eq:ftrl-hamiltonian} remains restricted to $\strategies{(t)}\in\sspace$, completing the proof.  
\end{proof}

We find that the Hamiltonians for both the entropic regularizer (the replicator dynamics \cite{taylor1978evolutionary,taylor1979evolutionarily,schuster1983replicator}) and the Euclidean regularizer (the projection dynamics \cite{friedman1991evolutionary,nagurney1997projected,lahkar2008projection,sandholm2008projection,mertikopoulos2016learning})
have translation symmetry in all the $\yy{i}\in\R^{\abs{A_i}}$, implying that they yield reasonable Hamiltonian dynamics in games:

\begin{cor}
    In the Hamiltonian systems of $\hftrl$ for the entropic regularizer and the Euclidean regularizer, the solution trajectories of $\strategies(t)$ are constrained onto the probability simplices $\sspace$.
\end{cor}

\begin{proof}
    From \propref{prop:simplex}, it suffices to show the translation invariance of the corresponding gradients of the dual regularizers, $\nabla\hdual{i}$.
    For the entropic regularizer, we have
    \begin{equation}
        \nabla \hdual{i}(\yy{i})_a = \frac{\exp(\yy{i}_a)}{\sum_{a'} \exp(\yy{i}_{a'})},
    \end{equation}
    and for the Euclidean regularizer,
    \begin{equation}
        \nabla \hdual{i}(\yy{i})_a = \yy{i}_a - \frac{1}{\abs{A_i}} \left( \sum_{a'} \yy{i}_{a'} - 1 \right).
    \end{equation}
    These explicit formulae and other details can be found in \appenref{append:properties}.
    One finds that the right-hand sides of these gradients are invariant under the constant shift in $\yy{i}$, $y'^i_a=\yy{i}_a+\epsilon$ for all $a\in A_i$, which proves the statement.
\end{proof}

In addition, for the entropic and Euclidean regularizers, we know that $\sum_a \nabla\hdual{i}(\yy{i})_a = 1$, which leads us to another conserved quantity.

\begin{prop}
    \label{prop:fenchel}
    Suppose we have a game $\zgame$ with the Hamiltonian (\ref{eq:ftrl-hamiltonian}).
    If a fully-mixed Nash equilibrium exists and the elements of $\nabla\hdual{j}(\yy{j})$ sum to a constant independent of $j$, then the following function is a conserved quantity of the Hamiltonian system:
    \begin{equation}
        \gfenchel(\strategies, y) = \sum_{i\in\agents} \left( \hreg{i}(\xs{i}) + \hdual{i}(\yy{i}) - \inner{\yy{i}}{\xs{i}} \right).
    \end{equation}
\end{prop}

\begin{proof}[Proof sketch]
    The proof is derived by applying Noether's theorem to the Hamiltonian (\ref{eq:ftrl-hamiltonian}) along with the symmetry under the canonical transformation,
    \begin{equation}
        x'^i = x^i + \nabla\hdual{i}(\yy{i}) - \xs{i},  \quad  y'^i = y^i,
        \label{eq:symmetry-fenchel}
    \end{equation}
    for all $i\in\agents$.
    Using the assumptions, one can confirm that the FTRL Hamiltonian $\hftrl$ is invariant under this transformation:
    \begin{equation}
        \hftrl(\strategies', y') = \hftrl(\strategies, y).
    \end{equation}
    It turns out that the generator of this canonical transformation takes the form
    \begin{equation}
        g_F(\strategies, y) = \sum_{i\in\agents} \left( \hdual{i}(\yy{i}) - \inner{\yy{i}}{\xs{i}} \right) + \text{const.}
    \end{equation}
    Since $\sum_i \hreg{i}(\xs{i})$ is a constant, from Noether's theorem it follows that the function $\gfenchel$ is the conserved quantity in the Hamiltonian system of $\hftrl$.
\end{proof}

The complete proof is provided in \appenref{append:proof-fenchel}.
In the literature, the following function,
\begin{equation}
    F(p, y) = \sum_{i\in\agents} \left( \hreg{i}(p^i) + \hdual{i}(\yy{i}) - \inner{\yy{i}}{p^i} \right),
\end{equation}
is referred to as the Fenchel coupling \cite{mertikopoulos2016learning,mertikopoulos2019learning}.
The Fenchel coupling at $p=\strategies^\ast$, which measures the distance to the Nash equilibrium, is a well-known conserved quantity of the FTRL dynamics in zero-sum games \cite{mertikopoulos2018cycles}.
\propref{prop:fenchel} shows that our Hamiltonian (\ref{eq:ftrl-hamiltonian}) involves the Fenchel coupling as a conserved quantity associated with the corresponding symmetry.
Due to this, we also refer to $\gfenchel$ as the Fenchel coupling.

If we interpret the symmetry transformation (\ref{eq:symmetry-fenchel}), it can be viewed as a coordinate shift under which the total utility of the system remains invariant when each agent's `current strategy' $\xx{i}$ is adjusted by the difference between its `regularized best response' $\nabla\hdual{i}(\yy{i})$ and the equilibrium $\xs{i}$.
If an agent is at equilibrium ($\nabla\hdual{i}(\yy{i})=\xs{i}$), no shift is applied; otherwise, the strategy $x$ is adjusted by the magnitude of this error.
(Due to the zero-sum property, the Hamiltonian (\ref{eq:ftrl-hamiltonian}) is insensitive to shifts in $x$ corresponding to the difference from the best response.)
This shift coincides exactly with the gradient of the Fenchel coupling.
Then, by Noether's theorem, the Fenchel coupling, which serves as the generating function for this shift, must remain constant.
This is analogous to the conservation of momentum in physics, where the symmetry is invariance under spatial translations (shifts).

Lastly, we propose two ``angular momentum'' conservation laws for the Hamiltonian systems of $\hftrl$.
The key points are the commutativity of rotation matrices and the payoff matrices, and the equivariance of $\nabla\hdual{i}$ under the rotation.

\begin{prop}
    \label{prop:angular}
    Suppose we have a game $\zgame$ with the Hamiltonian (\ref{eq:ftrl-hamiltonian}).
    \begin{enumerate}
        \item Let $\rot_i\in\SO(\abs{A_i}),~ i\in\agents$, be rotation matrices.
            If the followings hold for all $i,j\in\agents$,
            \begin{align}
                \payoffmat{ij} \rot_j &= \rot_i \payoffmat{ij},  \\
                \nabla \hdual{i}(\rot_i \yy{i}) &= \rot_i \nabla \hdual{i}(\yy{i}),
            \end{align}
            then the following function is a conserved quantity of the Hamiltonian system:
            \begin{equation}
                \gangular(\strategies, y) = \sum_{i\in\agents} \inner{\xx{i}}{\rotd_i \yy{i}},
            \end{equation}
            where $\rot_i = e^{\epsilon \rotd_i}$ and $\rotd_i\in\sof(\abs{A_i})$.  
        \item Let $\rott\in\SO(n)$ be a rotation matrix.
            Suppose that all the action sets and the regularizers are identical for all the agents: $A_i\simeq A$ and $\hdual{i}=\hdual{}$ for all $i\in\agents$.
            If the followings hold for all $i,j\in\agents$,
            \begin{align}
                U \rott &= \rott U,  \\
                \nabla \hdual{}(\rott \yy{}) &= \rott \nabla \hdual{}(\yy{}),
            \end{align}
            where $U$ is the block payoff matrix with the entries $U_{ij}=\payoffmat{ij}$,
            then the following function is a conserved quantity of the Hamiltonian system:
            \begin{equation}
                \gangularr(\strategies, y)
                    = \sum_{i,j\in\agents} \inner{\xx{i}}{\Omega_{ij} \yy{j}}
                    \eqcolon \llangle \xx{}, \Omega \yy{} \rrangle,
            \end{equation}
            where $\rott = e^{\epsilon \Omega}$ and $\Omega\in\sof(n)$.  
    \end{enumerate}
\end{prop}

The proof is provided in \appenref{append:proof-angular}.
This proposition indicates that there may be game-theoretic analogues of the conservation laws for angular momentum in physics.
The first one shows the conservation of the total ``angular momentum'' of each agent, and the second corresponds to the conservation of the system's total angular momentum.
We currently do not have a clear intuitive interpretation of these conserved quantities in game theory and are not aware of any existing literature on this topic.
This result suggests that the Hamiltonian (\ref{eq:ftrl-hamiltonian}) may generate new conserved quantities that are not recognized in the standard FTRL dynamics context.
These results remain interesting from a mathematical standpoint and may provide a deeper understanding of learning dynamics in games.

We here provide a useful expression of \propref{prop:angular} as a corollary.

\begin{cor}
    \label{cor:angular}
    Suppose we have a game $\zgame$ with the Hamiltonian (\ref{eq:ftrl-hamiltonian}), and assume that the equivariance condition holds for below cases.
    \begin{enumerate}
        \item Suppose that all the action sets are identical, $A_i\simeq A$ for all $i\in\agents$.
            If the payoff matrices are given by a single skew-symmetric matrix $B\in\mathfrak{so}(\abs{A})$ as,
            \begin{equation}
                \payoffmat{ij} = c B,  \quad  B^\top = -B,
            \end{equation}
            where $c\in\R$ is a constant, then the rotation matrices $\rot_i=e^{\epsilon B}\in\SO(\abs{A})$ commute with the payoff matrices and the following quantity is conserved:
            \begin{equation}
                \gangular(\strategies, y) = \sum_{i\in\agents} \inner{\xx{i}}{B \yy{i}}.
            \end{equation}
        \item If the payoff matrices are given by the tensor product of a skew-symmetric matrix $\Sigma =(\sigma_{ij}) \in\sof(n)$ and a symmetric matrix $S$ as,
            \begin{equation}
                \payoffmat{ij} = \sigma_{ij} S,  \quad  S^\top = S,  
            \end{equation}
            then the rotation matrix $\rott=e^{\epsilon \Sigma}\in\SO(n)$ commutes with the payoff matrix $U$ and the following quantity is conserved:
            \begin{equation}
                \gangularr(\strategies, y) = \llangle \xx{}, \Sigma \yy{} \rrangle.
            \end{equation}
    \end{enumerate}
\end{cor}

One should note that the payoff matrices in these cases, $\payoffmat{ij} = c B$ and $\payoffmat{ij} = \sigma_{ij} S$, satisfy the zero-sum condition, $(\payoffmat{ij})^\top = -\payoffmat{ji}$, due to the skew-symmetry of $B$ and the symmetry of $S$.
The payoff matrices in the first case define a symmetric game, $\payoffmat{ij} = \payoffmat{ji}$, while those in the second case define a \defname{skew-symmetric game}, $\payoffmat{ij} = -\payoffmat{ji}$.
Thus, roughly speaking, a class of symmetric zero-sum games has the rotation symmetry in the strategy space with the corresponding angular momentum(-like) conservation law, while a class of skew-symmetric zero-sum games has the rotation symmetry among agents with another angular momentum conservation law.
In \appenref{append:proof-angular}, we demonstrate the explicit examples of these angular momentum conservation laws; two-player Rock-Paper-Scissors as a symmetric game and two-player Matching Pennies as a skew-symmetric game.

\subsection{Dissipation FTRL dynamics}
\label{sec:dftrl}

We now have the conserved quantity $\gfenchel$ of the Hamiltonian system, which measures the distance from a fully-mixed Nash equilibrium.
However, it implies that the learning dynamics does not converge to the Nash equilibrium, but rather shows a cyclic behavior around it.
To resolve this issue, we introduce a perturbation to the FTRL dynamics that breaks the conservation law for $\gfenchel$, so that the learning dynamics converges to the Nash equilibrium, while the Hamiltonian itself (the total utility) is kept as in \eref{eq:energy-conservation} to sustain the global zero-sum property of the game.

We consider a modified FTRL dynamics by adding perturbation terms to Eqs.~(\ref{eq:dyftrl}) and (\ref{eq:dxftrl}), with $\coeff\in\R_{\geq0}$,
\begin{equation}
    \begin{aligned}
        \dd{\xx{i}(t)}{t}
            &= \hess\big(\hdual{i}(\yy{i}(t))\big) \sum_{j\in\agents} \payoffmat{ij} \xx{j}(t) + \coeff f^i(x(t), y(t)),  \\
        \dd{\yy{i}(t)}{t}
            &= \sum_{j\in\agents} \payoffmat{ij} \nabla \hdual{j}(\yy{j}(t)) + \coeff g^i(x(t), y(t)).
    \end{aligned}
    \label{eq:dftrl}
\end{equation}
We observe that the energy conservation law \eref{eq:energy-conservation} holds as long as the learning dynamics satisfies the relations $\xx{i}(t) = \nabla\hdual{i}(\yy{i}(t))$.
From this observation, one can confirm that the FTRL Hamiltonian remains conserved along the solutions to the modified FTRL dynamics \eref{eq:dftrl} under the perturbation relations,
\begin{equation}
    f^i(x, y) = \hess\big(\hdual{i}(\yy{i})\big) g^i(x, y).
    \label{eq:consistency}
\end{equation}

For generic perturbations $f$ and $g$ that do not satisfy \eref{eq:consistency}, the solution is not guaranteed to take the form $\xx{i}(t)=\nabla\hdual{i}(\yy{i}(t))$, and consequently the Hamiltonian is not necessarily conserved.
From the viewpoint of Hamiltonian conservation, there is freedom in the choice of $g$ (or $f$).
Leveraging this freedom, we can deform the dynamics while preserving the Hamiltonian, i.e., the zero-sum nature of the game.
This property is a significant consequence of the FTRL Hamiltonian (\ref{eq:ftrl-hamiltonian}); in general Hamiltonian systems, adding a perturbation such as \eref{eq:dftrl} would typically break conservation of the Hamiltonian.

\begin{thm}
    \label{thm:dftrl}
    Suppose we have a game $\zgame$.
    If a fully-mixed Nash equilibrium exists and the elements of $\nabla\hdual{j}(\yy{j})$ sum to a constant independent of $j$,
    then the modified FTRL dynamics \eref{eq:dftrl} together with the following perturbations for all $i\in\agents$ under the relations \eref{eq:consistency} lead to convergence to the Nash equilibrium:
    \begin{equation}
        g^i(x, y) = \sum_{j\in\agents} \payoffmat{ij} \hess\big(\hdual{j}(\yy{j})\big) \sum_{k\in\agents} \payoffmat{jk} \xx{k}.
        \label{eq:perturbation}
    \end{equation}
\end{thm}

\begin{proof}[Proof sketch]
    We show that the Fenchel coupling $\gfenchel$ is monotonically decreasing under the modified FTRL dynamics Eqs.~(\ref{eq:dftrl}) with (\ref{eq:consistency}) and (\ref{eq:perturbation}).
    Using the assumptions, the time evolution of the Fenchel coupling turns out to be
    \begin{equation}
        \begin{aligned}[b]
            \frac{d}{dt}\gfenchel\left( \strategies(t), y(t) \right)
            = \coeff \sum_{i\in\agents} \inner{g^i(\strategies(t), y(t))}{\xx{i}(t) - \xs{i}}  \\
                = -\coeff \sum_{j} \bigg\langle \hess\big(\hdual{j}(\yy{j}(t))\big) \sum_{k} \payoffmat{jk} \left(\xx{k}(t) - \xs{k}\right),  \\
                    {\sum_{i} \payoffmat{ji} \left(\xx{i}(t) - \xs{i}\right)} \bigg\rangle.
        \end{aligned}
    \end{equation}
    Note that $\strategies(t)$ and $y(t)$ here are the solution to \eref{eq:dftrl}.
    Since $\coeff\geq 0$ and the Hessian of the dual regularizers $\hdual{j}$ is positive-definite, we obtain
    \begin{equation}
        \frac{d}{dt}\gfenchel\left( \strategies(t), y(t) \right) \leq 0.
    \end{equation}
    The equality is reached by $\coeff=0$, which reduces to the ordinary FTRL dynamics.
    Thus, in the modified FTRL dynamics for $\coeff>0$, we have the convergent solutions to the Nash equilibrium.
\end{proof}

We illustrate the complete proof in \appenref{append:proof-dftrl}.
It is noted that, from the relations (\ref{eq:consistency}), the learning dynamics turns out to satisfy $\xx{i}(t) = \nabla\hdual{i}(\yy{i}(t))$.
Due to this dissipative property of the Fenchel coupling $\gfenchel$, we refer to the proposed modified FTRL dynamics \eref{eq:dftrl} with the relations (\ref{eq:consistency}) and (\ref{eq:perturbation}) as the \defname{dissipation FTRL} dynamics.
The DFTRL dynamics generically converges to the Nash equilibrium under the perturbations \eref{eq:perturbation}, which can be understood as an effective tuning of the payoff matrices of the game:
\begin{equation}
    \begin{aligned}
        &\xx{i}(t) = \nabla\hdual{i}(\yy{i}(t)),  \\
        &\dd{\yy{i}(t)}{t}
            = \sum_{j\in\agents} \payoffmat{ij} \nabla \hdual{j}(\yy{j}(t)) + \coeff g^i(x(t), y(t))  \\
            &= \sum_{j\in\agents} \left( \payoffmat{ij} + \coeff \sum_{k\in\agents} \payoffmat{ik} \hess\big(\hdual{k}(\yy{k}(t))\big) \payoffmat{kj} \right) \xx{j}(t).
    \end{aligned}
    \label{eq:dftrl2}
\end{equation}
The algorithm of DFTRL is interpreted as follows.
The agents play a game $\zgame$ with the payoff matrices satisfying the \emph{local} zero-sum property, $\left(\payoffmat{ij}\right)^\top=-\payoffmat{ji}$.
Meanwhile, a ``game master'' secretly tunes the actual payoffs by \eref{eq:perturbation} to converge the learning dynamics while maintaining the \emph{global} zero-sum property \eref{eq:energy-conservation}, namely, the conservation of the total utility.
We will demonstrate in numerical simulations in the next section that the DFTRL dynamics really converges to the Nash equilibrium.
To close this section, we would like to mention the relation between our DFTRL algorithm and the known last-iterate convergent FTRL algorithms.

\begin{prop}
    \label{prop:co-ceg}
    For small $\coeff$, the DFTRL dynamics is equivalent to the first-order expansion of the continuous optimistic \cite{cai2022finitetime,daskalakis2019last,daskalakis2018training,syrgkanis2015fast,rakhlin2013optimization}
    and the continuous extra-gradient \cite{cai2022finitetime,lee2021fast,mertikopoulos2018optimistic,Korpelevich1976TheEM} FTRL dynamics.
\end{prop}

The proof is provided in \appenref{append:cftrl}.
In the literature, the optimistic and extra-gradient FTRL dynamics are known to converge to the Nash equilibrium.
Since the perturbation coefficient is set to be small enough empirically, this proposition suggests that our DFTRL algorithm provides a unified and in-depth understanding of the mechanism behind the two.
We would like to note that the continuous negative-momentum \cite{hemmat2023lead,zhang2021suboptimality,kovachki2021continuous,gidel2019negative,polyak1964methods} FTRL dynamics reduces to the ordinary FTRL dynamics, rather than the DFTRL, which we also verify in \appenref{append:cftrl}.

\begin{figure*}[t]
    \begin{minipage}[b]{0.325\linewidth}
        \centering
        \includegraphics[keepaspectratio, scale=0.28]{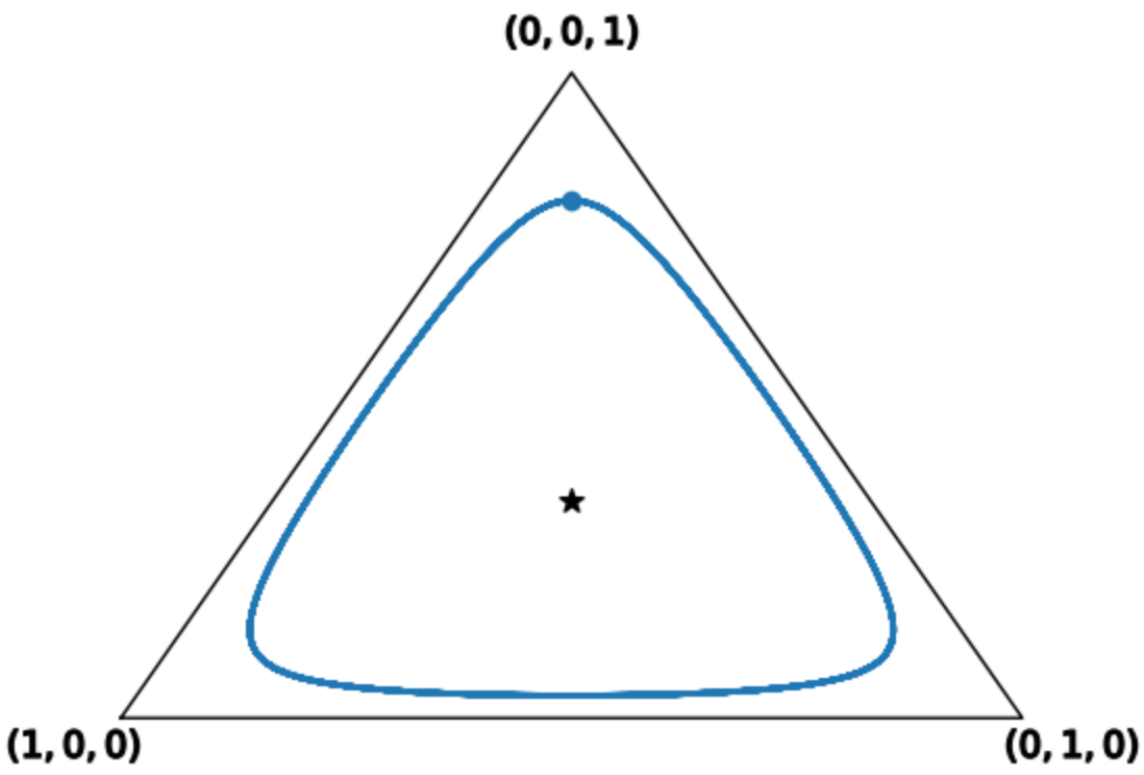}
        \subcaption{FTRL ($\coeff=0$)}
        \label{fig:rps0}
    \end{minipage}
    \begin{minipage}[b]{0.325\linewidth}
        \centering
        \includegraphics[keepaspectratio, scale=0.28]{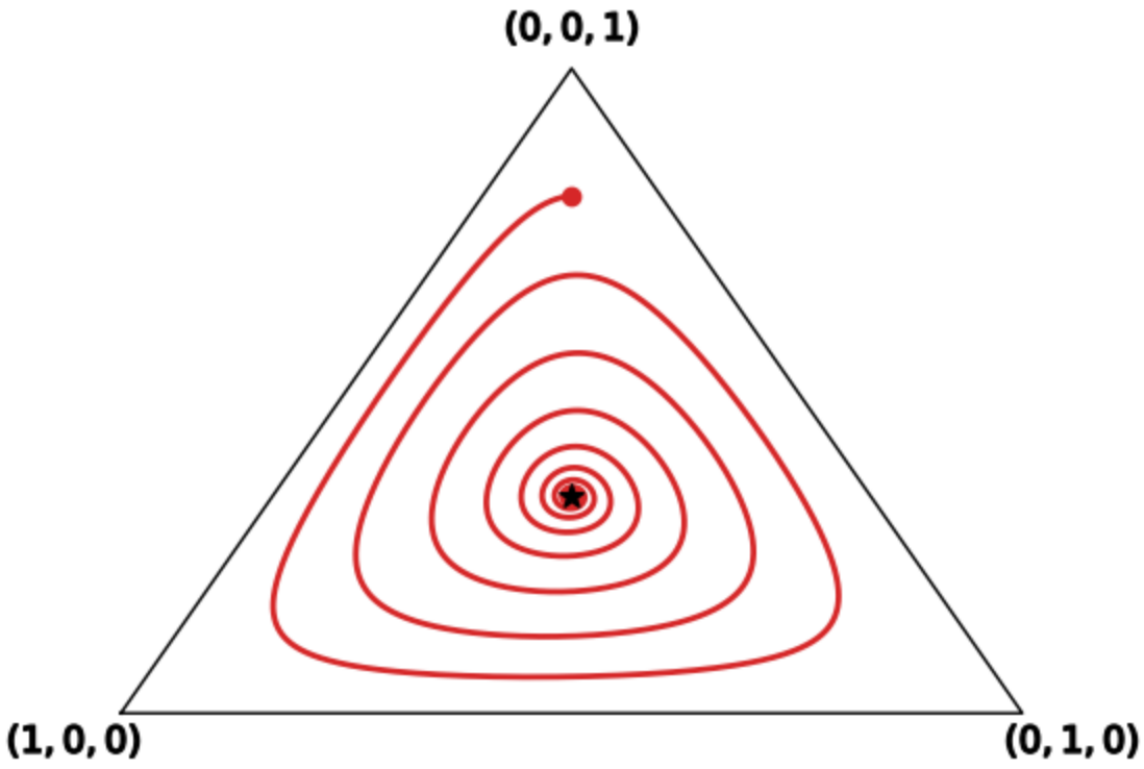}
        \subcaption{DFTRL ($\coeff=0.15$)}
        \label{fig:rps1}
    \end{minipage}
    \begin{minipage}[b]{0.33\linewidth}
        \centering
        \includegraphics[keepaspectratio, scale=0.3]{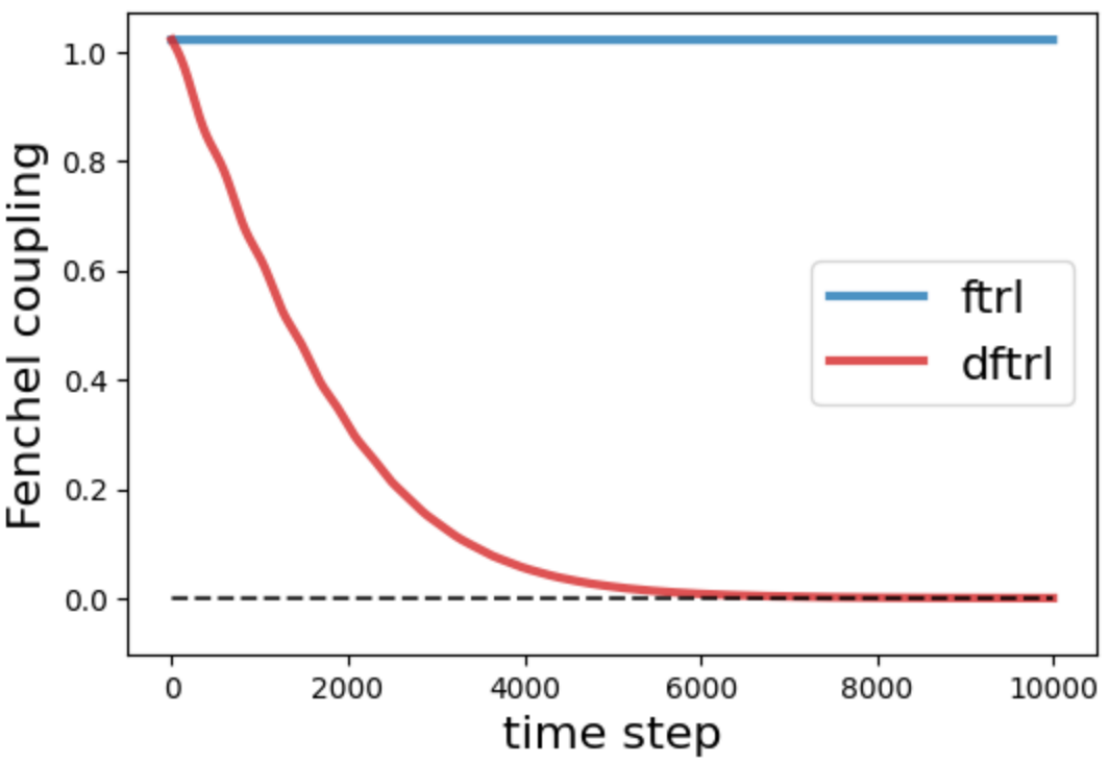}
        \subcaption{Fenchel coupling $\gfenchel$}
        \label{fig:fenchel-rps}
    \end{minipage}
    \caption{Two-player Rock-Paper-Scissors. The stars at the center of the simplices represent the Nash equilibrium.}
    \label{fig:rps}
\end{figure*}

\section{Experiments}
\label{sec:experiments}

We present two experimental simulations that demonstrate our theoretical results.%
\footnote{
    The code is available at \url{\githubrepo}.  
}
We illustrate typical examples in this section.
Further details and additional simulations are provided in \appenref{append:experiments}.

The first example is the two-player Rock-Paper-Scissors game.
We pick the entropic regularizer for $i=1$ and $2$:
\begin{equation}
    \hreg{i}(\xx{i}) = \sum_a \xx{i}_a \log\xx{i}_a,  \quad  \hdual{i}(\yy{i}) = \mathrm{lse}(\yy{i}).
\end{equation}
We solve the dynamical system \eref{eq:dftrl} under (\ref{eq:consistency}) and (\ref{eq:perturbation}), with the initial condition $\xx{1}(0)=\xx{2}(0)=(0.1, 0.1, 0.8)$.
Figure \ref{fig:rps} shows the result.
\fref{fig:rps0} and \fref{fig:rps1} represent the solution trajectories of $\xx{1}(t)$ for $\coeff=0$ and $0.15$, respectively.
The stars at the center of the simplices represent the Nash equilibrium: $\xs{i}=(1/3, 1/3, 1/3)$.
As mentioned in the previous section, $\coeff=0$ reduces the system to the ordinary FTRL dynamics, and we see the cyclic solution in \fref{fig:rps0}, as expected.
\fref{fig:fenchel-rps} shows that the Fenchel coupling $\gfenchel$ for DFTRL monotonically decreases, while it is conserved for FTRL, as proven in \propref{prop:fenchel} and \thmref{thm:dftrl}.

The second case is the three-player Matching Pennies game.
We solve the dynamical system as in the previous example with the entropic regularizer, and the result is shown in \fref{fig:mp}.
These figures depict the time evolution of the tuple $(\xx{1}_1(t), \xx{2}_1(t), \xx{3}_1(t))$ in the unit cube.
This game has the continuum of Nash equilibria,
\begin{equation}
    (\xs{1}_1, \xs{2}_1, \xs{3}_1)=(p, p, p),  \qquad  0\leq p \leq 1,
\end{equation}
which the diagonal straight lines represent in \fref{fig:mp}.
The solution trajectories in \fref{fig:mp0} and \fref{fig:mp1} correspond to the cases of $\coeff=0$ and $0.1$, respectively.
Again, we observe that the solutions for $\coeff=0$ exhibit cyclic behavior, and for $\coeff=0.1$ monotonically converge to the Nash equilibrium.

\section{Discussion}
\label{sec:discussion}

This paper discussed a novel formulation and algorithms for the learning dynamics in games.
We established the Hamiltonian of poly-matrix zero-sum games and showed that the FTRL dynamics is indeed the Hamiltonian system of the proposed $\hftrl$, \eref{eq:ftrl-hamiltonian}, in \thmref{thm:ftrl-hamiltonian}.
We examined the Hamiltonian system through the lens of symmetry and found that the conserved quantities in the FTRL dynamics are intrinsically encoded in $\hftrl$ (Propositions \ref{prop:simplex}, \ref{prop:fenchel}, and \ref{prop:angular}).
The Fenchel coupling $\gfenchel$, one such conserved quantity, measures the distance to the Nash equilibrium, the conservation of which leads to the non-convergent learning dynamics.
Considering the perturbation that monotonically decreases $\gfenchel$ but preserves the conservation law of $\hftrl$, we proposed the DFTRL algorithm in \thmref{thm:dftrl}, which yields convergent dynamics to the Nash equilibrium.
We also clarified the relation between DFTRL and the known last-iterate convergent algorithms in \propref{prop:co-ceg}.

A limitation of this study is that we considered only fully-mixed strategies.
This is simply to ensure compatibility with the conservation laws.
Including a proper treatment on the boundary of the strategy space $\sspace$ would provide a more complete description of the formulation.
Beyond this, our work opens many avenues for future research.
At first, different types of perturbations may exist that yield convergent dynamics, or there might be a uniqueness theorem for our perturbation.
In fact, under certain conditions, the known last-iterate convergent FTRL algorithms reduce to the DFTRL as demonstrated in \propref{prop:co-ceg}.
Such considerations would further deepen our understanding of convergence in learning in games.
Next, since our Hamiltonian $\hftrl$ reproduces well-known conserved quantities and even generates new ones (see also \appenref{append:extra-conservation-laws}), other properties characterizing poly-matrix zero-sum games may also be derived from $\hftrl$.
The symmetries (of agents or actions) have been of interest in game theory for a long time \cite{nash1951non,tewolde2025computing}.
In connection with these symmetries, our Hamiltonian could further derive the corresponding conservation laws and potentially provide insights into solving games.
Third, the Hamiltonian $\hftrl$ could serve as a promising starting point for the quantum version of the FTRL dynamics in poly-matrix zero-sum games.
As we have the canonically conjugate variables satisfying $\{\xx{i}_a,\yy{j}_b\}_\mathrm{P}=\delta^{ij}\delta_{ab}$, we can perform the canonical quantization by replacing the variables with the operators satisfying the canonical commutation relations, $[\hat{x}^i_a,\hat{y}^j_b]=\mathrm{i}\hbar\,\delta^{ij}\delta_{ab}$.
In analogy with the standard quantum mechanics context, we can naturally define the poly-matrix zero-sum \emph{quantum} game through $\hftrl(\hat{x},\hat{y})$, and compare this system with existing works \cite{eisert1999quantum,eisert2000quantum,jain2022matrix,lotidis2023learning,lin2025learning}.
Lastly, the conditions of poly-matrix and zero-sum could be further relaxed by broadening the scope of our research.
In the study of such situations, the key would be to establish an appropriate Hamiltonian and to explore the symmetry thereof.
We believe that this paper not only reorganizes the FTRL dynamics, but also provides a compass for establishing emerging research fields.

\onecolumngrid
\appendix


\section{Properties and examples of FTRL dynamics}
\label{append:properties}

We here summarize some known properties and examples of learning in games relevant to our discussions in the main text.
To begin with, we recap the Karush--Kuhn--Tucker (KKT) condition, also known as the generalized method of Lagrange multipliers.
The notations overlap somewhat with those in the main text, but we believe there is no confusion.

Let $M$ be an $n$-dimensional smooth manifold and
\begin{equation}
    f: M \to \R,  \quad  g: M \to \R^m,  \quad  h: M \to \R^l
\end{equation}
are assumed to be in $C^1$.
Let us consider a nonlinear optimization problem,
\begin{equation}
    \min_{p\in M} f(p),
\end{equation}
with the constraints
\begin{equation}
    g(p) \leq 0,  \quad  h(p) = 0.
\end{equation}

\begin{quote}
    \underline{The Karush--Kuhn--Tucker condition.} \cite{bertsekas1997nonlinear,bergmann2019intrinsic} \\
    Suppose that $p^\ast\in M$ is a local minimizer of the above constrained optimization problem, and that the GCQ holds at $p^\ast$.
    Then there exist (unique) Lagrange multipliers
    \begin{equation}
        \mu \in \R^m,  \quad  \lambda \in \R^l,
    \end{equation}
    such that at $p^\ast$ the following conditions hold:
    \begin{align}
        &df + \mu^\top dg + \lambda^\top dh = 0,  \\
        &h(p^\ast) = 0,  \\
        &\mu\geq 0,  \quad  g(p^\ast)\leq 0,  \quad  \mu^\top g(p^\ast) = 0.
    \end{align}
\end{quote}

\noindent \textbf{Properties of Nash equilibrium.}
As introduced in the main text, the Nash equilibrium is defined as follows.
\begin{equation}
    \strategies^\ast = \left( \xs{1}, \dots, \xs{n} \right) \in \sspace
\end{equation}
is a Nash equilibrium if it satisfies
\begin{equation}
    {}^{\forall}i\in\agents,
    \qquad
    \xs{i} \in \argmax_{\xx{i}\in \sspace^{N_i}} u_i(\xx{i}; (\xx{j(\neq i)})^\ast).
\end{equation}
This means that the Nash equilibrium is a solution to the optimization problems of
\begin{equation}
    f(\xx{i}) = -u_i(\xx{i}; (\xx{j(\neq i)})^\ast) = -(\xx{i})^\top \sum_{j\in\agents} \payoffmat{ij} \xs{j},
\end{equation}
over $M_i=\R^{N_i+1}$ with the constraints
\begin{align}
    g(\xx{i}) &= -(\xx{i}_1, \dots, \xx{i}_{N_i+1}) \leq 0,  \\
    h(\xx{i}) &= \sum_{a\in A_i} \xx{i}_a - 1 = 0.
\end{align}
Thus from the KKT condition, at the Nash equilibrium $\xs{i}$ there exist Lagrange multipliers $\mu^i\in\R^{N_i+1}$ and $\lambda^i\in\R$, such that
\begin{align}
    &\sum_a \left( -\bigg( \sum_j \payoffmat{ij} \xs{j} \bigg)_a - \mu^i_a + \lambda^i \right) d\xx{i}_a = 0,  \\
    &\sum_a \xs{i}_a = 1,  \\
    &\mu^i_a \geq 0,  \quad  \xs{i}_a \geq 0,  \quad  \sum_a \mu^i_a \xs{i}_a = 0  .
\end{align}
From the last condition, we find that for each $a\in A_i$, either $\mu^i_a=0$ or $\xs{i}_a=0$ holds.
If $\xs{i}_a>0$ for all $a\in A_i$, then $\mu^i=0$ and the KKT condition reduces to the ordinary method of Lagrange multipliers:
\begin{equation}
    \sum_j \payoffmat{ij} \xs{j} = \lambda^i \onevec,
\end{equation}
with the constraint $\sum_a \xs{i}_a=1$.
$\onevec\coloneq (1, \dots, 1)$ is the all-one vector.
Otherwise, if $\xs{i}_{a'}=0$ for some $a'$, then $\mu^i_{a'}\geq 0$ and the condition becomes
\begin{equation}
    \bigg( \sum_j \payoffmat{ij} \xs{j} \bigg)_{a'}
        = \lambda^i - \mu^i_{a'}  \leq  \lambda^i,
\end{equation}
with the constraint $\sum_a \xs{i}_a=1$.

For a zero-sum game, the total utility for all the agents always vanishes due to the antisymmetricity of the payoff matrices:
\begin{equation}
    \begin{aligned}[b]
        \sum_{i\in\agents} u_i(\xx{i}; \xx{j(\neq i)})
            &=  \sum_{i,j} (\xx{i})^\top \payoffmat{ij} \xx{j}  \\
            &= -\sum_{i,j} (\xx{j})^\top \payoffmat{ji} \xx{i}  \\
            &= 0.
    \end{aligned}
    \label{eq:zerosum-property}
\end{equation}
This fact implies that the constants $\lambda^i$ in the aforementioned property for the Nash equilibrium satisfy the corresponding condition,
\begin{equation}
    \begin{aligned}[b]
        0   &= \sum_{i\in\agents} u_i(\xs{i}; (\xx{j(\neq i)})^\ast)  \\
            &= \sum_{i\in\agents} (\xs{i})^\top \sum_j \payoffmat{ij} \xs{j}  \\
            &= \sum_{i\in\agents} (\xs{i})^\top \lambda^i \onevec  \\
            &= \sum_{i\in\agents} \lambda^i \sum_a \xs{i}_a  \\
            &= \sum_{i\in\agents} \lambda^i,
    \end{aligned}
\end{equation}
where we assumed $\strategies^\ast$ is fully mixed.
Overall, the fully-mixed Nash equilibrium for a zero-sum game has the property that for $i=1, \dots, n$, there exists a set of constants $\lambda^i\in\R$ that satisfies
\begin{align}
    \sum_{j\in\agents} \payoffmat{ij} \xs{j} = \lambda^i \onevec,
    \qquad
    \sum_{i\in\agents} \lambda^i = 0.
    \label{eq:nash-property}
\end{align}

Next, we provide two examples of the FTRL dynamics.
To describe them, we solve the right-hand side of \eref{eq:xftrl} or \eref{eq:hdual}, a nonlinear optimization problem.
By assuming that strategies are always fully mixed, as in the Nash equilibrium properties above, this optimization problem is solved by the reduced KKT condition, namely the ordinary method of Lagrange multipliers.
Suppose the $\zeta^i$ solve the optimization problems of
\begin{equation}
    f(\xx{i}) = - \left( \inner{\xx{i}}{\yy{i}} - \hreg{i}(\xx{i}) \right),
\end{equation}
over $M_i=\R^{N_i+1}$ with the constraint $h(\xx{i})=\sum_a \xx{i}_a - 1 = 0$, for $i=1,\dots,n$.
Then there exist Lagrange multipliers $\nu^i\in\R$ which satisfy
\begin{align}
    \sum_a \left( -\yy{i}_a + \nabla\hreg{i}(\zeta^i)_a + \nu^i \right) d\xx{i}_a
        &= 0,  \label{eq:kkt-ftrl}  \\
    \sum_a \zeta^i_a
        &= 1.
\end{align}
By definition, $\zeta^i=\argmax_{z\in\sspace^{N_i}} \left( \inner{z}{\yy{i}} - \hreg{i}(z) \right)$.
In many cases, the constants $\nu^i$ are determined by the second condition on $\zeta^i$.

\noindent \textbf{Entropic regularizer and the replicator dynamics.}
One of the most widely known examples of learning dynamics is so-called the \emph{replicator dynamics} \cite{taylor1978evolutionary,taylor1979evolutionarily,schuster1983replicator},
which is given by the following entropic regularizer function:
\begin{equation}
    \hreg{i}(\xx{i}) = \sum_a \xx{i}_a \log \xx{i}_a.
\end{equation}
For this regularizer function, \eref{eq:kkt-ftrl} becomes
\begin{equation}
    \zeta^i_a = \exp \left( \yy{i}_a - \nu^i - 1 \right),
\end{equation}
for all $a\in A_i$.
By summing over $a$, the constraint for $\zeta^i$ gives
\begin{equation}
    e^{\nu^i + 1} = \sum_a \exp(\yy{i}_a).
\end{equation}
Thus, we find
\begin{equation}
    \zeta^i_a = \frac{\exp(\yy{i}_a)}{\sum_{a'} \exp \left( \yy{i}_{a'} \right)}.
\end{equation}
By substituting this expression to the right-hand side of \eref{eq:hdual}, we obtain
\begin{equation}
    \hdual{i}(\yy{i}) = \mathrm{lse}(\yy{i}) \coloneq \log \sum_a \exp(\yy{i}_a),
\end{equation}
and
\begin{align}
    \nabla \hdual{i}(\yy{i})_a
        &= \frac{\exp(\yy{i}_a)}{\sum_{a'} \exp \left( \yy{i}_{a'} \right)} \eqcolon \softmax(\yy{i})_a,  \\
    \hess(\hdual{i}(\yy{i}))_{ab}
        &= \delta_{ab} \, \softmax(\yy{i})_a - \softmax(\yy{i})_a \, \softmax(\yy{i})_b.
\end{align}
With these at hand, we find the learning rule,
\begin{align}
    \dd{\xx{i}(t)}{t}
        &= \hess\big(\hdual{i}(\yy{i}(t))\big) \sum_{j\in\agents} \payoffmat{ij} \xx{j}(t)  \nonumber  \\
        &= \xx{i}(t) \odot \left( \sum_{j} \payoffmat{ij} \xx{j}(t) - \inner{\xx{i}(t)}{\sum_{j} \payoffmat{ij} \xx{j}(t)} \onevec \right),
\end{align}
where $\odot$ is the element-wise product, also known as the Hadamard product, and we used the relation $\xx{i}(t)=\nabla\hdual{i}(\yy{i}(t))=\softmax(\yy{i}(t))$.

\noindent \textbf{Euclidean regularizer and the projection dynamics.}
Another widely used example of learning dynamics is given by the Euclidean quadratic regularizer:
\begin{equation}
    \hreg{i}(\xx{i}) = \frac12 \norm{\xx{i}}_2^2.
\end{equation}
Similarly to the entropic regularizer case, the optimization problem is solved by \eref{eq:kkt-ftrl} and the constraint as
\begin{equation}
    \zeta^i_a = \yy{i}_a - \frac{1}{\abs{A_i}} \left( \sum_{a'} \yy{i}_{a'} - 1 \right).
\end{equation}
We now have
\begin{equation}
    \hdual{i}(\yy{i}) = \frac12 \norm{\yy{i}}_2^2 - \frac12 \norm{c^i}_2^2,
\end{equation}
where $c^i\in\R^{\abs{A_i}}$ are identical-entry vectors:
\begin{equation}
    c^i = \frac{1}{\abs{A_i}} \left( \sum_{a'} \yy{i}_{a'} - 1 \right) \onevec.
\end{equation}
The gradients and the Hessians are
\begin{align}
    \nabla \hdual{i}(\yy{i})_a
        &= \yy{i}_a - \frac{1}{\abs{A_i}} \left( \sum_{a'} \yy{i}_{a'} - 1 \right),  \\
    \hess(\hdual{i}(\yy{i}))_{ab}
        &= \delta_{ab} - \frac{1}{\abs{A_i}},
\end{align}
which lead to the projected reinforcement learning process.
The agents' mixed strategies are then known to follow the \emph{projection dynamics} \cite{friedman1991evolutionary,nagurney1997projected,lahkar2008projection,sandholm2008projection,mertikopoulos2016learning}:
\begin{align}
    \dd{\xx{i}(t)}{t}
        &= \hess\big(\hdual{i}(\yy{i}(t))\big) \sum_{j\in\agents} \payoffmat{ij} \xx{j}(t)  \nonumber  \\
        &= \sum_{j} \payoffmat{ij} \xx{j}(t) - \frac{1}{\abs{A_i}} \sum_a \bigg( \sum_{j} \payoffmat{ij} \xx{j}(t) \bigg)_a \onevec.
\end{align}

\section{Lightning review of Hamiltonian dynamics}
\label{append:hamiltonian}

For accesibility to broader readers, we describe minimal facts on Hamiltonian dynamics relevant to our main discussions.
Here again, some notations overlap with those in the main text, but there should be no confusion.
For more details on Hamiltonian dynamics, we refer the reader to \cite{libermann,arnold}.

Let $(M, \simplectic)$ be a real $2n$-dimensional symplectic manifold equipped with the symplectic form $\simplectic$.
A typical example is the space of ``position'' and ``momentum'', $M=T^\ast\R^n \simeq \R^n \times \R^n \ni (q, p)$, on which the symplectic form takes the form
\begin{equation}
    \simplectic = \sum_{a=1}^n dp_a \wedge dq_a
        \simeq
        \begin{pmatrix}
            0    &  1_n  \\
            -1_n &  0
        \end{pmatrix}
        \eqcolon J,
\end{equation}
where $1_n$ is the $n\times n$ identity matrix.
Even if the global structure of $M$ is subtle, the symplectic form can be expressed as in the above at least locally, owing to Darboux's theorem.
Due to this fact, in this appendix we mostly consider this typical example since it suffices for our purpose.

Given a real smooth function $f$ on $M$, $f\in C^{\infty}(M)$, the \defname{Hamiltonian system} induced by $f$ is defined by the following flow equation for the curve $\phi: I \rightarrow M$, accompanied by the symplectic form,
\begin{equation}
    \dd{\phi(t)}{t} = J \nabla f(z)\mid_{z=\phi(t)},  \quad  z\in M.
    \label{eq:flow1}
\end{equation}
Specifically, using the local coordinate $z=(q, p)$ and the abuse of notation $\phi(t)=(q(t), p(t))$, the flow equation becomes
\begin{equation}
    \begin{aligned}
        \dd{q(t)}{t} &=  \nabla_p f(q,p)\mid_{q=q(t),\, p=p(t)},  \\
        \dd{p(t)}{t} &= -\nabla_q f(q,p)\mid_{q=q(t),\, p=p(t)}.
    \end{aligned}
    \label{eq:flow2}
\end{equation}
Such a flow-generating function $f$ of a Hamiltonian system is itself called the \defname{Hamiltonian} and is often written as $H$.
In Hamiltonian dynamics,%
\footnote{
    In a narrow sense, the term ``Hamiltonian dynamics'' may be referred to as the solution to the flow equations (\ref{eq:flow1}) or (\ref{eq:flow2}).
    However, we use this term collectively including the phase space, the Hamiltonian system, and its solutions, etc.
}
the symplectic manifold $M$ is referred to as the \defname{phase space} or \defname{state space}.

\noindent \textbf{Energy conservation law.}
In Hamiltonian dynamics, the value of the flow-generating function $f(z)$ \underline{on the trajectory of the solution to the flow equation (\ref{eq:flow1})} is called \defname{energy}:
\begin{equation}
    E : I \rightarrow \R,
    \qquad
    E(t) \coloneq f(\phi(t)).
\end{equation}
A well-known fact is that this energy is conserved along the time interval $I$, as long as the function $f$ has no explicit dependence on $I$ (time).
In other words, the value of the flow-generating function $f(z)$ is invariant along the trajectory of the solution to the flow equation:
Using again the local coordinate $z=(q, p)$,
\begin{equation}
    \begin{aligned}[b]
        \dd{E(t)}{t}
            &= \frac{d}{dt} f(q(t), p(t))  \\
            &= \inner{\nabla_q f(q,p)\mid_{q=q(t),\, p=p(t)}}{\dd{q(t)}{t}} + \inner{\nabla_p f(q,p)\mid_{q=q(t),\, p=p(t)}}{\dd{p(t)}{t}}  \\
            &= \inner{-\dd{p(t)}{t}}{\dd{q(t)}{t}} + \inner{\dd{q(t)}{t}}{\dd{p(t)}{t}}  \\
            &= 0,
    \end{aligned}
\end{equation}
where we have used the fact that $\phi(t)=(q(t), p(t))$ satisfies the flow equation (\ref{eq:flow2}) and the symmetricity of the inner product.

\noindent \textbf{Noether's theorem in Hamiltonian dynamics.}
Conservation law in Hamiltonian dynamics is generalized by the so-called Noether's theorem.
In a Hamiltonian system, besides the energy, we can identify conserved quantities according to the symmetry of the Hamiltonian.

\begin{thm*}[Noether, in the Hamiltonian formalism]
    If a Hamiltonian function $f$ on a symplectic manifold $(M, \simplectic)$ admits the one-parameter group of canonical transformations generated by a function $G\in C^{\infty}(M)$,
    then $G$ is conserved on the Hamiltonian flow of $f$.
\end{thm*}

For the details, see \cite{arnold}.
In a word, this theorem states that if a Hamiltonian is invariant under a certain (canonical) transformation, the Hamiltonian system has an associated conserved quantity.
In the physics context, we refer to such a transformation (or equivalently a group action) as the \emph{symmetry} of the Hamiltonian system.
We will illustrate two examples below for intuition.
Instead of providing the proof or more details of this theorem, here we only describe a rough procedure for the construction of the conserved quantity associated with the corresponding symmetry.
Let us consider a Hamiltonian system generated by $f\in C^{\infty}(M)$ and the following infinitesimal canonical transformation on the local coordinate $(q,p)$ of $M$:
\begin{equation}
    q_a' = q_a + \delta q_a,
    \quad
    p_a' = p_a + \delta p_a,
\end{equation}
where
\begin{equation}
    \delta q_a = \epsilon \pbracket{q_a}{G},
    \quad
    \delta p_a = \epsilon \pbracket{p_a}{G},
\end{equation}
for $G\in C^{\infty}(M)$ and $\epsilon\in\R_{\geq 0}$. $\pbracket{\cdot}{\cdot}$ denotes the Poisson bracket.
What Noether's theorem tells us is that if the Hamiltonian is invariant under this canonical transformation for a certain $G$,
\begin{equation}
    f(q', p') = f(q, p),
\end{equation}
which leads to $\pbracket{G}{f}=0$, then the function $G$ is a conserved quantity in the Hamiltonian system:
Along the trajectory of the solution to the flow equation (\ref{eq:flow2}), 
\begin{equation}
    \begin{aligned}[b]
        \frac{d}{dt} G(q(t), p(t))
            &= \inner{\nabla_q G(q,p)\mid_{q=q(t),\, p=p(t)}}{\dd{q(t)}{t}} + \inner{\nabla_p G(q,p)\mid_{q=q(t),\, p=p(t)}}{\dd{p(t)}{t}}  \\
            &= \left( \inner{\nabla_q G(q,p)}{\nabla_p f(q,p)} - \inner{\nabla_p G(q,p)}{\nabla_q f(q,p)} \right)\mid_{q=q(t),\, p=p(t)}  \\
            &= \pbracket{G}{f}  \\
            &= 0.
    \end{aligned}
\end{equation}

To provide intuition, we present two well-known examples from the physics literature.
Let $M=T^\ast\R^3\simeq\R^3\times\R^3$ be the phase space of a single particle of mass $m$ moving in the three-dimensional space.
The Hamiltonian $H:M\to\R$, describing generic classical mechanical systems, takes the form
\begin{equation}
    H(q, p) = \frac{1}{2m}\norm{p}^2_2 + V(q),
\end{equation}
where $V:\R^3\to\R$ is referred to as the potential function, which characterizes the system.
\begin{itemize}
    \item Momentum conservation law: If the Hamiltonian $H$ has a translation symmetry,
        \begin{align}
            H(q', p') = H(q, p),  \\
            q' = q + e,  \quad  p' = p,
        \end{align}
        where $e$ is a constant vector, then the associated conserved quantity is
        \begin{equation}
            G^P(q, p; e) = \inner{e}{p}.
            \label{eq:momentumconservation}
        \end{equation}
        In the physical sense, this implies that if the potential function is invariant under a certain constant shift of the position coordinate, the corresponding momentum of the particle is conserved in the Hamiltonian system.
    \item Angular momentum conservation law: If the Hamiltonian $H$ has a rotation symmetry,
        \begin{align}
            H(q', p') = H(q, p),  \\
            q' = R q,  \quad  p' = R p,
        \end{align}
        where $R\in\mathrm{SO}(3)$, then the associated conserved quantity is the angular momentum along $R$.
        If the invariance holds for any rotation, ${}^{\forall} R\in\mathrm{SO}(3)$, the whole angular momentum vector of the particle is conserved in the Hamiltonian system,
        \begin{equation}
            G^L_a(q, p) = (q \times p)_a,  \quad a = 1, \dots, 3,
        \end{equation}
        where $\times$ is the three-dimensional exterior product.
\end{itemize}
As these examples demonstrate, the examination of symmetry of the Hamiltonian has a direct connection to the conserved quantity of the system.
Conserved quantities provide us an enormous amount of information of the system under consideration, as a sufficient number of conserved quantities even leads to the solvability of the system.
Due to this, in the study of a dynamical system, it is essential in understanding the system to establish an appropriate Hamiltonian function and to elucidate the symmetry thereof.

\section{Comments on the paper by Bailey and Piliouras}
\label{append:bp}

In this appendix, we recap the main arguments in the paper by Bailey and Piliouras \cite{bailey2019multi}, emphasizing that their use of the term ``Hamiltonian'' deviates from the standard convention, explained in \appenref{append:hamiltonian}.
It turns out that their central claim essentially amounts to the tautological statement that “the solution to the FTRL dynamics satisfies the FTRL dynamics,” and does not offer novel insights.
As a result, their analysis ultimately reduces to a study of the FTRL dynamics and previously known conserved quantities.

With the notations in \sref{sec:preliminaries}, suppose that a game $\sgame$ is given with a set of regularizer functions $\hreg{i}\in\pspace{\sspace^{N_i}}$.
One defines the ``cumulative strategies'' by
\begin{equation}
    \XX{i}(t) \coloneq \int_0^t \xx{i}(s) ds.
\end{equation}
Then, for the choice of $\hreg{i}$, the FTRL dynamics describes the time evolution of the learning dynamics $(X(t), y(t))$ in $\sgame$ by
\begin{equation}
    \begin{aligned}
        \dd{\XX{i}(t)}{t} &= \nabla\hdual{i}(\yy{i})\mid_{\yy{i}=\yy{i}(t)},  \\
        \dd{\yy{i}(t)}{t} &= \sum_{j(\neq i)} \payoffmat{ij} \xx{j}(t),
    \end{aligned}
    \label{eq:ftrl-bp}
\end{equation}
with the relations $\xx{i}(t)=\dot{X}^i(t)=\nabla\hdual{i}(\yy{i}(t))$.
One introduces the ``free Hamiltonian'' for the agents by
\begin{equation}
    H_0(X, y) = \sum_{i\in\agents} \hdual{i}(\yy{i}),
\end{equation}
and the ``interaction Hamiltonian'',
\begin{equation}
    H^{\sigma}_{\text{int}}(X, y) = -\sigma \sum_{i\in\agents} \hdual{i} \bigg(\beta^i + \sum_{k(\neq i)} \payoffmat{ik} \XX{k} \bigg),
\end{equation}
where $\XX{i},\, \yy{i}\in\R^{\abs{A_i}}$, and $\beta^i\in\R^{\abs{A_i}}$ are some constants.
The total ``Hamiltonian'' of the system is a function on $\R^{\sum_i\abs{A_i}}\times\R^{\sum_i\abs{A_i}}$ defined by their sum,
\begin{equation}
    \begin{aligned}
        H^\sigma(X, y)
            &= H_0(X, y) + H^{\sigma}_{\text{int}}(X, y)  \\
            &= \sum_{i\in\agents} \left[ \hdual{i}(\yy{i}) -\sigma \hdual{i} \bigg(\beta^i + \sum_k \payoffmat{ik} \XX{k} \bigg) \right]  \\
            &= \sum_{i\in\agents} \hdual{i}(\yy{i}) -\sigma \sum_{j\in\agents} \hdual{j} \bigg(\beta^j + \sum_{i(\neq j)} \payoffmat{ji} \XX{i} \bigg).
    \end{aligned}
    \label{eq:hamiltonian-bp}
\end{equation}

\begin{thm*}[Bailey and Piliouras \cite{bailey2019multi}, Theorem 4.3]
    The function $H^\sigma$, \eref{eq:hamiltonian-bp}, with $\sigma=-1$ and $\beta^i=\yy{i}(0)$ for all $i\in\agents$, is invariant along the solution trajectory of the FTRL dynamics \eref{eq:ftrl-bp}.
\end{thm*}

Using \eref{eq:ftrl-bp} and the assumption, the invariance of $H^{\sigma=-1}$ along the solution trajectory is readily proven as
\begin{equation}
    \begin{aligned}[b]
        &\frac{d}{dt}H^{\sigma=-1}(X(t), y(t))  \\
        &= \sum_{i\in\agents} \left[ \inner{\nabla\hdual{i}(\yy{i})\mid_{\yy{i}=\yy{i}(t)}}{\dd{\yy{i}(t)}{t}}
            + \inner{\nabla\hdual{i}(z^i)\mid_{z^i=\yy{i}(0) + \sum_k \payoffmat{ik} \XX{k}(t)}}{\sum_j \payoffmat{ij} \dd{\XX{j}(t)}{t}} \right]  \\
        &= 2\sum_{i\in\agents} \inner{\xx{i}(t)}{\sum_j \payoffmat{ij} \xx{j}(t)}  \\
        &= 0,
    \end{aligned}
\end{equation}
where at the last line we have used the zero-sum property of the game, \eref{eq:zerosum-property}.
From this theorem, we recognize that the function $H^{\sigma=-1}$ is a conserved quantity for poly-matrix zero-sum games along the solution trajectories of the FTRL dynamics.
However, the function $H^{\sigma=-1}$ is naively not regarded as a Hamiltonian function for the FTRL dynamics.
Let $M'=T^\ast\R^{\sum_i\abs{A_i}}\simeq\R^{\sum_i\abs{A_i}}\times\R^{\sum_i\abs{A_i}}$ be a phase space and $(X, y)$ be the (global) canonical coordinate of $M'$.
Suppose we have a function $H^{\sigma=-1}:M'\to\R$ given as in \eref{eq:hamiltonian-bp}.
Then, the Hamiltonian system induced by $H^{\sigma=-1}$ is described by the flow equations (\ref{eq:flow2}) as
\begin{align}
    \dd{\XX{k}(t)}{t}
        &= \nabla_{\yy{k}} H^{\sigma=-1}(X, y)\mid_{X=X(t),\, y=y(t)}  \nonumber  \\
        &= \nabla \hdual{k}(\yy{k})\mid_{\yy{k}=\yy{k}(t)},  \\
    \dd{\yy{k}(t)}{t}
        &= -\nabla_{\XX{k}} H^{\sigma=-1}(X, y)\mid_{X=X(t),\, y=y(t)}  \nonumber  \\
        &= \sum_j \payoffmat{kj} \nabla \hdual{j}(\yy{j})\mid_{\yy{j}=\beta^j + \sum_l \payoffmat{jl} \XX{l}(t)}.  \label{eq:yftrl-not}
\end{align}
As can be seen from the right-hand side of \eref{eq:yftrl-not}, this Hamiltonian system is generally not equivalent to the FTRL dynamics (\ref{eq:ftrl-bp}).
If one imposes a constraint by hand on the variables $\XX{i}$ and $\yy{i}$ such as the projection onto the submanifold
\begin{equation}
    \tilde{M} = \{ \yy{i}=\beta^i+\sum_j \payoffmat{ij} \XX{j} \,\mid\, i\in\agents \}\subset M',
    \label{eq:submfd}
\end{equation}
then the right-hand side becomes
\begin{equation}
    \sum_j \payoffmat{kj} \nabla \hdual{j}(\yy{j})\mid_{\yy{j}=\yy{j}(t)} = \sum_j \payoffmat{kj} \dd{\XX{j}(t)}{t},
\end{equation}
and this Hamiltonian system could be thought of as an equivalent of \eref{eq:ftrl-bp}.
Nonetheless, imposing constraints on the variables, such as \eref{eq:submfd}, implies that ``the FTRL dynamics is the Hamiltonian system on the manifold on which the solution is the FTRL dynamics,'' which is a tautological statement and does not yield any new insight.
Overall, the function $H^{\sigma=-1}$ is merely a conserved quantity, and not the Hamiltonian function for the FTRL dynamics on the total phase space $(M'; X, y)$.

As Bailey and Piliouras mentioned in \cite[Sec.~5]{bailey2019multi}, the function $H^{\sigma=-1}$ restricted to the submanifold (\ref{eq:submfd}),
\begin{equation}
    H^{\sigma=-1}\mid_{\tilde{M}}(X, y) = 2 \sum_{i\in\agents} \hdual{i}(\yy{i}),
\end{equation}
is essentially a certain shift of the Fenchel coupling.
As a byproduct of the FTRL Hamiltonian (\ref{eq:ftrl-hamiltonian}), this result emerges in our Hamiltonian system as simply another conserved quantity through the corresponding symmetry.
From the proof of \propref{prop:fenchel} in \appenref{append:proof-fenchel}, it follows that the FTRL Hamiltonian $\hftrl$ is also invariant under the following canonical transformation for all $i\in\agents$, which is a variant of \eref{eq:symmetry-fenchel},
\begin{equation}
    x'^i = x^i + \nabla\hdual{i}(\yy{i}),  \quad  y'^i = y^i.
\end{equation}
A generator of this canonical transformation, $g\in C^\infty(M_{\zgame})$, satisfies the following for all $i\in\agents$ and for all $a\in A_i$,
\begin{align}
    \pbracket{\xx{i}_a}{g} &=  \dpdp{g(x,y)}{\yy{i}_a} = \dpdp{\hdual{i}(\yy{i})}{\yy{i}_a},  \\
    \pbracket{\yy{i}_a}{g} &= -\dpdp{g(x,y)}{\xx{i}_a} = 0.
\end{align}
Integration of these equations gives us the conserved quantity in the Hamiltonian system of $\hftrl$,
\begin{equation}
    g(\strategies, y) = \sum_{i\in\agents} \hdual{i}(\yy{i}) + \text{const.},
\end{equation}
which is a part of the Fenchel coupling.
This result is consistent with that of \cite[Sec.~5]{bailey2019multi} and provides a natural explanation from the perspective of the symmetry of our Hamiltonian.

\section{Proofs}  
\label{append:proofs}

We show the complete proofs for \lemref{lem:simplex}, \propref{prop:fenchel}, \propref{prop:angular}, and \thmref{thm:dftrl} in order.
In the proofs of \lemref{lem:simplex}, \propref{prop:fenchel} and \propref{prop:angular}, we use Noether's theorem.
For the details on Hamiltonian dynamics, see \appenref{append:hamiltonian}.

\subsection{Proof of \lemref{lem:simplex}}
\label{append:proof-simplex}

\begin{proof}
    Let us suppose that a Hamiltonian function $f$ on $M_{\zgame}$ has translation symmetry in $\yy{i}\in\R^{\abs{A_i}}$, that is, $f$ is invariant under the following canonical transformation:
    \begin{align}
        &f(\strategies', y') = f(\strategies, y),  \\
        &\strategies' = \strategies, \quad y'^i = y^i + \epsilon^i,
    \end{align}
    where the other $\yy{j(\neq i)}$ are kept intact and $\epsilon^i$ is an identical-entry constant vector, $\epsilon^i=\epsilon\onevec$.
    $\onevec\coloneq (1, \dots, 1)$ is the all-one vector.
    Then for all $a\in A_i$, a generator $g\in C^\infty(M_{\zgame})$ of this canonical transformation satisfies
    \begin{align}
        \pbracket{\xx{i}_a}{g} &=  \dpdp{g(x,y)}{\yy{i}_a} = 0,  \\
        \pbracket{\yy{i}_a}{g} &= -\dpdp{g(x,y)}{\xx{i}_a} = 1.
    \end{align}
    We thus obtain $g(x,y)=-\sum_a \xx{i}_a + \text{const.}$
    By Noether's theorem, the function $g$ is conserved in the Hamiltonian system generated by the Hamiltonian $f$, implying that $\sum_a \xx{i}_a$ is conserved.
\end{proof}

\lemref{lem:simplex} holds for any Hamiltonian functions that possess translation symmetry in the $\yy{i}$ direction, not only for ours.
Roughly speaking, this lemma corresponds to the conservation law of the center-of-mass momentum, which is obtained from \eref{eq:momentumconservation} with $e=\onevec$.
It is noted, however, that the role of canonical variables is interchanged between \lemref{lem:simplex} and \eref{eq:momentumconservation}.

\subsection{Proof of \propref{prop:fenchel}}
\label{append:proof-fenchel}

\begin{proof}
    The proof is straightforward by applying Noether's theorem to the Hamiltonian (\ref{eq:ftrl-hamiltonian}).
    We aim to show that our Hamiltonian (\ref{eq:ftrl-hamiltonian}) is invariant under the canonical transformation below, and that its generator is given by $\gfenchel$.
    Let us consider the canonical transformation,
    \begin{equation}
        x'^i = x^i + \nabla\hdual{i}(\yy{i}) - \xs{i},  \quad  y'^i = y^i,
        \label{eq:symmetry-fenchel-proof}
    \end{equation}
    for all $i\in\agents$. Under this transformation, the FTRL Hamiltonian is invariant as:
    \begin{equation}
        \begin{aligned}[b]
            \hftrl(\strategies', y') - \hftrl(\strategies, y)
                &= \sum_{i,j} \inner{\nabla\hdual{j}(\yy{j})}{\payoffmat{ji} \nabla\hdual{i}(\yy{i}) - \payoffmat{ji} \xs{i}}  \\
                &= -\sum_j \inner{\nabla\hdual{j}(\yy{j})}{\sum_i \payoffmat{ji} \xs{i}}  \\
                &= -\sum_j \lambda^j \sum_a \nabla\hdual{j}(\yy{j})_a  \\
                &= - C \sum_j \lambda^j = 0.
        \end{aligned}
    \end{equation}
    From the first line to the second, we used the zero-sum property of the game.
    From the second line to the third, we apply the property of the Nash equilibrium \eref{eq:nash-property}.
    Finally, we used the assumption that the elements of $\nabla\hdual{j}(\yy{j})$ sum to a constant, and again \eref{eq:nash-property}.
    A generator of this canonical transformation, $g\in C^\infty(M_{\zgame})$, satisfies the following for all $i\in\agents$ and for all $a\in A_i$,
    \begin{align}
        \pbracket{\xx{i}_a}{g} &=  \dpdp{g(x,y)}{\yy{i}_a} = \dpdp{\hdual{i}(\yy{i})}{\yy{i}_a} - \xs{i}_a,  \\
        \pbracket{\yy{i}_a}{g} &= -\dpdp{g(x,y)}{\xx{i}_a} = 0.
    \end{align}
    Thus, integrating these equations, we obtain the generator of this canonical transformation,
    \begin{equation}
        g(\strategies, y) = \sum_{i\in\agents} \left( \hdual{i}(\yy{i}) - \inner{\yy{i}}{\xs{i}} \right) + \text{const.}
    \end{equation}
    Since $\sum_i \hreg{i}(\xs{i})$ is a constant, it follows that the Fenchel coupling $\gfenchel$ is the conserved quantity generating the symmetry \eref{eq:symmetry-fenchel-proof} of the FTRL Hamiltonian.
\end{proof}

\subsection{Proof of \propref{prop:angular}}
\label{append:proof-angular}

\begin{proof}
    \begin{enumerate}
        \item We show that the FTRL Hamiltonian (\ref{eq:ftrl-hamiltonian}) is invariant under the canonical transformation below, and that its generator is given by $\gangular$.
            Let us consider the canonical transformation,
            \begin{align}
                x'^i = \rot_i x^i,  \quad  y'^i = \rot_i y^i,
            \end{align}
            for all $i\in\agents$.
            Under this transformation, the FTRL Hamiltonian is invariant:
            \begin{equation}
                \begin{aligned}[b]
                    \hftrl(\strategies', y')
                        &= \sum_{i,j} \inner{\nabla\hdual{j}(y'^j)}{\payoffmat{ji} x'^i}  \\
                        &= \sum_{i,j} \inner{\rot_j \nabla\hdual{j}(y^j)}{\payoffmat{ji} \rot_i x^i}  \\
                        &= \sum_{i,j} \inner{\nabla\hdual{j}(y^j)}{\rot_j^\top \payoffmat{ji} \rot_i x^i}  \\
                        &= \hftrl(\strategies, y).
                \end{aligned}
            \end{equation}
            From the first line to the second, we used the equivariance assumption for $\nabla\hdual{j}(y^j)$.
            From the third line to the fourth, we used the facts that $\rot_j^\top = \rot_j^{-1}$ and rotation matrices and the payoff matrices commute.
            Since $\rot_i\in\SO(\abs{A_i})$ are rotation matrices, there exist skew-symmetric matrices $\rotd_i\in\sof(\abs{A_i})$ such that $\rot_i = e^{\epsilon \rotd_i}$ and $\rotd_i^\top = -\rotd_i$.
            So a generator of this canonical transformation, $g\in C^\infty(M_{\zgame})$, satisfies
            \begin{align}
                \pbracket{\xx{i}_a}{g} &=  \dpdp{g(x,y)}{\yy{i}_a} = \sum_b (\rotd_i)_{ab} \xx{i}_b,  \\
                \pbracket{\yy{i}_a}{g} &= -\dpdp{g(x,y)}{\xx{i}_a} = \sum_b (\rotd_i)_{ab} \yy{i}_b.
            \end{align}
            Integrating these equations, we obtain the generator of this canonical transformation,
            \begin{equation}
                g(\strategies, y) = -\sum_{i\in\agents} \inner{\xx{i}}{\rotd_i \yy{i}} + \text{const.}
            \end{equation}
            From Noether's theorem it follows that the function $\gangular$ is the conserved quantity in the Hamiltonian system of $\hftrl$.
        \item Similarly to the first statement, we can show the FTRL Hamiltonian (\ref{eq:ftrl-hamiltonian}) is invariant under the canonical transformation,
            \begin{align}
                \strategies' = \rott \strategies,  \quad  y' = \rott y,
            \end{align}
            where $\mathcal{R}\in\SO(n)$ acts on the index of $\agents$ and rotates the whole spaces $\R_{\geq 0}^{\sum_i \abs{A_i}}$ and $\R^{\sum_i \abs{A_i}}$ for $\strategies$ and $y$, respectively.
            Using the notation $\llangle \alpha, \beta \rrangle \coloneq \sum_{i\in\agents} \inner{\alpha^i}{\beta^i}$, we have
            \begin{equation}
                \begin{aligned}[b]
                    \hftrl(\strategies', y')
                        &= \llangle \nabla\hdual{}(y'), U \strategies' \rrangle  \\
                        &= \llangle \rott \nabla\hdual{}(y), U \rott \strategies \rrangle  \\
                        &= \llangle \nabla\hdual{}(y), \rott^{-1} U \rott \strategies \rrangle  \\
                        &= \hftrl(\strategies, y).
                \end{aligned}
            \end{equation}
            Thus, with $\rott = e^{\epsilon \rottd},~ \rottd^\top = -\rottd$, we find from Noether's theorem that the function $\gangularr$ is the conserved quantity in the Hamiltonian system of $\hftrl$:
            \begin{equation}
                \gangularr(\strategies, y)
                    = \llangle \xx{}, \rottd \yy{} \rrangle
                    = \sum_{i,j\in\agents} \rottd_{ij} \inner{\xx{i}}{\yy{j}}.
            \end{equation}
    \end{enumerate}
\end{proof}

As mentioned in the main text, these two conservation laws correspond to the conservation of the total ``angular momentum'' of each agent and the total angular momentum of the system, respectively.
To see this more explicitly, one can consider the cases of $\abs{A_i}=3$ and/or $n=3$;
\begin{equation}
    \rotd_i,\, \rottd = L_a,
\end{equation}
where $L_a$ is either of the standard basis vectors of $\sof(3)$:
\begin{equation}
    L_1 = \begin{pmatrix}
        0 &  0 & 0 \\
        0 &  0 & 1 \\
        0 & -1 & 0
    \end{pmatrix},
    \quad
    L_2 = \begin{pmatrix}
        0 & 0 & -1 \\
        0 & 0 &  0 \\
        1 & 0 &  0
    \end{pmatrix},
    \quad
    L_3 = \begin{pmatrix}
        0  & 1 & 0 \\
        -1 & 0 & 0 \\
        0  & 0 & 0
    \end{pmatrix}.
\end{equation}
Then, these ``angular momentum'' conservation laws show that the following quantities are conserved:
\begin{align}
    \gangular(\strategies, y)
        &= \sum_{i\in\agents} \inner{\xx{i}}{L_a \yy{i}}
        = \sum_{i\in\agents} \left(\xx{i} \times \yy{i}\right)_a,  \\
    \gangularr(\strategies, y)
        &= \llangle \xx{}, L_a \yy{} \rrangle
        \eqcolon \llangle \xx{} \times \yy{} \rrangle_a.
\end{align}
where $\times$ is the three-dimensional exterior product.

One can indeed find an angular momentum-like conserved quantity in the FTRL dynamics of the two-player Rock-Paper-Scissors game with the Euclidean regularizer.
The two-player Rock-Paper-Scissors game has $\agents=\set{1,2}$, $A_i=\set{\mathrm{R},\mathrm{P},\mathrm{S}}\simeq\set{1,2,3}$, and the payoff matrices are given by
\begin{equation}
    \payoffmat{12} = \payoffmat{21} =
    \begin{pmatrix}
        0  & -1 &  1 \\
        1  &  0 & -1 \\
        -1 &  1 &  0
    \end{pmatrix}.
\end{equation}
In this case, the symmetry transformation is given by the following rotation matrices:
\begin{equation}
    \rot_i = e^{\epsilon \rotd},
    \quad
    \rotd = L_1 + L_2 + L_3 = \begin{pmatrix}
        0  &  1 & -1 \\
        -1 &  0 &  1 \\
        1  & -1 &  0
    \end{pmatrix},
\end{equation}
which represent rotations around the $\onevec=(1, 1, 1)^\top$ axis, that is,
\begin{equation}
    \rot_i \onevec = \onevec,  \quad  \rotd \onevec = 0,  \quad \onevec^\top \rotd = 0.
\end{equation}
To see this is a symmetry, we can check that this rotation matrices commute with the payoff matrices, since $\payoffmat{ij}=-\rotd$:
\begin{equation}
    \payoffmat{ij} \rot_j = \rot_i \payoffmat{ij}.
\end{equation}
In addition, the equivariance condition is also satisfied.
A point is that for the Euclidean regularizer we have
\begin{equation}
    \begin{aligned}[b]
        \nabla \hdual{i}(\yy{i})
            &= \yy{i} - \frac{1}{3} \onevec \left( \inner{\onevec}{\yy{i}} - 1 \right)  \\
            &= P_i \yy{i} + \frac{1}{3} \onevec,
    \end{aligned}
\end{equation}
where $P_i = I_3 - \frac{1}{3} \onevec \onevec^\top$ is the projection matrices and commute with the rotation matrices, $P_i \rot_i = \rot_i P_i$.
In this form, we can readily see the equivariance condition with the aforementioned properties for the rotation matrices as
\begin{equation}
    \begin{aligned}[b]
        \rot_i \nabla \hdual{i}(\yy{i})
        &= \rot_i P_i \yy{i} + \frac{1}{3} \rot_i \onevec  \\
        &= P_i \rot_i \yy{i} + \frac{1}{3} \onevec  \\
        &= \nabla \hdual{i}(\rot_i \yy{i}).
    \end{aligned}
\end{equation}
Thus, the FTRL Hamiltonian is invariant under this symmetry transformation, and the conserved quantity is given by
\begin{equation}
    \begin{aligned}[b]
        \gangular(\strategies, y)
        &= \sum_{i=1,2} \inner{\xx{i}}{\rotd \yy{i}}  \\
        &= \sum_{i=1,2} \left( \xx{i}_1 (\yy{i}_2 - \yy{i}_3) + \xx{i}_2 (\yy{i}_3 - \yy{i}_1) + \xx{i}_3 (\yy{i}_1 - \yy{i}_2) \right).
    \end{aligned}
\end{equation}

We also have the other angular momentum-like conserved quantity in the FTRL dynamics of the two-player Matching Pennies game with the Euclidean regularizer.
This game has $\agents=\set{1,2}$, $A_i=\set{\mathrm{H},\mathrm{T}}\simeq\set{1,2}$, and the payoff matrices are given by
\begin{equation}
    \payoffmat{12} = - \payoffmat{21} =
    \begin{pmatrix}
        1  & -1 \\
        -1 &  1
    \end{pmatrix}.
\end{equation}
One finds that this is the case of,
\begin{equation}
    U = \Sigma \otimes S,
    \quad
    \Sigma =
    \begin{pmatrix}
        0  & 1 \\
        -1 & 0
    \end{pmatrix},
    \quad
    S =
    \begin{pmatrix}
        1  & -1 \\
        -1 & 1
    \end{pmatrix}.
\end{equation}
We then have the following rotation matrix as a symmetry transformation:
\begin{equation}
    \rott = e^{\epsilon \Sigma},
    \quad
    \Sigma =
    \begin{pmatrix}
        0 & 1 \\
        -1 & 0
    \end{pmatrix}.
\end{equation}
Thus, the FTRL Hamiltonian is invariant under this symmetry transformation, and the conserved quantity is given by
\begin{align}
    \gangularr(\strategies, y)
        &= \llangle \xx{}, \Sigma \yy{} \rrangle  \\
        &= \sum_{i,j=1,2} \sigma_{ij} \inner{\xx{i}}{\yy{j}}  \\
        &= \inner{\xx{1}}{\yy{2}} - \inner{\xx{2}}{\yy{1}}.
\end{align}

\subsection{Proof of \thmref{thm:dftrl}}
\label{append:proof-dftrl}

\begin{proof}
    From the relations (\ref{eq:consistency}), the learning dynamics is found to satisfy $\xx{i}(t) = \nabla\hdual{i}(\yy{i}(t))$.
    Therefore, it suffices to show that the Fenchel coupling $\gfenchel$ is monotonically decreasing under the modified FTRL dynamics given by Eqs.~(\ref{eq:dftrl}) and (\ref{eq:xftrl-dual}) together with (\ref{eq:perturbation}).
    \begin{equation}
        \begin{aligned}[b]
            \frac{d}{dt}\gfenchel\left( \strategies(t), y(t) \right)
            &= \sum_{i\in\agents} \inner{\dd{\yy{i}(t)}{t}}{\nabla\hdual{i}(\yy{i}(t)) - \xs{i}}  \\
            &= \sum_{i,j} \inner{\payoffmat{ij} \xx{j}(t)}{\xx{i}(t) - \xs{i}} + \coeff \sum_{i} \inner{g^i(\strategies(t), y(t))}{\xx{i}(t) - \xs{i}}
        \end{aligned}
    \end{equation}
    Note that $\strategies(t)$ and $y(t)$ here are the solution to \eref{eq:dftrl}.
    Using the assumptions, we know that the first term vanishes:
    \begin{equation}
        \begin{aligned}[b]
            \sum_{i,j} \inner{\payoffmat{ij} \xx{j}(t)}{\xx{i}(t) - \xs{i}}
                &= \sum_{i,j} \inner{\xx{j}(t)}{\payoffmat{ji} \xs{i}}  \\
                &= \sum_j \lambda^j \sum_a \xx{j}_a(t)  \\
                &= 0,
        \end{aligned}
        \label{eq:fenchel-conserved}
    \end{equation}
    where we used the zero-sum property of the game in the first line, and the property of the Nash equilibrium \eref{eq:nash-property} in the last.
    This is nothing but the consequence of \propref{prop:fenchel}.
    Then, the time evolution of the Fenchel coupling reads as follows:
    \begin{equation}
        \begin{aligned}[b]
            \frac{d}{dt}\gfenchel\left( \strategies(t), y(t) \right)
                &= \coeff \sum_{i\in\agents} \inner{g^i(\strategies(t), y(t))}{\xx{i}(t) - \xs{i}}  \\
                &= \coeff \sum_{i,j} \inner{\payoffmat{ij} \hess\big(\hdual{j}(\yy{j}(t))\big) \sum_{k\in\agents} \payoffmat{jk} \xx{k}(t)}{\xx{i}(t) - \xs{i}}  \\
                &= -\coeff \sum_{j} \inner{\hess\big(\hdual{j}(\yy{j}(t))\big) \sum_{k} \payoffmat{jk} \xx{k}}{\sum_{i} \payoffmat{ji} \left(\xx{i}(t) - \xs{i}\right)}  \\
                &= -\coeff \sum_{j} \inner{\hess\big(\hdual{j}(\yy{j}(t))\big) \sum_{k} \payoffmat{jk} \left(\xx{k}(t) - \xs{k}\right)}{\sum_{i} \payoffmat{ji} \left(\xx{i}(t) - \xs{i}\right)} \\
                &\leq 0,
        \end{aligned}
    \end{equation}
    since $\coeff\geq 0$ and the Hessian of the dual regularizer $h^*$ is positive-definite.
    The equality is attained when $\strategies(t)=\strategies^\ast$ or $\coeff=0$, which reduces to the ordinary FTRL dynamics.
    From the second line to the third, we used the zero-sum property, $\left(\payoffmat{ij}\right)^\top=-\payoffmat{ji}$.
    From the third line to the fourth, we apply the assumptions and use the fact that
    \begin{equation}
        \begin{aligned}[b]
            &\sum_{j} \inner{\hess\big(\hdual{j}(\yy{j}(t))\big) \sum_k \payoffmat{jk} \xs{k}}{\sum_{i} \payoffmat{ji} \left(\xx{i}(t) - \xs{i}\right)}  \\
            &= \sum_{j} \inner{\hess\big(\hdual{j}(\yy{j}(t))\big) \lambda^j\onevec}{\sum_{i} \payoffmat{ji} \left(\xx{i}(t) - \xs{i}\right)}  \\
            &= \sum_{j} \lambda^j \inner{\nabla_{\yy{j}}\left( \sum_a \nabla_{\yy{j}}\hdual{j}(\yy{j}(t))_a \right)}{\sum_{i} \payoffmat{ji} \left(\xx{i}(t) - \xs{i}\right)}  \\
            &= 0.
        \end{aligned}
    \end{equation}
\end{proof}

\section{Continuous FTRL algorithms and perturbative expansion}
\label{append:cftrl}

We derive the discrete-to-continuum limit of the last-iterate convergent FTRL algorithms known in the literature.
We also demonstrate interesting perturbative expansions thereof and show the proof of \propref{prop:co-ceg} along the way.

First, let us describe the setup for the discrete FTRL dynamics.
Let $I= \left[ 0, T \right]$ be the time interval and $L = \set{0, \step, 2\step, \dots, N\step=T}$ be a discrete lattice of $I$ with the step size $\step$;
\begin{equation}
    L \hookrightarrow I \hookrightarrow \R.
\end{equation}
The step size $\step$ is by definition scaled in terms of the number of lattice points as $T/N$ for the fixed $T$.
We denote the discrete learning dynamics by hatted variables such as $\hat{\phi}: L\to \sspace^{N_i} \text{ or } \R^{\abs{A_i}}$.
We assume that there exists a continuous learning dynamics $\phi$, to which $\hat{\phi}$ converges at large $N$ ($\step\to 0$ limit) in an appropriate sense.

\noindent \textbf{Optimistic FTRL.}
The optimistic FTRL algorithm \cite{cai2022finitetime,daskalakis2019last,daskalakis2018training,syrgkanis2015fast,rakhlin2013optimization} is defined by the following update rule:%
\footnote{In the literature, each time step of $\step$ is implicitly assumed. Compared to our notation, our $\hat{\phi}(n\step)$ corresponds to their $x_n$, etc.}
\begin{equation}
    \begin{aligned}
        \xd{i}(t)
            &= \nabla\hdual{i}(\yy{i})\mid_{\yy{i}=\yd{i}(t)},  \\
        \yd{i}(t+\step)
            &= \yd{i}(t) + \step \sum_{j\in\agents} \payoffmat{ij} \xd{j}(t)
                + \coeff \sum_{j\in\agents} \payoffmat{ij} \left( \xd{j}(t) - \xd{j}(t-\step) \right),
    \end{aligned}
\end{equation}
where we assume $\coeff\in\R_{\geq 0}$.
By taking the step size $\step\to 0$ while keeping the combination $N\step=T$, i.e., $n\step=t$, we obtain the continuous optimistic (CO) FTRL dynamics,
\begin{equation}
    \begin{aligned}
        \xx{i}(t)
            &= \nabla\hdual{i}(\yy{i})\mid_{\yy{i}=\yy{i}(t)},  \\
        \dd{\yy{i}(t)}{t}
            &= \sum_{j\in\agents} \payoffmat{ij} \xx{j}(t) + \coeff \sum_{j\in\agents} \payoffmat{ij} \dd{\xx{j}(t)}{t}.
    \end{aligned}
    \label{eq:co-ftrl}
\end{equation}

For $\coeff \ll 1$,%
\footnote{To be more precise, it should be $\coeff/T \ll 1$, since the perturbative coefficient $\coeff$ is a dimensionful parameter.}
we consider the perturbative expansion of $\yy{i}(t)$ in powers of $\coeff$,
\begin{align}
    \yy{i}(t) = \sum_{n=0}^\infty \coeff^n \yy{i}_{(n)}(t),
\end{align}
while $\xx{i}(t)$ are kept untouched.
Substituting this expression into the CO FTRL dynamics \eref{eq:co-ftrl}, we obtain (assuming the term-wise differentiability)
\begin{align}
    \xx{i}(t)
        &= \nabla\hdual{i}(\yy{i}(t)),  \\
    \dd{\yy{i}(t)}{t}
        &= \sum_{n=0}^\infty \coeff^n \dd{\yy{i}_{(n)}(t)}{t}  \nonumber  \\
        &= \sum_{n=0}^\infty \coeff^n \payoffmat{ij_1} H_{j_1} \payoffmat{j_1 j_2}  \cdots  H_{j_n} \payoffmat{j_n j_{n+1}} \xx{j_{n+1}}(t),
\end{align}
where $H_{j_l}\coloneq\hess(\hdual{j_l}(\yy{j_l}))$ and the repeated indices are summed over all possible values.
At order $\coeff^1$, we have
\begin{equation}
    \begin{aligned}
        \xx{i}(t)
            &= \nabla\hdual{i}(\yy{i}(t)),  \\
        \dd{\yy{i}(t)}{t}
            &\simeq \sum_{j\in\agents} \payoffmat{ij} \xx{j}(t) + \coeff \sum_{j,k} \payoffmat{ij} \hess(\hdual{j}(\yy{j}(t))) \payoffmat{jk} \xx{k}(t),
    \end{aligned}
\end{equation}
which is equivalent to our DFTRL algorithm, \eref{eq:dftrl2}.
This proves half of \propref{prop:co-ceg}.

\noindent \textbf{Extra-gradient FTRL.}
The extra-gradient FTRL algorithm \cite{cai2022finitetime,lee2021fast,mertikopoulos2018optimistic,Korpelevich1976TheEM} is defined through
\begin{equation}
    \begin{aligned}
        \xd{i}(t)
            &= \nabla\hdual{i}(\yy{i})\mid_{\yy{i}=\yd{i}(t)},  \\
        \yd{i}(t+\step)
            &= \yd{i}(t) + \step \sum_{j\in\agents} \payoffmat{ij} \nabla\hdual{j}(\yy{j})\mid_{\yy{j}=\yd{j}(t)+\coeff\sum_k\payoffmat{jk}\xd{k}(t)}.
    \end{aligned}
    \label{eq:eg-ftrl}
\end{equation}
In the literature, this update rule is often expressed using auxiliary variables $\hat{b}^i$ and $\hat{c}^i$ as follows:
\begin{align}
    \xd{i}(t)
        &= \nabla\hdual{i}(\yy{i})\mid_{\yy{i}=\yd{i}(t)},  \\
    \hat{b}^i(t)
        &= \yd{i}(t) + \coeff \sum_{j\in\agents} \payoffmat{ij} \xd{j}(t),  \\
    \hat{c}^i(t)
        &= \nabla\hdual{i}(\yy{i})\mid_{\yy{i}=\hat{b}^i(t)},  \\
    \yd{i}(t+\step)
        &= \yd{i}(t) + \step \sum_{j\in\agents} \payoffmat{ij} \hat{c}^j(t).
\end{align}
By eliminating the auxiliary variables $\hat{b}^i$ and $\hat{c}^i$, one finds that these equations reduce to the aforementioned extra-gradient FTRL update rule \eref{eq:eg-ftrl}.
In the above four equations, the second and fourth equations imply that during one time step $\step$ from $\yd{i}(t)$ to $\yd{i}(t+\step)$, it acquires an ``extra gradient'' due to the additional term proportional to $\coeff$:
\begin{equation}
    \coeff \sum_{j\in\agents} \payoffmat{ij} \xd{j}(t).
\end{equation}

Now, by taking the step size $\step\to 0$, we derive the continuous extra-gradient (CEG) FTRL dynamics from \eref{eq:eg-ftrl},
\begin{equation}
    \begin{aligned}
        \xx{i}(t)
            &= \nabla\hdual{i}(\yy{i})\mid_{\yy{i}=\yy{i}(t)},  \\
        \dd{\yy{i}(t)}{t}
            &= \sum_{j\in\agents} \payoffmat{ij} \nabla\hdual{j}(\yy{j})\mid_{\yy{j}=\yy{j}(t)+\coeff\sum_k\payoffmat{jk}\xx{k}(t)}.
    \end{aligned}
    \label{eq:ceg-ftrl}
\end{equation}
For $\coeff\ll 1$, we again consider the perturbative expansion of $\yy{i}(t)$ in powers of $\coeff$.
It should be noted that for a multivariable function $f$, its Taylor expansion is given by
\begin{align}
    f(x + c) = \sum_{n=0}^\infty \frac{1}{n!} \sum_{a_1,\dots,a_n} \left( \frac{\partial^n}{\partial x_{a_1} \cdots \partial x_{a_n}} f(x) \right) c_{a_1}\cdots c_{a_n}.
\end{align}
Then, we find that
\begin{equation}
    \dd{\yy{i}_a(t)}{t}
        = \sum_j \sum_{a'} \payoffmat{ij}_{aa'} \sum_{n=0}^\infty \frac{1}{n!} \sum_{a_1,\dots,a_n} \left( \frac{\partial^{n+1}}{\partial \yy{j}_{a_1} \cdots \partial \yy{j}_{a_n} \partial \yy{j}_{a'}} \hdual{j}(\yy{j}(t)) \right)
        \left( \coeff\sum_k\payoffmat{jk}\xx{k}(t) \right)_{a_1}  \cdots  \left( \coeff\sum_k\payoffmat{jk}\xx{k}(t) \right)_{a_n}.
\end{equation}
Thus, at order $\coeff^1$, the CEG FTRL dynamics \eref{eq:ceg-ftrl} becomes
\begin{equation}
    \begin{aligned}
        \xx{i}(t)
            &= \nabla\hdual{i}(\yy{i}(t)),  \\
        \dd{\yy{i}(t)}{t}
            &\simeq \sum_{j\in\agents} \payoffmat{ij} \nabla\hdual{j}(\yy{j}(t)) + \coeff \sum_{j,k} \payoffmat{ij} \hess(\hdual{j}(\yy{j}(t))) \payoffmat{jk} \xx{k}(t),
    \end{aligned}
\end{equation}
which is equivalent to our DFTRL algorithm, \eref{eq:dftrl2}.
This proves the other half of \propref{prop:co-ceg}.

\noindent \textbf{Negative-momentum FTRL.}
The negative-momentum FTRL algorithm \cite{hemmat2023lead,zhang2021suboptimality,kovachki2021continuous,gidel2019negative,polyak1964methods} is defined by the following update rule:
\begin{equation}
    \begin{aligned}
        \xd{i}(t)
            &= \nabla\hdual{i}(\yy{i})\mid_{\yy{i}=\yd{i}(t)},  \\
        \yd{i}(t+\step)
            &= \yd{i}(t) + \step \sum_{j\in\agents} \payoffmat{ij} \xd{j}(t) - \coeff \left( \yd{i}(t) - \yd{i}(t-\step) \right).
    \end{aligned}
\end{equation}
By taking the step size $\step\to 0$, we obtain the continuous negative-momentum (CNM) FTRL dynamics,
\begin{equation}
    \begin{aligned}
        \xx{i}(t)
            &= \nabla\hdual{i}(\yy{i})\mid_{\yy{i}=\yy{i}(t)},  \\
        \dd{\yy{i}(t)}{t}
            &= \left( 1 + \coeff \right)^{-1} \sum_{j\in\agents} \payoffmat{ij} \xx{j}(t).
    \end{aligned}
    \label{eq:cnm-ftrl}
\end{equation}
This dynamics is essentially the same as the ordinary FTRL dynamics, since an overall constant to $\yy{i}$ does not affect the qualitative behavior of the dynamics.
In fact, one can readily confirm that the Fenchel coupling $\gfenchel$ is conserved under \eref{eq:cnm-ftrl}:
\begin{equation}
    \begin{aligned}[b]
        \frac{d}{dt} \gfenchel(\strategies(t), y(t))
            &= \sum_{i\in\agents} \inner{\dd{\yy{i}(t)}{t}}{\nabla\hdual{i}(\yy{i}(t)) - \xs{i}}  \\
            &= \left( 1 + \coeff \right)^{-1} \sum_{i,j} \inner{\payoffmat{ij} \xx{j}(t)}{\xx{i}(t) - \xs{i}}  \\
            &= 0,
    \end{aligned}
\end{equation}
where in the last step we used the fact \eref{eq:fenchel-conserved}.
Thus, the solution to \eref{eq:cnm-ftrl} does not converge in general, regardless of the value of $\alpha$.
As this example demonstrates, it is generically non-trivial whether the continuum limit of a discrete convergent dynamics turns out to be convergent again, and vice versa.

\section{Additional results}
\label{append:additional-results}

\subsection{Extra conservation laws in the FTRL Hamiltonian system}
\label{append:extra-conservation-laws}

We provide two extra conservation laws in the FTRL Hamiltonian system.
These two represent the conservation laws under the translation symmetry and the dilation symmetry of the FTRL Hamiltonian, respectively.
In this subsection, we always have in mind that we have a game $\zgame$ with the Hamiltonian (\ref{eq:ftrl-hamiltonian}).

\begin{prop*}
    If there exist $c^i\in\R^{\abs{A_i}},~ i\in\agents$, such that
    \begin{equation}
        \sum_{i\in\agents} \payoffmat{ji} c^i = 0,
    \end{equation}
    for all $j\in\agents$, then the following function is a conserved quantity of the Hamiltonian system:
    \begin{equation}
        \gcumul(\strategies, y) = \sum_{i\in\agents} \inner{c^i}{\yy{i}}.
    \end{equation}
\end{prop*}

\begin{proof}
    The FTRL Hamiltonian (\ref{eq:ftrl-hamiltonian}) is invariant under the canonical transformation,
    \begin{equation}
        x'^i = x^i + \epsilon c^i,  \quad  y'^i = y^i,
    \end{equation}
    for all $i\in\agents$:
    \begin{equation}
        \begin{aligned}[b]
            \hftrl(\strategies', y') - \hftrl(\strategies, y)
                &= \epsilon \sum_j \inner{\nabla\hdual{j}(\yy{j})}{\sum_i \payoffmat{ji} c^i}  \\
                &= 0.
        \end{aligned}
    \end{equation}
    A generator $g\in C^\infty(M_{\zgame})$ of this canonical transformation satisfies
    \begin{align}
        \pbracket{\xx{i}_a}{g} &=  \dpdp{g(x,y)}{\yy{i}_a} = c^i_a,  \\
        \pbracket{\yy{i}_a}{g} &= -\dpdp{g(x,y)}{\xx{i}_a} = 0.
    \end{align}
    We thus obtain $g(x,y)=\sum_{i\in\agents} \inner{c^i}{\yy{i}} + \text{const.}$
    By Noether's theorem, the function $g$ is conserved in the FTRL Hamiltonian system, implying that $\gcumul$ is conserved.
\end{proof}

When $c^i=\onevec$, this proposition shows the conservation of the total cumulative payoff, which is quite reasonable for zero-sum games.

\begin{prop*}
    If the regularizers satisfy the ``homogeneity'' condition,
    \begin{equation}
        \hess\left( \hdual{i}(\yy{i}) \right) \yy{i} = \nabla\hdual{i}(\yy{i}),
    \end{equation}
    for all $i$ and $\yy{i}$, then the following function is a conserved quantity of the FTRL Hamiltonian system:
    \begin{equation}
        \gdilation(\strategies, y) = \sum_{i\in\agents} \inner{\xx{i}}{\yy{i}}.
    \end{equation}
\end{prop*}

\begin{proof}
    Under the assumption, the FTRL Hamiltonian (\ref{eq:ftrl-hamiltonian}) has the dilation symmetry, namely, it is invariant under the canonical transformation,
    \begin{equation}
        x'^i = x^i + \epsilon x^i,  \quad  y'^i = y^i - \epsilon y^i,
    \end{equation}
    for all $i\in\agents$:
    \begin{equation}
        \begin{aligned}[b]
            \hftrl(\strategies', y') - \hftrl(\strategies, y)
                &= \sum_{i,j} \bigg[ \inner{\nabla\hdual{j}(\yy{j})}{\payoffmat{ji} \epsilon \xx{i}} + \inner{\hess\left( \hdual{j}(\yy{j}) \right) (-\epsilon \yy{i})}{\payoffmat{ji} \xx{i}} \bigg]  \\
                &= \epsilon \sum_{i,j} \inner{\nabla\hdual{j}(\yy{j}) - \hess\left( \hdual{j}(\yy{j}) \right) \yy{i}}{\payoffmat{ji} \xx{i}}  \\
                &= 0.
        \end{aligned}
    \end{equation}
    A generator $g\in C^\infty(M_{\zgame})$ of this canonical transformation satisfies
    \begin{align}
        \pbracket{\xx{i}_a}{g} &=  \dpdp{g(x,y)}{\yy{i}_a} = \xx{i}_a,  \\
        \pbracket{\yy{i}_a}{g} &= -\dpdp{g(x,y)}{\xx{i}_a} = -\yy{i}_a.
    \end{align}
    We thus obtain $g(x,y)=\sum_{i\in\agents} \inner{\xx{i}}{\yy{i}} + \text{const.}$
    By Noether's theorem, the function $g$ is conserved in the FTRL Hamiltonian system, implying that $\gdilation$ is conserved.
\end{proof}

\subsection{A generalization of DFTRL}
\label{append:dftrl-general}

By extending the argument in the proof of \thmref{thm:dftrl}, i.e., \appenref{append:proof-dftrl}, we identify a potential generalization of the DFTRL algorithm.

\begin{prop*}
    \label{prop:dftrl-general}
    For $m\in\Z_{\geq 0}$, the DFTRL algorithm (\thmref{thm:dftrl}) can be generalized to the $(4m+1)$-power perturbations:
    \begin{equation}
        \begin{aligned}[b]
            g^i(x, y)
                &= \payoffmat{i \blacktriangle} \left( H U \right)^{4m+1} \xx{\blacktriangledown}  \\
                &\coloneq \payoffmat{ij_1} H_{j_1} \payoffmat{j_1 j_2}  \cdots  H_{j_{4m+1}} \payoffmat{j_{4m+1} j_{4m+2}} \xx{j_{4m+2}},
        \end{aligned}
        \label{eq:perturbation-general}
    \end{equation}
    where $H_{j_l}\coloneq\hess(\hdual{j_l}(\yy{j_l}))$ and the repeated indices are summed over all possible values.
\end{prop*}

\begin{proof}
    The proof follows straightforwardly by expanding the argument in the proof of \thmref{thm:dftrl}.
    Under the modified FTRL dynamics Eqs.~(\ref{eq:dftrl}) and (\ref{eq:xftrl-dual}) together with (\ref{eq:perturbation-general}), the time evolution of the Fenchel coupling reads as
    \begin{equation}
        \begin{aligned}
            \frac{d}{dt}\gfenchel\left( \strategies(t), y(t) \right)
                &= \coeff \sum_{i\in\agents} \inner{g^i(\strategies(t), y(t))}{\xx{i}(t) - \xs{i}}  \\
                &= \coeff \inner{\payoffmat{ij_1} H_{j_1} \payoffmat{j_1 j_2}  \cdots  H_{j_{4m+1}} \payoffmat{j_{4m+1} j_{4m+2}} \xx{j_{4m+2}}(t)}{\xx{i}(t) - \xs{i}}.
        \end{aligned}
    \end{equation}
    It is understood that the repeated indices are summed over all possible values.
    By moving the $2m+1$ payoff matrices $\payoffmat{ij_1},\dots,\payoffmat{j_{2m} j_{2m+1}}$ to the opposite side in the inner product, along with the $2m$ Hessians, we have
    \begin{equation}
        \frac{d}{dt}\gfenchel\left( \strategies(t), y(t) \right)
            = -\coeff \inner{H_{j_{2m+1}} (UH)^{2m} \payoffmat{\blacktriangle j} \left(\xx{j}(t) - \xs{j}\right)}{(UH)^{2m} \payoffmat{\blacktriangledown i}\left(\xx{i}(t) - \xs{i}\right)}.
    \end{equation}
    The overall minus sign arises from the zero-sum property of the $2m+1$ payoff matrices.
    The following abbreviations are introduced:
    \begin{align}
        (UH)^{2m} \payoffmat{\blacktriangle j}
            &= \payoffmat{j_{2m+1} j_{2m+2}} H_{j_{2m+2}}  \cdots  \payoffmat{j_{4m} j_{4m+1}} H_{j_{4m+1}} \payoffmat{j_{4m+1} j},  \\
        (UH)^{2m} \payoffmat{\blacktriangledown i}
            &= \payoffmat{j_{2m+1} j_{2m}} H_{j_{2m}}  \cdots  \payoffmat{j_2 j_1} H_{j_1} \payoffmat{j_1 i}.
    \end{align}
    Since the repeated indices are summed over and hence act as the dummies, one finds that the vectors on both sides of the Hessian $H_{j_{2m+1}}$ are the same.
    Thus, from the positive-definiteness of the Hessian, we obtain
    \begin{equation}
        \frac{d}{dt}\gfenchel\left( \strategies(t), y(t) \right) \leq 0.
    \end{equation}
    The equality is attained when $\strategies(t)=\strategies^\ast$ or $\coeff=0$, for which the system reduces to the ordinary FTRL dynamics.
\end{proof}

This generalized DFTRL algorithm reduces to \thmref{thm:dftrl} when $m=0$.
As discussed in the main text, the Hamiltonian (\ref{eq:ftrl-hamiltonian}) admits a one-parameter deformation of the dynamics that preserves the Hamiltonian.
As an example, we examined the damping of the Fenchel coupling, for which \eref{eq:perturbation} was designed.
While we have not yet been able to demonstrate its uniqueness, it is likely that any one-parameter deformation without derivative terms that damps the Fenchel coupling while conserving the Hamiltonian (\ref{eq:ftrl-hamiltonian}) will, up to a constant, be equivalent to \eref{eq:perturbation-general}.
The generalized perturbation (\ref{eq:perturbation-general}) might have additional nice properties for learning in games, such as an improvement in the convergence rate.
However, we do not discuss these aspects here and leave them for future work.

\section{Experimental details and additional simulations}
\label{append:experiments}

In \sref{sec:experiments}, we numerically solve the FTRL/DFTRL dynamics for the well-known examples.
Beyond the typical setups for experimental simulations given in the main text, we provide two additional examples to validate our theoretical results.
To this end, we begin with more detailed descriptions of the experimental setups.

\begin{figure}[t]
    \begin{minipage}[b]{0.246\linewidth}
        \centering
        \includegraphics[keepaspectratio, scale=0.22]{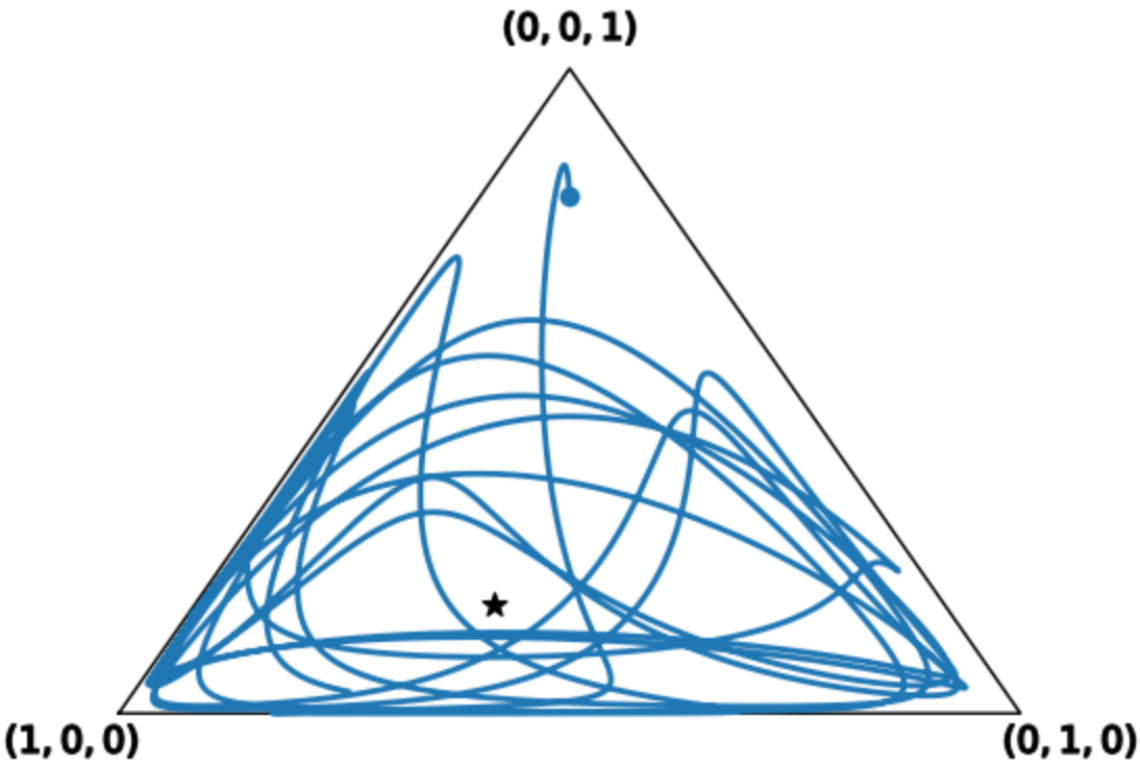}
        \subcaption{FTRL ($\coeff=0$)}
        \label{fig:rpsw0}
    \end{minipage}
    \begin{minipage}[b]{0.246\linewidth}
        \centering
        \includegraphics[keepaspectratio, scale=0.22]{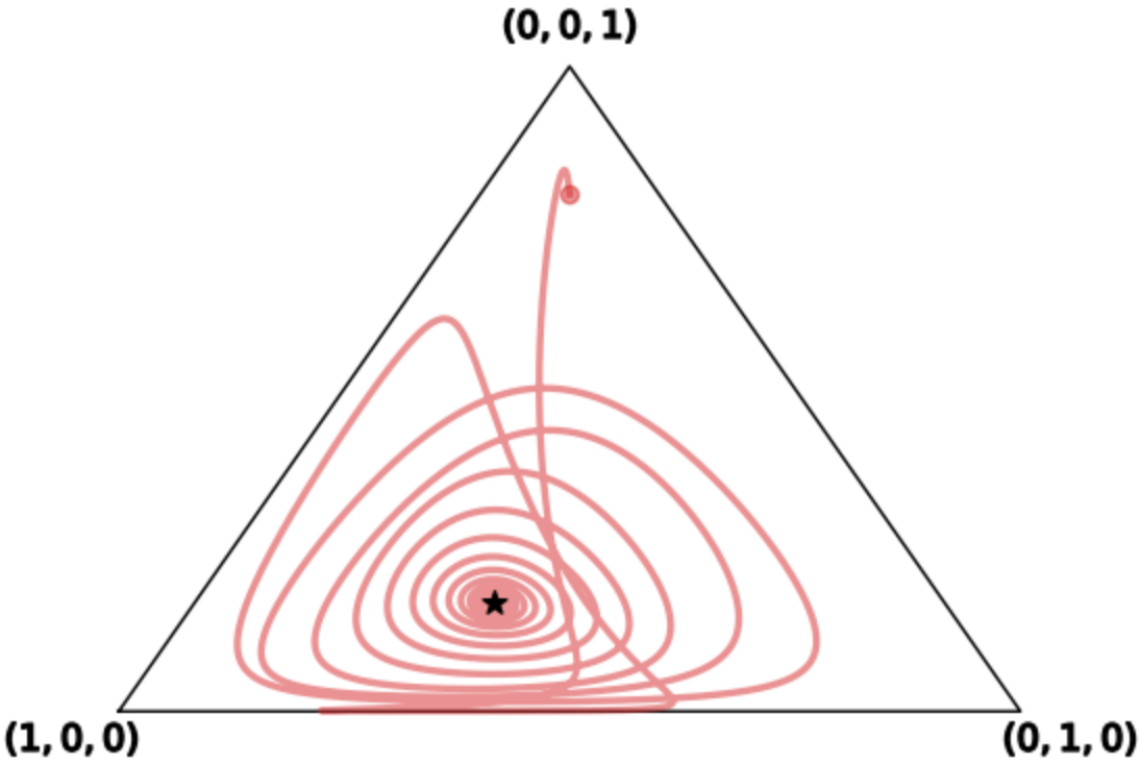}
        \subcaption{DFTRL ($\coeff=0.05$)}
        \label{fig:rpsw05}
    \end{minipage}
    \begin{minipage}[b]{0.246\linewidth}
        \centering
        \includegraphics[keepaspectratio, scale=0.22]{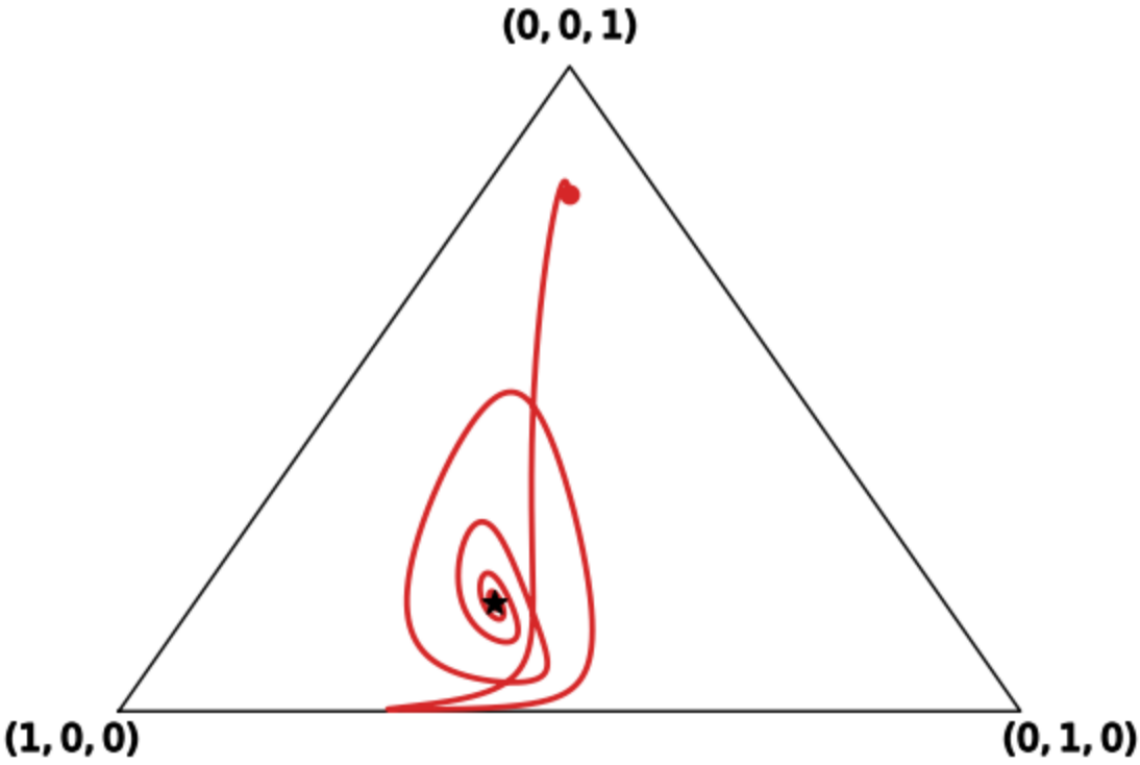}
        \subcaption{DFTRL ($\coeff=0.15$)}
        \label{fig:rpsw1}
    \end{minipage}
    \begin{minipage}[b]{0.246\linewidth}
        \centering
        \includegraphics[keepaspectratio, scale=0.23]{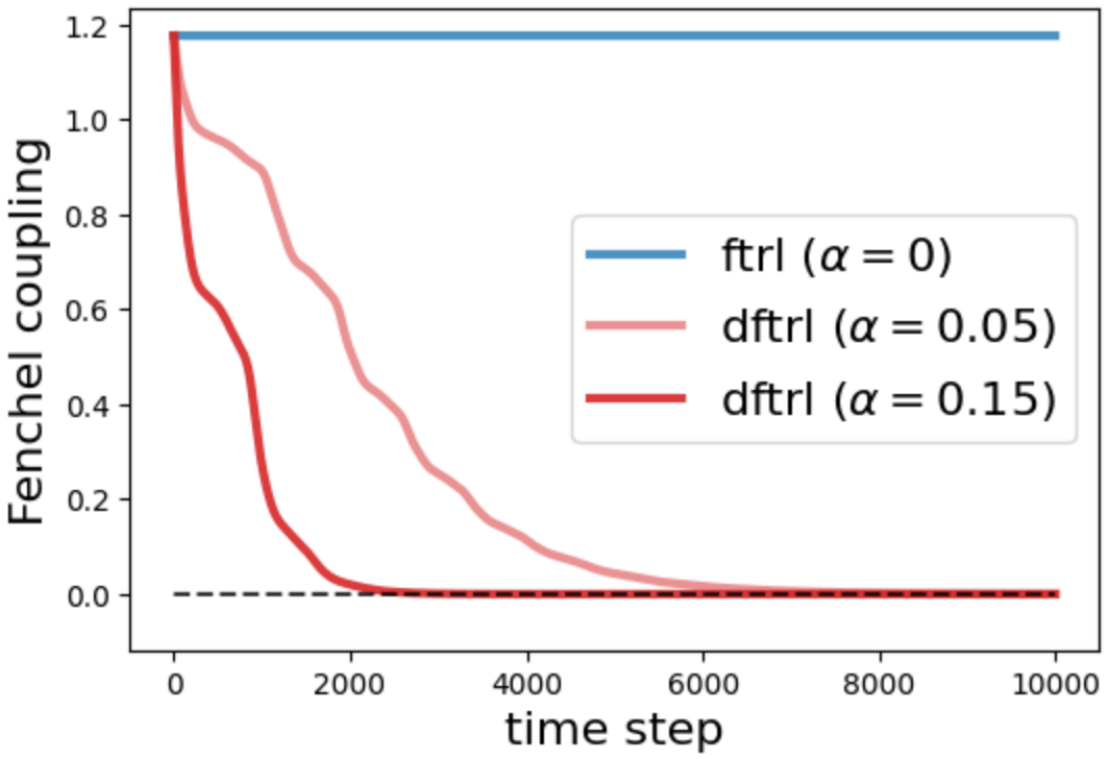}
        \subcaption{Fenchel coupling $\gfenchel$}
        \label{fig:fenchel-rpsw}
    \end{minipage}
    \caption{Two-player weighted Rock-Paper-Scissors.}
    \label{fig:rpsw}
\end{figure}

Let us examine the two-player \emph{weighted} Rock-Paper-Scissors game, where $\agents=\set{1,2}$ and $A_i=\set{\mathrm{R},\mathrm{P},\mathrm{S}}\simeq\set{1,2,3}$.%
\footnote{
    Since this paper has established the Hamiltonian structure for learning dynamics in games, also known as evolutionary dynamics,
    applying or extending our results to various analyses of Rock-Paper-Scissors games using physics-based approaches (such as \cite{kabir2021role,tenorio2022adaptive,justino2025critical}) could yield new insights.
}
This game is characterized by the payoff matrices,
\begin{equation}
    \payoffmat{12} = \payoffmat{21} =
    \begin{pmatrix}
        0  & -a  &  b \\
        a  &  0  & -c \\
        -b &  c  &  0
    \end{pmatrix},
\end{equation}
where $a,b,c\in\R$.
In this game, the Nash equilibrium is explicitly given by
\begin{equation}
    \xs{i} = \left( \frac{c}{a+b+c}, \frac{b}{a+b+c}, \frac{a}{a+b+c} \right)
    \qquad
    i=1,2.
\end{equation}
In addition to these data, a choice of regularizer functions defines the dynamical system \eref{eq:dftrl}.
Given initial conditions $\xx{1}(0)$ and $\xx{2}(0)$, the dynamical system with the relations (\ref{eq:consistency}) and (\ref{eq:perturbation}) generates the solution.
In \sref{sec:experiments}, we take $a=b=c=1$, the entropic regularizer $\hdual{i}(\yy{i})=\mathrm{lse}(\yy{i})$, and the initial conditions $\xx{1}(0)=\xx{2}(0)=(0.1, 0.1, 0.8)$, which result in \fref{fig:rps}.
We now consider the case of $a=1$, $b=2$, $c=3$, and the initial conditions $\xx{1}(0)=(0.1, 0.1, 0.8)$ and $\xx{2}(0)=(0.2, 0.6, 0.2)$, where the entropic regularizer is kept.
Figure \ref{fig:rpsw} depicts the results.
As one can see, the solution trajectory for $\coeff=0$ (FTRL) exhibits non-convergent behavior, and its Fenchel coupling is conserved over time.
Our DFTRL algorithm yields convergent dynamics and we see that the larger perturbation coefficient $\coeff$ promotes the faster convergence, as expected.

\begin{figure}[t]
    \begin{minipage}[b]{0.329\linewidth}
        \centering
        \includegraphics[keepaspectratio, scale=0.16]{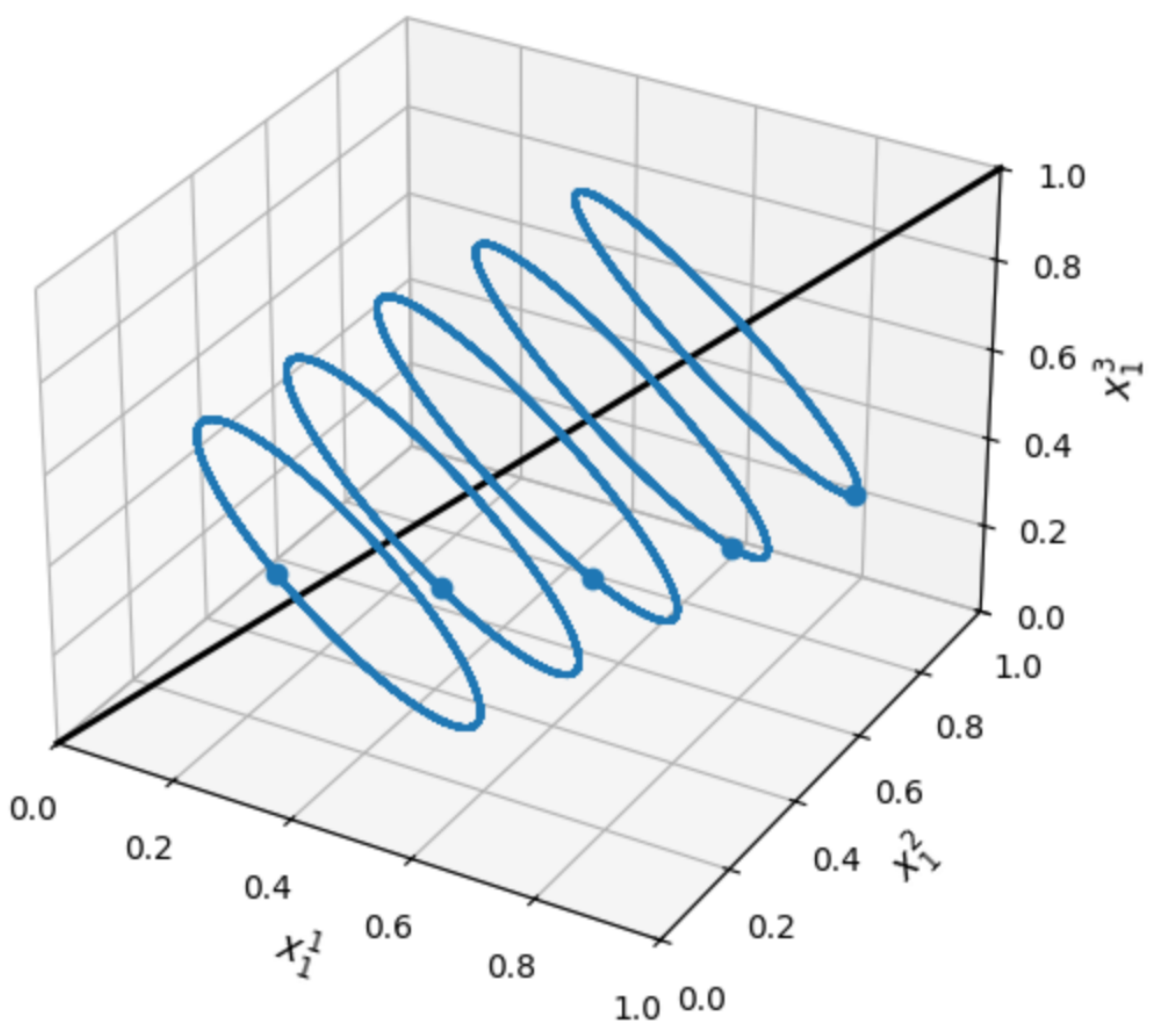}
        \subcaption{FTRL ($\coeff=0$)}
        \label{fig:mpe0}
    \end{minipage}
    \begin{minipage}[b]{0.329\linewidth}
        \centering
        \includegraphics[keepaspectratio, scale=0.16]{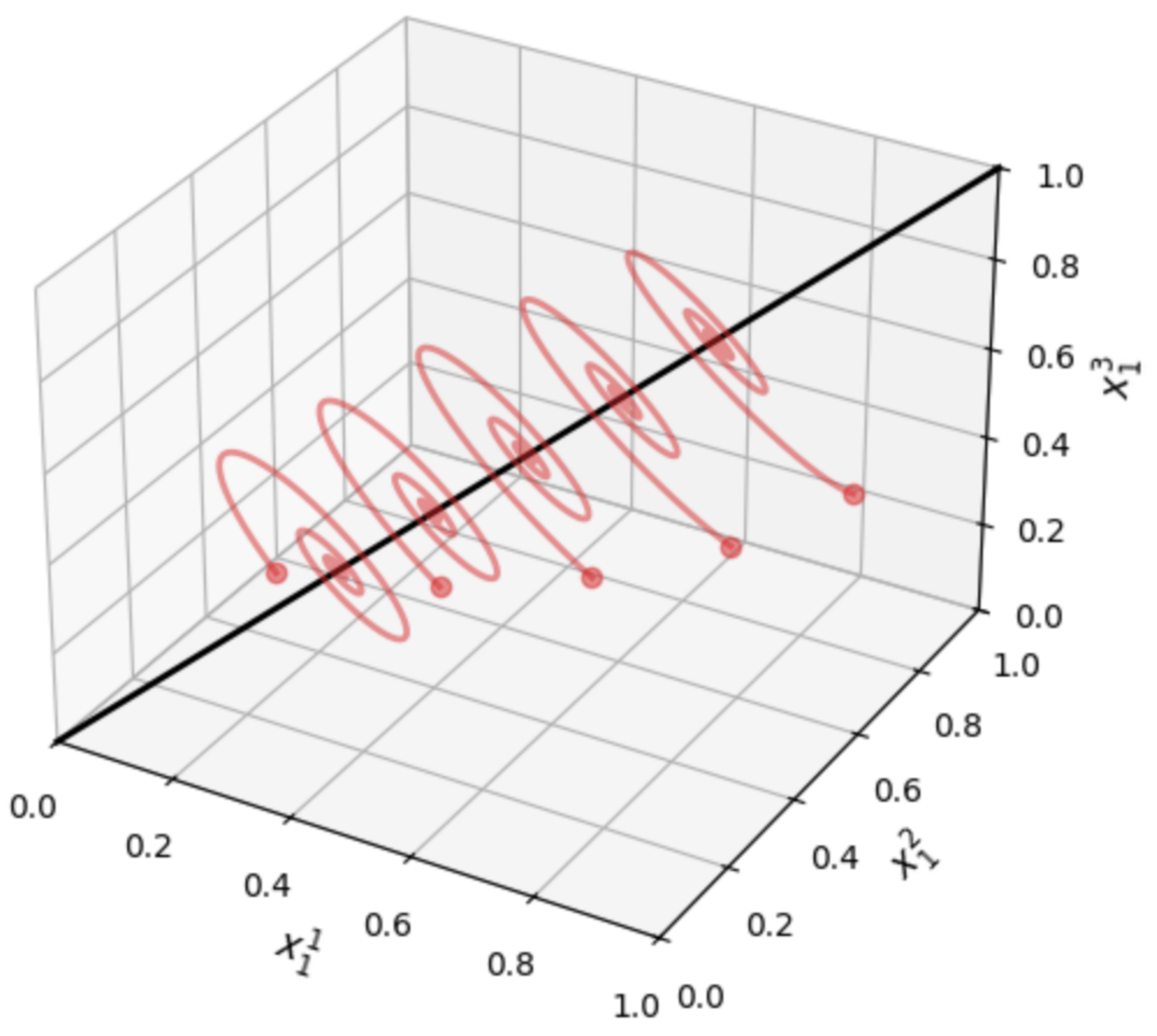}
        \subcaption{DFTRL ($\coeff=0.05$)}
        \label{fig:mpe05}
    \end{minipage}
    \begin{minipage}[b]{0.329\linewidth}
        \centering
        \includegraphics[keepaspectratio, scale=0.16]{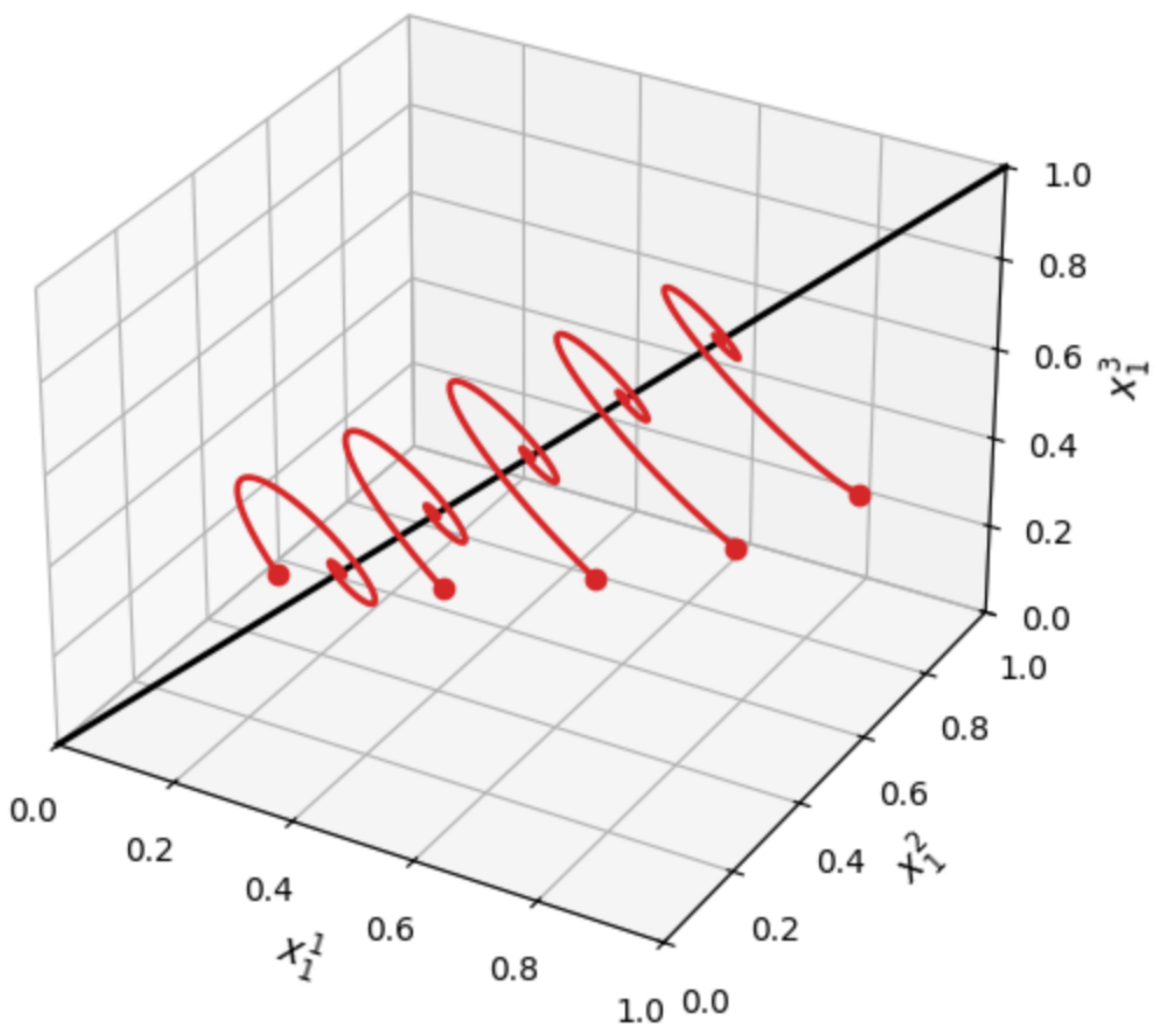}
        \subcaption{DFTRL ($\coeff=0.1$)}
        \label{fig:mpe1}
    \end{minipage}
    \caption{Three-player Matching Pennies.}  
    \label{fig:mpe}
\end{figure}

We also consider the three-player Matching Pennies, with different regularizer functions.
This game is characterized by $\agents=\set{1,2,3}$ and $A_i=\set{\mathrm{H}, \mathrm{T}}\simeq\set{1,2}$.
We set the payoff matrices as
\begin{equation}
    \payoffmat{12} =
    \payoffmat{23} =
    \payoffmat{31} =
    \begin{pmatrix}
        a  & -1  \\
        -1 &  a
    \end{pmatrix},
    \quad
    \payoffmat{ij} = -\left( \payoffmat{ji} \right)^\top,
\end{equation}
where $a\in\R$.
In this game, we have a continuum of Nash equilibria distributed along a straight line:
\begin{equation}
    (\xs{1}, \xs{2}, \xs{3}) = \left( (p, 1-p), (p, 1-p), (p, 1-p) \right),
\end{equation}
where $0\leq p \leq 1$.
Again, building on these data with a choice of regularizer functions, we solve the dynamical system \eref{eq:dftrl} for a given initial condition.
In \fref{fig:mp}, we take $a=1$, the entropic regularizer, and initial conditions chosen at random.
Here, we employ the Euclidean regularizer, while the other setups are kept intact.
The result is shown in \fref{fig:mpe}.
Similarly to the previous example, we observe that the solution trajectories for $\coeff=0$ (FTRL) exhibit cyclic behavior, and that the DFTRL dynamics for larger $\coeff$ converges faster to the Nash equilibrium.

\noindent \textbf{Reproducibility statement.}
All the experiments in this paper are conducted on an Intel(R) Xeon(R) CPU @ 2.20GHz.
The operating system is Ubuntu 22.04.4 LTS.
We use the \texttt{torchdiffeq} framework \cite{chen2018neuralode,torchdiffeq} for solving the differential equations.
Running the code for each simulation takes at most a few minutes.
The code is publicly available at \url{\githubrepo}.  

\bibliography{references}

\begin{thebibliography}{57}%
\makeatletter
\providecommand \@ifxundefined [1]{%
 \@ifx{#1\undefined}
}%
\providecommand \@ifnum [1]{%
 \ifnum #1\expandafter \@firstoftwo
 \else \expandafter \@secondoftwo
 \fi
}%
\providecommand \@ifx [1]{%
 \ifx #1\expandafter \@firstoftwo
 \else \expandafter \@secondoftwo
 \fi
}%
\providecommand \natexlab [1]{#1}%
\providecommand \enquote  [1]{``#1''}%
\providecommand \bibnamefont  [1]{#1}%
\providecommand \bibfnamefont [1]{#1}%
\providecommand \citenamefont [1]{#1}%
\providecommand \href@noop [0]{\@secondoftwo}%
\providecommand \href [0]{\begingroup \@sanitize@url \@href}%
\providecommand \@href[1]{\@@startlink{#1}\@@href}%
\providecommand \@@href[1]{\endgroup#1\@@endlink}%
\providecommand \@sanitize@url [0]{\catcode `\\12\catcode `\$12\catcode `\&12\catcode `\#12\catcode `\^12\catcode `\_12\catcode `\%12\relax}%
\providecommand \@@startlink[1]{}%
\providecommand \@@endlink[0]{}%
\providecommand \url  [0]{\begingroup\@sanitize@url \@url }%
\providecommand \@url [1]{\endgroup\@href {#1}{\urlprefix }}%
\providecommand \urlprefix  [0]{URL }%
\providecommand \Eprint [0]{\href }%
\providecommand \doibase [0]{https://doi.org/}%
\providecommand \selectlanguage [0]{\@gobble}%
\providecommand \bibinfo  [0]{\@secondoftwo}%
\providecommand \bibfield  [0]{\@secondoftwo}%
\providecommand \translation [1]{[#1]}%
\providecommand \BibitemOpen [0]{}%
\providecommand \bibitemStop [0]{}%
\providecommand \bibitemNoStop [0]{.\EOS\space}%
\providecommand \EOS [0]{\spacefactor3000\relax}%
\providecommand \BibitemShut  [1]{\csname bibitem#1\endcsname}%
\let\auto@bib@innerbib\@empty
\bibitem [{\citenamefont {Fudenberg}\ and\ \citenamefont {Levine}(1998)}]{fudenberg1998theory}%
  \BibitemOpen
  \bibfield  {author} {\bibinfo {author} {\bibfnamefont {D.}~\bibnamefont {Fudenberg}}\ and\ \bibinfo {author} {\bibfnamefont {D.~K.}\ \bibnamefont {Levine}},\ }\href@noop {} {\emph {\bibinfo {title} {The theory of learning in games}}},\ Vol.~\bibinfo {volume} {2}\ (\bibinfo  {publisher} {MIT press},\ \bibinfo {year} {1998})\BibitemShut {NoStop}%
\bibitem [{\citenamefont {Nash}(1950)}]{nash1950equilibrium}%
  \BibitemOpen
  \bibfield  {author} {\bibinfo {author} {\bibfnamefont {J.}~\bibnamefont {Nash}},\ }\bibfield  {title} {\bibinfo {title} {Equilibrium points in n-person games},\ }\href@noop {} {\bibfield  {journal} {\bibinfo  {journal} {Proceedings of the national academy of sciences}\ }\textbf {\bibinfo {volume} {36}},\ \bibinfo {pages} {48} (\bibinfo {year} {1950})}\BibitemShut {NoStop}%
\bibitem [{\citenamefont {Mertikopoulos}\ \emph {et~al.}(2018)\citenamefont {Mertikopoulos}, \citenamefont {Papadimitriou},\ and\ \citenamefont {Piliouras}}]{mertikopoulos2018cycles}%
  \BibitemOpen
  \bibfield  {author} {\bibinfo {author} {\bibfnamefont {P.}~\bibnamefont {Mertikopoulos}}, \bibinfo {author} {\bibfnamefont {C.}~\bibnamefont {Papadimitriou}},\ and\ \bibinfo {author} {\bibfnamefont {G.}~\bibnamefont {Piliouras}},\ }\bibfield  {title} {\bibinfo {title} {Cycles in adversarial regularized learning},\ }in\ \href@noop {} {\emph {\bibinfo {booktitle} {Proceedings of the twenty-ninth annual ACM-SIAM symposium on discrete algorithms}}}\ (\bibinfo {organization} {SIAM},\ \bibinfo {year} {2018})\ pp.\ \bibinfo {pages} {2703--2717}\BibitemShut {NoStop}%
\bibitem [{\citenamefont {Shalev-Shwartz}\ and\ \citenamefont {Singer}(2006)}]{shalev2006convex}%
  \BibitemOpen
  \bibfield  {author} {\bibinfo {author} {\bibfnamefont {S.}~\bibnamefont {Shalev-Shwartz}}\ and\ \bibinfo {author} {\bibfnamefont {Y.}~\bibnamefont {Singer}},\ }\bibfield  {title} {\bibinfo {title} {Convex repeated games and {F}enchel duality},\ }\href@noop {} {\bibfield  {journal} {\bibinfo  {journal} {Advances in neural information processing systems}\ }\textbf {\bibinfo {volume} {19}} (\bibinfo {year} {2006})}\BibitemShut {NoStop}%
\bibitem [{\citenamefont {Hofbauer}(1996)}]{hofbauer1996evolutionary}%
  \BibitemOpen
  \bibfield  {author} {\bibinfo {author} {\bibfnamefont {J.}~\bibnamefont {Hofbauer}},\ }\bibfield  {title} {\bibinfo {title} {Evolutionary dynamics for bimatrix games: A {H}amiltonian system?},\ }\href@noop {} {\bibfield  {journal} {\bibinfo  {journal} {Journal of mathematical biology}\ }\textbf {\bibinfo {volume} {34}},\ \bibinfo {pages} {675} (\bibinfo {year} {1996})}\BibitemShut {NoStop}%
\bibitem [{\citenamefont {Bailey}\ and\ \citenamefont {Piliouras}(2019)}]{bailey2019multi}%
  \BibitemOpen
  \bibfield  {author} {\bibinfo {author} {\bibfnamefont {J.~P.}\ \bibnamefont {Bailey}}\ and\ \bibinfo {author} {\bibfnamefont {G.}~\bibnamefont {Piliouras}},\ }\bibfield  {title} {\bibinfo {title} {Multi-agent learning in network zero-sum games is a {H}amiltonian system},\ }in\ \href@noop {} {\emph {\bibinfo {booktitle} {Proceedings of the 18th International Conference on Autonomous Agents and MultiAgent Systems}}}\ (\bibinfo {year} {2019})\ pp.\ \bibinfo {pages} {233--241}\BibitemShut {NoStop}%
\bibitem [{\citenamefont {Daskalakis}\ and\ \citenamefont {Papadimitriou}(2009)}]{daskalakis2009network}%
  \BibitemOpen
  \bibfield  {author} {\bibinfo {author} {\bibfnamefont {C.}~\bibnamefont {Daskalakis}}\ and\ \bibinfo {author} {\bibfnamefont {C.~H.}\ \bibnamefont {Papadimitriou}},\ }\bibfield  {title} {\bibinfo {title} {On a network generalization of the minmax theorem},\ }in\ \href@noop {} {\emph {\bibinfo {booktitle} {International Colloquium on Automata, Languages, and Programming}}}\ (\bibinfo {organization} {Springer},\ \bibinfo {year} {2009})\ pp.\ \bibinfo {pages} {423--434}\BibitemShut {NoStop}%
\bibitem [{\citenamefont {Cai}\ and\ \citenamefont {Daskalakis}(2011)}]{cai2011minmax}%
  \BibitemOpen
  \bibfield  {author} {\bibinfo {author} {\bibfnamefont {Y.}~\bibnamefont {Cai}}\ and\ \bibinfo {author} {\bibfnamefont {C.}~\bibnamefont {Daskalakis}},\ }\bibfield  {title} {\bibinfo {title} {On minmax theorems for multiplayer games},\ }in\ \href@noop {} {\emph {\bibinfo {booktitle} {Proceedings of the twenty-second annual ACM-SIAM symposium on Discrete algorithms}}}\ (\bibinfo {organization} {SIAM},\ \bibinfo {year} {2011})\ pp.\ \bibinfo {pages} {217--234}\BibitemShut {NoStop}%
\bibitem [{\citenamefont {Cai}\ \emph {et~al.}(2016)\citenamefont {Cai}, \citenamefont {Candogan}, \citenamefont {Daskalakis},\ and\ \citenamefont {Papadimitriou}}]{cai2016zero}%
  \BibitemOpen
  \bibfield  {author} {\bibinfo {author} {\bibfnamefont {Y.}~\bibnamefont {Cai}}, \bibinfo {author} {\bibfnamefont {O.}~\bibnamefont {Candogan}}, \bibinfo {author} {\bibfnamefont {C.}~\bibnamefont {Daskalakis}},\ and\ \bibinfo {author} {\bibfnamefont {C.}~\bibnamefont {Papadimitriou}},\ }\bibfield  {title} {\bibinfo {title} {Zero-sum polymatrix games: A generalization of minmax},\ }\href@noop {} {\bibfield  {journal} {\bibinfo  {journal} {Mathematics of Operations Research}\ }\textbf {\bibinfo {volume} {41}},\ \bibinfo {pages} {648} (\bibinfo {year} {2016})}\BibitemShut {NoStop}%
\bibitem [{\citenamefont {Sato}\ \emph {et~al.}(2002)\citenamefont {Sato}, \citenamefont {Akiyama},\ and\ \citenamefont {Farmer}}]{sato2002chaos}%
  \BibitemOpen
  \bibfield  {author} {\bibinfo {author} {\bibfnamefont {Y.}~\bibnamefont {Sato}}, \bibinfo {author} {\bibfnamefont {E.}~\bibnamefont {Akiyama}},\ and\ \bibinfo {author} {\bibfnamefont {J.~D.}\ \bibnamefont {Farmer}},\ }\bibfield  {title} {\bibinfo {title} {Chaos in learning a simple two-person game},\ }\href@noop {} {\bibfield  {journal} {\bibinfo  {journal} {Proceedings of the National Academy of Sciences}\ }\textbf {\bibinfo {volume} {99}},\ \bibinfo {pages} {4748} (\bibinfo {year} {2002})}\BibitemShut {NoStop}%
\bibitem [{\citenamefont {Ostrovski}\ and\ \citenamefont {van Strien}(2011)}]{ostrovski2011piecewise}%
  \BibitemOpen
  \bibfield  {author} {\bibinfo {author} {\bibfnamefont {G.}~\bibnamefont {Ostrovski}}\ and\ \bibinfo {author} {\bibfnamefont {S.}~\bibnamefont {van Strien}},\ }\bibfield  {title} {\bibinfo {title} {Piecewise linear {H}amiltonian flows associated to zero-sum games: transition combinatorics and questions on ergodicity},\ }\href@noop {} {\bibfield  {journal} {\bibinfo  {journal} {Regular and Chaotic Dynamics}\ }\textbf {\bibinfo {volume} {16}},\ \bibinfo {pages} {128} (\bibinfo {year} {2011})}\BibitemShut {NoStop}%
\bibitem [{\citenamefont {Van~Strien}(2011)}]{van2011hamiltonian}%
  \BibitemOpen
  \bibfield  {author} {\bibinfo {author} {\bibfnamefont {S.}~\bibnamefont {Van~Strien}},\ }\bibfield  {title} {\bibinfo {title} {Hamiltonian flows with random-walk behaviour originating from zero-sum games and fictitious play},\ }\href@noop {} {\bibfield  {journal} {\bibinfo  {journal} {Nonlinearity}\ }\textbf {\bibinfo {volume} {24}},\ \bibinfo {pages} {1715} (\bibinfo {year} {2011})}\BibitemShut {NoStop}%
\bibitem [{\citenamefont {Balduzzi}\ \emph {et~al.}(2018)\citenamefont {Balduzzi}, \citenamefont {Racaniere}, \citenamefont {Martens}, \citenamefont {Foerster}, \citenamefont {Tuyls},\ and\ \citenamefont {Graepel}}]{balduzzi2018mechanics}%
  \BibitemOpen
  \bibfield  {author} {\bibinfo {author} {\bibfnamefont {D.}~\bibnamefont {Balduzzi}}, \bibinfo {author} {\bibfnamefont {S.}~\bibnamefont {Racaniere}}, \bibinfo {author} {\bibfnamefont {J.}~\bibnamefont {Martens}}, \bibinfo {author} {\bibfnamefont {J.}~\bibnamefont {Foerster}}, \bibinfo {author} {\bibfnamefont {K.}~\bibnamefont {Tuyls}},\ and\ \bibinfo {author} {\bibfnamefont {T.}~\bibnamefont {Graepel}},\ }\bibfield  {title} {\bibinfo {title} {The mechanics of n-player differentiable games},\ }in\ \href@noop {} {\emph {\bibinfo {booktitle} {International Conference on Machine Learning}}}\ (\bibinfo {organization} {PMLR},\ \bibinfo {year} {2018})\ pp.\ \bibinfo {pages} {354--363}\BibitemShut {NoStop}%
\bibitem [{\citenamefont {Letcher}\ \emph {et~al.}(2019)\citenamefont {Letcher}, \citenamefont {Balduzzi}, \citenamefont {Racaniere}, \citenamefont {Martens}, \citenamefont {Foerster}, \citenamefont {Tuyls},\ and\ \citenamefont {Graepel}}]{letcher2019differentiable}%
  \BibitemOpen
  \bibfield  {author} {\bibinfo {author} {\bibfnamefont {A.}~\bibnamefont {Letcher}}, \bibinfo {author} {\bibfnamefont {D.}~\bibnamefont {Balduzzi}}, \bibinfo {author} {\bibfnamefont {S.}~\bibnamefont {Racaniere}}, \bibinfo {author} {\bibfnamefont {J.}~\bibnamefont {Martens}}, \bibinfo {author} {\bibfnamefont {J.}~\bibnamefont {Foerster}}, \bibinfo {author} {\bibfnamefont {K.}~\bibnamefont {Tuyls}},\ and\ \bibinfo {author} {\bibfnamefont {T.}~\bibnamefont {Graepel}},\ }\bibfield  {title} {\bibinfo {title} {Differentiable game mechanics},\ }\href@noop {} {\bibfield  {journal} {\bibinfo  {journal} {Journal of Machine Learning Research}\ }\textbf {\bibinfo {volume} {20}},\ \bibinfo {pages} {1} (\bibinfo {year} {2019})}\BibitemShut {NoStop}%
\bibitem [{\citenamefont {Alishah}\ and\ \citenamefont {Duarte}(2015)}]{alishah2015hamiltonian}%
  \BibitemOpen
  \bibfield  {author} {\bibinfo {author} {\bibfnamefont {H.~N.}\ \bibnamefont {Alishah}}\ and\ \bibinfo {author} {\bibfnamefont {P.}~\bibnamefont {Duarte}},\ }\bibfield  {title} {\bibinfo {title} {Hamiltonian evolutionary games},\ }\href@noop {} {\bibfield  {journal} {\bibinfo  {journal} {Journal of Dynamics and Games}\ }\textbf {\bibinfo {volume} {2}},\ \bibinfo {pages} {33} (\bibinfo {year} {2015})}\BibitemShut {NoStop}%
\bibitem [{\citenamefont {Najafi}(2020)}]{najafi2020conservative}%
  \BibitemOpen
  \bibfield  {author} {\bibinfo {author} {\bibfnamefont {A.~H.}\ \bibnamefont {Najafi}},\ }\bibfield  {title} {\bibinfo {title} {Conservative replicator and {L}otka-{V}olterra equations in the context of {D}irac$\backslash$big-isotropic structures},\ }\href@noop {} {\bibfield  {journal} {\bibinfo  {journal} {Journal of Geometric Mechanics}\ }\textbf {\bibinfo {volume} {12}},\ \bibinfo {pages} {149} (\bibinfo {year} {2020})}\BibitemShut {NoStop}%
\bibitem [{\citenamefont {Cai}\ \emph {et~al.}(2022)\citenamefont {Cai}, \citenamefont {Oikonomou},\ and\ \citenamefont {Zheng}}]{cai2022finitetime}%
  \BibitemOpen
  \bibfield  {author} {\bibinfo {author} {\bibfnamefont {Y.}~\bibnamefont {Cai}}, \bibinfo {author} {\bibfnamefont {A.}~\bibnamefont {Oikonomou}},\ and\ \bibinfo {author} {\bibfnamefont {W.}~\bibnamefont {Zheng}},\ }\bibfield  {title} {\bibinfo {title} {Finite-time last-iterate convergence for learning in multi-player games},\ }in\ \href@noop {} {\emph {\bibinfo {booktitle} {Advances in Neural Information Processing Systems}}}\ (\bibinfo {year} {2022})\BibitemShut {NoStop}%
\bibitem [{\citenamefont {Daskalakis}\ and\ \citenamefont {Panageas}(2019)}]{daskalakis2019last}%
  \BibitemOpen
  \bibfield  {author} {\bibinfo {author} {\bibfnamefont {C.}~\bibnamefont {Daskalakis}}\ and\ \bibinfo {author} {\bibfnamefont {I.}~\bibnamefont {Panageas}},\ }\bibfield  {title} {\bibinfo {title} {Last-iterate convergence: Zero-sum games and constrained min-max optimization},\ }\href@noop {} {\bibfield  {journal} {\bibinfo  {journal} {10th Innovations in Theoretical Computer Science}\ } (\bibinfo {year} {2019})}\BibitemShut {NoStop}%
\bibitem [{\citenamefont {Daskalakis}\ \emph {et~al.}(2018)\citenamefont {Daskalakis}, \citenamefont {Ilyas}, \citenamefont {Syrgkanis},\ and\ \citenamefont {Zeng}}]{daskalakis2018training}%
  \BibitemOpen
  \bibfield  {author} {\bibinfo {author} {\bibfnamefont {C.}~\bibnamefont {Daskalakis}}, \bibinfo {author} {\bibfnamefont {A.}~\bibnamefont {Ilyas}}, \bibinfo {author} {\bibfnamefont {V.}~\bibnamefont {Syrgkanis}},\ and\ \bibinfo {author} {\bibfnamefont {H.}~\bibnamefont {Zeng}},\ }\bibfield  {title} {\bibinfo {title} {Training {GAN}s with optimism},\ }in\ \href@noop {} {\emph {\bibinfo {booktitle} {International Conference on Learning Representations}}}\ (\bibinfo {year} {2018})\BibitemShut {NoStop}%
\bibitem [{\citenamefont {Syrgkanis}\ \emph {et~al.}(2015)\citenamefont {Syrgkanis}, \citenamefont {Agarwal}, \citenamefont {Luo},\ and\ \citenamefont {Schapire}}]{syrgkanis2015fast}%
  \BibitemOpen
  \bibfield  {author} {\bibinfo {author} {\bibfnamefont {V.}~\bibnamefont {Syrgkanis}}, \bibinfo {author} {\bibfnamefont {A.}~\bibnamefont {Agarwal}}, \bibinfo {author} {\bibfnamefont {H.}~\bibnamefont {Luo}},\ and\ \bibinfo {author} {\bibfnamefont {R.~E.}\ \bibnamefont {Schapire}},\ }\bibfield  {title} {\bibinfo {title} {Fast convergence of regularized learning in games},\ }\href@noop {} {\bibfield  {journal} {\bibinfo  {journal} {Advances in Neural Information Processing Systems}\ }\textbf {\bibinfo {volume} {28}} (\bibinfo {year} {2015})}\BibitemShut {NoStop}%
\bibitem [{\citenamefont {Rakhlin}\ and\ \citenamefont {Sridharan}(2013)}]{rakhlin2013optimization}%
  \BibitemOpen
  \bibfield  {author} {\bibinfo {author} {\bibfnamefont {S.}~\bibnamefont {Rakhlin}}\ and\ \bibinfo {author} {\bibfnamefont {K.}~\bibnamefont {Sridharan}},\ }\bibfield  {title} {\bibinfo {title} {Optimization, learning, and games with predictable sequences},\ }\href@noop {} {\bibfield  {journal} {\bibinfo  {journal} {Advances in Neural Information Processing Systems}\ }\textbf {\bibinfo {volume} {26}} (\bibinfo {year} {2013})}\BibitemShut {NoStop}%
\bibitem [{\citenamefont {Lee}\ and\ \citenamefont {Kim}(2021)}]{lee2021fast}%
  \BibitemOpen
  \bibfield  {author} {\bibinfo {author} {\bibfnamefont {S.}~\bibnamefont {Lee}}\ and\ \bibinfo {author} {\bibfnamefont {D.}~\bibnamefont {Kim}},\ }\bibfield  {title} {\bibinfo {title} {Fast extra gradient methods for smooth structured nonconvex-nonconcave minimax problems},\ }\href@noop {} {\bibfield  {journal} {\bibinfo  {journal} {Advances in Neural Information Processing Systems}\ }\textbf {\bibinfo {volume} {34}},\ \bibinfo {pages} {22588} (\bibinfo {year} {2021})}\BibitemShut {NoStop}%
\bibitem [{\citenamefont {Mertikopoulos}\ \emph {et~al.}(2019)\citenamefont {Mertikopoulos}, \citenamefont {Lecouat}, \citenamefont {Zenati}, \citenamefont {Foo}, \citenamefont {Chandrasekhar},\ and\ \citenamefont {Piliouras}}]{mertikopoulos2018optimistic}%
  \BibitemOpen
  \bibfield  {author} {\bibinfo {author} {\bibfnamefont {P.}~\bibnamefont {Mertikopoulos}}, \bibinfo {author} {\bibfnamefont {B.}~\bibnamefont {Lecouat}}, \bibinfo {author} {\bibfnamefont {H.}~\bibnamefont {Zenati}}, \bibinfo {author} {\bibfnamefont {C.-S.}\ \bibnamefont {Foo}}, \bibinfo {author} {\bibfnamefont {V.}~\bibnamefont {Chandrasekhar}},\ and\ \bibinfo {author} {\bibfnamefont {G.}~\bibnamefont {Piliouras}},\ }\bibfield  {title} {\bibinfo {title} {Optimistic mirror descent in saddle-point problems: Going the extra(-gradient) mile},\ }in\ \href@noop {} {\emph {\bibinfo {booktitle} {International Conference on Learning Representations}}}\ (\bibinfo {year} {2019})\BibitemShut {NoStop}%
\bibitem [{\citenamefont {Korpelevich}(1976)}]{Korpelevich1976TheEM}%
  \BibitemOpen
  \bibfield  {author} {\bibinfo {author} {\bibfnamefont {G.~M.}\ \bibnamefont {Korpelevich}},\ }\bibfield  {title} {\bibinfo {title} {The extragradient method for finding saddle points and other problems},\ }\href@noop {} {\bibfield  {journal} {\bibinfo  {journal} {Ekonomika i matematicheskie metody}\ }\textbf {\bibinfo {volume} {12}},\ \bibinfo {pages} {747} (\bibinfo {year} {1976})}\BibitemShut {NoStop}%
\bibitem [{\citenamefont {Hemmat}\ \emph {et~al.}(2023)\citenamefont {Hemmat}, \citenamefont {Mitra}, \citenamefont {Lajoie},\ and\ \citenamefont {Mitliagkas}}]{hemmat2023lead}%
  \BibitemOpen
  \bibfield  {author} {\bibinfo {author} {\bibfnamefont {R.~A.}\ \bibnamefont {Hemmat}}, \bibinfo {author} {\bibfnamefont {A.}~\bibnamefont {Mitra}}, \bibinfo {author} {\bibfnamefont {G.}~\bibnamefont {Lajoie}},\ and\ \bibinfo {author} {\bibfnamefont {I.}~\bibnamefont {Mitliagkas}},\ }\bibfield  {title} {\bibinfo {title} {{LEAD}: Min-max optimization from a physical perspective},\ }in\ \href@noop {} {\emph {\bibinfo {booktitle} {ICML Workshop on New Frontiers in Learning, Control, and Dynamical Systems}}}\ (\bibinfo {year} {2023})\BibitemShut {NoStop}%
\bibitem [{\citenamefont {Zhang}\ and\ \citenamefont {Wang}(2021)}]{zhang2021suboptimality}%
  \BibitemOpen
  \bibfield  {author} {\bibinfo {author} {\bibfnamefont {G.}~\bibnamefont {Zhang}}\ and\ \bibinfo {author} {\bibfnamefont {Y.}~\bibnamefont {Wang}},\ }\bibfield  {title} {\bibinfo {title} {On the suboptimality of negative momentum for minimax optimization},\ }in\ \href@noop {} {\emph {\bibinfo {booktitle} {International Conference on Artificial Intelligence and Statistics}}}\ (\bibinfo {organization} {PMLR},\ \bibinfo {year} {2021})\ pp.\ \bibinfo {pages} {2098--2106}\BibitemShut {NoStop}%
\bibitem [{\citenamefont {Kovachki}\ and\ \citenamefont {Stuart}(2021)}]{kovachki2021continuous}%
  \BibitemOpen
  \bibfield  {author} {\bibinfo {author} {\bibfnamefont {N.~B.}\ \bibnamefont {Kovachki}}\ and\ \bibinfo {author} {\bibfnamefont {A.~M.}\ \bibnamefont {Stuart}},\ }\bibfield  {title} {\bibinfo {title} {Continuous time analysis of momentum methods},\ }\href@noop {} {\bibfield  {journal} {\bibinfo  {journal} {Journal of Machine Learning Research}\ }\textbf {\bibinfo {volume} {22}},\ \bibinfo {pages} {1} (\bibinfo {year} {2021})}\BibitemShut {NoStop}%
\bibitem [{\citenamefont {Gidel}\ \emph {et~al.}(2019)\citenamefont {Gidel}, \citenamefont {Hemmat}, \citenamefont {Pezeshki}, \citenamefont {Le~Priol}, \citenamefont {Huang}, \citenamefont {Lacoste-Julien},\ and\ \citenamefont {Mitliagkas}}]{gidel2019negative}%
  \BibitemOpen
  \bibfield  {author} {\bibinfo {author} {\bibfnamefont {G.}~\bibnamefont {Gidel}}, \bibinfo {author} {\bibfnamefont {R.~A.}\ \bibnamefont {Hemmat}}, \bibinfo {author} {\bibfnamefont {M.}~\bibnamefont {Pezeshki}}, \bibinfo {author} {\bibfnamefont {R.}~\bibnamefont {Le~Priol}}, \bibinfo {author} {\bibfnamefont {G.}~\bibnamefont {Huang}}, \bibinfo {author} {\bibfnamefont {S.}~\bibnamefont {Lacoste-Julien}},\ and\ \bibinfo {author} {\bibfnamefont {I.}~\bibnamefont {Mitliagkas}},\ }\bibfield  {title} {\bibinfo {title} {Negative momentum for improved game dynamics},\ }in\ \href@noop {} {\emph {\bibinfo {booktitle} {The 22nd International Conference on Artificial Intelligence and Statistics}}}\ (\bibinfo {organization} {PMLR},\ \bibinfo {year} {2019})\ pp.\ \bibinfo {pages} {1802--1811}\BibitemShut {NoStop}%
\bibitem [{\citenamefont {Polyak}(1964)}]{polyak1964methods}%
  \BibitemOpen
  \bibfield  {author} {\bibinfo {author} {\bibfnamefont {B.~T.}\ \bibnamefont {Polyak}},\ }\bibfield  {title} {\bibinfo {title} {Some methods of speeding up the convergence of iteration methods},\ }\href@noop {} {\bibfield  {journal} {\bibinfo  {journal} {USSR Computational Mathematics and Mathematical Physics}\ }\textbf {\bibinfo {volume} {4}},\ \bibinfo {pages} {1} (\bibinfo {year} {1964})}\BibitemShut {NoStop}%
\bibitem [{\citenamefont {Feng}\ \emph {et~al.}(2023)\citenamefont {Feng}, \citenamefont {Fu}, \citenamefont {Hu}, \citenamefont {Li}, \citenamefont {Panageas}, \citenamefont {Wang} \emph {et~al.}}]{feng2023last}%
  \BibitemOpen
  \bibfield  {author} {\bibinfo {author} {\bibfnamefont {Y.}~\bibnamefont {Feng}}, \bibinfo {author} {\bibfnamefont {H.}~\bibnamefont {Fu}}, \bibinfo {author} {\bibfnamefont {Q.}~\bibnamefont {Hu}}, \bibinfo {author} {\bibfnamefont {P.}~\bibnamefont {Li}}, \bibinfo {author} {\bibfnamefont {I.}~\bibnamefont {Panageas}}, \bibinfo {author} {\bibfnamefont {X.}~\bibnamefont {Wang}}, \emph {et~al.},\ }\bibfield  {title} {\bibinfo {title} {On the last-iterate convergence in time-varying zero-sum games: Extra gradient succeeds where optimism fails},\ }\href@noop {} {\bibfield  {journal} {\bibinfo  {journal} {Advances in Neural Information Processing Systems}\ }\textbf {\bibinfo {volume} {36}},\ \bibinfo {pages} {21933} (\bibinfo {year} {2023})}\BibitemShut {NoStop}%
\bibitem [{\citenamefont {Feng}\ \emph {et~al.}(2024)\citenamefont {Feng}, \citenamefont {Li}, \citenamefont {Panageas},\ and\ \citenamefont {Wang}}]{feng2024last}%
  \BibitemOpen
  \bibfield  {author} {\bibinfo {author} {\bibfnamefont {Y.}~\bibnamefont {Feng}}, \bibinfo {author} {\bibfnamefont {P.}~\bibnamefont {Li}}, \bibinfo {author} {\bibfnamefont {I.}~\bibnamefont {Panageas}},\ and\ \bibinfo {author} {\bibfnamefont {X.}~\bibnamefont {Wang}},\ }\bibfield  {title} {\bibinfo {title} {Last-iterate convergence separation between extra-gradient and optimism in constrained periodic games},\ }in\ \href@noop {} {\emph {\bibinfo {booktitle} {Uncertainty in Artificial Intelligence}}}\ (\bibinfo {organization} {PMLR},\ \bibinfo {year} {2024})\ pp.\ \bibinfo {pages} {1339--1370}\BibitemShut {NoStop}%
\bibitem [{\citenamefont {Shalev-Shwartz}(2012)}]{shalev2012online}%
  \BibitemOpen
  \bibfield  {author} {\bibinfo {author} {\bibfnamefont {S.}~\bibnamefont {Shalev-Shwartz}},\ }\bibfield  {title} {\bibinfo {title} {Online learning and online convex optimization},\ }\href@noop {} {\bibfield  {journal} {\bibinfo  {journal} {Foundations and Trends{\textregistered} in Machine Learning}\ }\textbf {\bibinfo {volume} {4}},\ \bibinfo {pages} {107} (\bibinfo {year} {2012})}\BibitemShut {NoStop}%
\bibitem [{\citenamefont {Taylor}\ and\ \citenamefont {Jonker}(1978)}]{taylor1978evolutionary}%
  \BibitemOpen
  \bibfield  {author} {\bibinfo {author} {\bibfnamefont {P.~D.}\ \bibnamefont {Taylor}}\ and\ \bibinfo {author} {\bibfnamefont {L.~B.}\ \bibnamefont {Jonker}},\ }\bibfield  {title} {\bibinfo {title} {Evolutionary stable strategies and game dynamics},\ }\href@noop {} {\bibfield  {journal} {\bibinfo  {journal} {Mathematical biosciences}\ }\textbf {\bibinfo {volume} {40}},\ \bibinfo {pages} {145} (\bibinfo {year} {1978})}\BibitemShut {NoStop}%
\bibitem [{\citenamefont {Taylor}(1979)}]{taylor1979evolutionarily}%
  \BibitemOpen
  \bibfield  {author} {\bibinfo {author} {\bibfnamefont {P.~D.}\ \bibnamefont {Taylor}},\ }\bibfield  {title} {\bibinfo {title} {Evolutionarily stable strategies with two types of player},\ }\href@noop {} {\bibfield  {journal} {\bibinfo  {journal} {Journal of applied probability}\ }\textbf {\bibinfo {volume} {16}},\ \bibinfo {pages} {76} (\bibinfo {year} {1979})}\BibitemShut {NoStop}%
\bibitem [{\citenamefont {Schuster}\ and\ \citenamefont {Sigmund}(1983)}]{schuster1983replicator}%
  \BibitemOpen
  \bibfield  {author} {\bibinfo {author} {\bibfnamefont {P.}~\bibnamefont {Schuster}}\ and\ \bibinfo {author} {\bibfnamefont {K.}~\bibnamefont {Sigmund}},\ }\bibfield  {title} {\bibinfo {title} {Replicator dynamics},\ }\href@noop {} {\bibfield  {journal} {\bibinfo  {journal} {Journal of theoretical biology}\ }\textbf {\bibinfo {volume} {100}},\ \bibinfo {pages} {533} (\bibinfo {year} {1983})}\BibitemShut {NoStop}%
\bibitem [{\citenamefont {Friedman}(1991)}]{friedman1991evolutionary}%
  \BibitemOpen
  \bibfield  {author} {\bibinfo {author} {\bibfnamefont {D.}~\bibnamefont {Friedman}},\ }\bibfield  {title} {\bibinfo {title} {Evolutionary games in economics},\ }\href@noop {} {\bibfield  {journal} {\bibinfo  {journal} {Econometrica: journal of the econometric society}\ ,\ \bibinfo {pages} {637}} (\bibinfo {year} {1991})}\BibitemShut {NoStop}%
\bibitem [{\citenamefont {Nagurney}\ and\ \citenamefont {Zhang}(1997)}]{nagurney1997projected}%
  \BibitemOpen
  \bibfield  {author} {\bibinfo {author} {\bibfnamefont {A.}~\bibnamefont {Nagurney}}\ and\ \bibinfo {author} {\bibfnamefont {D.}~\bibnamefont {Zhang}},\ }\bibfield  {title} {\bibinfo {title} {Projected dynamical systems in the formulation, stability analysis, and computation of fixed-demand traffic network equilibria},\ }\href@noop {} {\bibfield  {journal} {\bibinfo  {journal} {Transportation Science}\ }\textbf {\bibinfo {volume} {31}},\ \bibinfo {pages} {147} (\bibinfo {year} {1997})}\BibitemShut {NoStop}%
\bibitem [{\citenamefont {Lahkar}\ and\ \citenamefont {Sandholm}(2008)}]{lahkar2008projection}%
  \BibitemOpen
  \bibfield  {author} {\bibinfo {author} {\bibfnamefont {R.}~\bibnamefont {Lahkar}}\ and\ \bibinfo {author} {\bibfnamefont {W.~H.}\ \bibnamefont {Sandholm}},\ }\bibfield  {title} {\bibinfo {title} {The projection dynamic and the geometry of population games},\ }\href@noop {} {\bibfield  {journal} {\bibinfo  {journal} {Games and Economic Behavior}\ }\textbf {\bibinfo {volume} {64}},\ \bibinfo {pages} {565} (\bibinfo {year} {2008})}\BibitemShut {NoStop}%
\bibitem [{\citenamefont {Sandholm}\ \emph {et~al.}(2008)\citenamefont {Sandholm}, \citenamefont {Dokumac{\i}},\ and\ \citenamefont {Lahkar}}]{sandholm2008projection}%
  \BibitemOpen
  \bibfield  {author} {\bibinfo {author} {\bibfnamefont {W.~H.}\ \bibnamefont {Sandholm}}, \bibinfo {author} {\bibfnamefont {E.}~\bibnamefont {Dokumac{\i}}},\ and\ \bibinfo {author} {\bibfnamefont {R.}~\bibnamefont {Lahkar}},\ }\bibfield  {title} {\bibinfo {title} {The projection dynamic and the replicator dynamic},\ }\href@noop {} {\bibfield  {journal} {\bibinfo  {journal} {Games and Economic Behavior}\ }\textbf {\bibinfo {volume} {64}},\ \bibinfo {pages} {666} (\bibinfo {year} {2008})}\BibitemShut {NoStop}%
\bibitem [{\citenamefont {Mertikopoulos}\ and\ \citenamefont {Sandholm}(2016)}]{mertikopoulos2016learning}%
  \BibitemOpen
  \bibfield  {author} {\bibinfo {author} {\bibfnamefont {P.}~\bibnamefont {Mertikopoulos}}\ and\ \bibinfo {author} {\bibfnamefont {W.~H.}\ \bibnamefont {Sandholm}},\ }\bibfield  {title} {\bibinfo {title} {Learning in games via reinforcement and regularization},\ }\href@noop {} {\bibfield  {journal} {\bibinfo  {journal} {Mathematics of Operations Research}\ }\textbf {\bibinfo {volume} {41}},\ \bibinfo {pages} {1297} (\bibinfo {year} {2016})}\BibitemShut {NoStop}%
\bibitem [{\citenamefont {Mertikopoulos}\ and\ \citenamefont {Zhou}(2019)}]{mertikopoulos2019learning}%
  \BibitemOpen
  \bibfield  {author} {\bibinfo {author} {\bibfnamefont {P.}~\bibnamefont {Mertikopoulos}}\ and\ \bibinfo {author} {\bibfnamefont {Z.}~\bibnamefont {Zhou}},\ }\bibfield  {title} {\bibinfo {title} {Learning in games with continuous action sets and unknown payoff functions},\ }\href@noop {} {\bibfield  {journal} {\bibinfo  {journal} {Mathematical Programming}\ }\textbf {\bibinfo {volume} {173}},\ \bibinfo {pages} {465} (\bibinfo {year} {2019})}\BibitemShut {NoStop}%
\bibitem [{\citenamefont {Nash}(1951)}]{nash1951non}%
  \BibitemOpen
  \bibfield  {author} {\bibinfo {author} {\bibfnamefont {J.}~\bibnamefont {Nash}},\ }\bibfield  {title} {\bibinfo {title} {Non-cooperative games},\ }\href@noop {} {\bibfield  {journal} {\bibinfo  {journal} {Annals of Mathematics}\ }\textbf {\bibinfo {volume} {54}},\ \bibinfo {pages} {286} (\bibinfo {year} {1951})}\BibitemShut {NoStop}%
\bibitem [{\citenamefont {Tewolde}\ \emph {et~al.}(2025)\citenamefont {Tewolde}, \citenamefont {Zhang}, \citenamefont {Oesterheld}, \citenamefont {Sandholm},\ and\ \citenamefont {Conitzer}}]{tewolde2025computing}%
  \BibitemOpen
  \bibfield  {author} {\bibinfo {author} {\bibfnamefont {E.}~\bibnamefont {Tewolde}}, \bibinfo {author} {\bibfnamefont {B.~H.}\ \bibnamefont {Zhang}}, \bibinfo {author} {\bibfnamefont {C.}~\bibnamefont {Oesterheld}}, \bibinfo {author} {\bibfnamefont {T.}~\bibnamefont {Sandholm}},\ and\ \bibinfo {author} {\bibfnamefont {V.}~\bibnamefont {Conitzer}},\ }\bibfield  {title} {\bibinfo {title} {Computing game symmetries and equilibria that respect them},\ }\href@noop {} {\bibfield  {journal} {\bibinfo  {journal} {Proceedings of the AAAI Conference on Artificial Intelligence}\ }\textbf {\bibinfo {volume} {39}},\ \bibinfo {pages} {14148} (\bibinfo {year} {2025})}\BibitemShut {NoStop}%
\bibitem [{\citenamefont {Eisert}\ \emph {et~al.}(1999)\citenamefont {Eisert}, \citenamefont {Wilkens},\ and\ \citenamefont {Lewenstein}}]{eisert1999quantum}%
  \BibitemOpen
  \bibfield  {author} {\bibinfo {author} {\bibfnamefont {J.}~\bibnamefont {Eisert}}, \bibinfo {author} {\bibfnamefont {M.}~\bibnamefont {Wilkens}},\ and\ \bibinfo {author} {\bibfnamefont {M.}~\bibnamefont {Lewenstein}},\ }\bibfield  {title} {\bibinfo {title} {Quantum games and quantum strategies},\ }\href@noop {} {\bibfield  {journal} {\bibinfo  {journal} {Physical Review Letters}\ }\textbf {\bibinfo {volume} {83}},\ \bibinfo {pages} {3077} (\bibinfo {year} {1999})}\BibitemShut {NoStop}%
\bibitem [{\citenamefont {Eisert}\ and\ \citenamefont {Wilkens}(2000)}]{eisert2000quantum}%
  \BibitemOpen
  \bibfield  {author} {\bibinfo {author} {\bibfnamefont {J.}~\bibnamefont {Eisert}}\ and\ \bibinfo {author} {\bibfnamefont {M.}~\bibnamefont {Wilkens}},\ }\bibfield  {title} {\bibinfo {title} {Quantum games},\ }\href@noop {} {\bibfield  {journal} {\bibinfo  {journal} {Journal of Modern Optics}\ }\textbf {\bibinfo {volume} {47}},\ \bibinfo {pages} {2543} (\bibinfo {year} {2000})}\BibitemShut {NoStop}%
\bibitem [{\citenamefont {Jain}\ \emph {et~al.}(2022)\citenamefont {Jain}, \citenamefont {Piliouras},\ and\ \citenamefont {Sim}}]{jain2022matrix}%
  \BibitemOpen
  \bibfield  {author} {\bibinfo {author} {\bibfnamefont {R.}~\bibnamefont {Jain}}, \bibinfo {author} {\bibfnamefont {G.}~\bibnamefont {Piliouras}},\ and\ \bibinfo {author} {\bibfnamefont {R.}~\bibnamefont {Sim}},\ }\bibfield  {title} {\bibinfo {title} {Matrix multiplicative weights updates in quantum zero-sum games: Conservation laws \& recurrence},\ }\href@noop {} {\bibfield  {journal} {\bibinfo  {journal} {Advances in Neural Information Processing Systems}\ }\textbf {\bibinfo {volume} {35}},\ \bibinfo {pages} {4123} (\bibinfo {year} {2022})}\BibitemShut {NoStop}%
\bibitem [{\citenamefont {Lotidis}\ \emph {et~al.}(2023)\citenamefont {Lotidis}, \citenamefont {Mertikopoulos},\ and\ \citenamefont {Bambos}}]{lotidis2023learning}%
  \BibitemOpen
  \bibfield  {author} {\bibinfo {author} {\bibfnamefont {K.}~\bibnamefont {Lotidis}}, \bibinfo {author} {\bibfnamefont {P.}~\bibnamefont {Mertikopoulos}},\ and\ \bibinfo {author} {\bibfnamefont {N.}~\bibnamefont {Bambos}},\ }\bibfield  {title} {\bibinfo {title} {Learning in quantum games},\ }\href@noop {} {\bibfield  {journal} {\bibinfo  {journal} {arXiv preprint arXiv:2302.02333}\ } (\bibinfo {year} {2023})}\BibitemShut {NoStop}%
\bibitem [{\citenamefont {Lin}\ \emph {et~al.}(2025)\citenamefont {Lin}, \citenamefont {Piliouras}, \citenamefont {Sim},\ and\ \citenamefont {Varvitsiotis}}]{lin2025learning}%
  \BibitemOpen
  \bibfield  {author} {\bibinfo {author} {\bibfnamefont {W.}~\bibnamefont {Lin}}, \bibinfo {author} {\bibfnamefont {G.}~\bibnamefont {Piliouras}}, \bibinfo {author} {\bibfnamefont {R.}~\bibnamefont {Sim}},\ and\ \bibinfo {author} {\bibfnamefont {A.}~\bibnamefont {Varvitsiotis}},\ }\bibfield  {title} {\bibinfo {title} {Learning in quantum common-interest games and the separability problem},\ }\href@noop {} {\bibfield  {journal} {\bibinfo  {journal} {Quantum}\ }\textbf {\bibinfo {volume} {9}},\ \bibinfo {pages} {1689} (\bibinfo {year} {2025})}\BibitemShut {NoStop}%
\bibitem [{\citenamefont {Bertsekas}(1997)}]{bertsekas1997nonlinear}%
  \BibitemOpen
  \bibfield  {author} {\bibinfo {author} {\bibfnamefont {D.~P.}\ \bibnamefont {Bertsekas}},\ }\bibfield  {title} {\bibinfo {title} {Nonlinear programming},\ }\href@noop {} {\bibfield  {journal} {\bibinfo  {journal} {Journal of the Operational Research Society}\ }\textbf {\bibinfo {volume} {48}},\ \bibinfo {pages} {334} (\bibinfo {year} {1997})}\BibitemShut {NoStop}%
\bibitem [{\citenamefont {Bergmann}\ and\ \citenamefont {Herzog}(2019)}]{bergmann2019intrinsic}%
  \BibitemOpen
  \bibfield  {author} {\bibinfo {author} {\bibfnamefont {R.}~\bibnamefont {Bergmann}}\ and\ \bibinfo {author} {\bibfnamefont {R.}~\bibnamefont {Herzog}},\ }\bibfield  {title} {\bibinfo {title} {Intrinsic formulation of {KKT} conditions and constraint qualifications on smooth manifolds},\ }\href@noop {} {\bibfield  {journal} {\bibinfo  {journal} {SIAM Journal on Optimization}\ }\textbf {\bibinfo {volume} {29}},\ \bibinfo {pages} {2423} (\bibinfo {year} {2019})}\BibitemShut {NoStop}%
\bibitem [{\citenamefont {Libermann}\ and\ \citenamefont {Marle}(1987)}]{libermann}%
  \BibitemOpen
  \bibfield  {author} {\bibinfo {author} {\bibfnamefont {P.}~\bibnamefont {Libermann}}\ and\ \bibinfo {author} {\bibfnamefont {C.-M.}\ \bibnamefont {Marle}},\ }\href@noop {} {\emph {\bibinfo {title} {Symplectic Geometry and Analytical Mechanics}}}\ (\bibinfo  {publisher} {Springer Dordrecht},\ \bibinfo {year} {1987})\BibitemShut {NoStop}%
\bibitem [{\citenamefont {Arnol'd}(1989)}]{arnold}%
  \BibitemOpen
  \bibfield  {author} {\bibinfo {author} {\bibfnamefont {V.~I.}\ \bibnamefont {Arnol'd}},\ }\href@noop {} {\emph {\bibinfo {title} {Mathematical Methods of Classical Mechanics}}}\ (\bibinfo  {publisher} {Springer},\ \bibinfo {year} {1989})\BibitemShut {NoStop}%
\bibitem [{\citenamefont {Kabir}\ and\ \citenamefont {Tanimoto}(2021)}]{kabir2021role}%
  \BibitemOpen
  \bibfield  {author} {\bibinfo {author} {\bibfnamefont {K.~M.~A.}\ \bibnamefont {Kabir}}\ and\ \bibinfo {author} {\bibfnamefont {J.}~\bibnamefont {Tanimoto}},\ }\bibfield  {title} {\bibinfo {title} {The role of pairwise nonlinear evolutionary dynamics in the rock--paper--scissors game with noise},\ }\href@noop {} {\bibfield  {journal} {\bibinfo  {journal} {Applied Mathematics and Computation}\ }\textbf {\bibinfo {volume} {394}},\ \bibinfo {pages} {125767} (\bibinfo {year} {2021})}\BibitemShut {NoStop}%
\bibitem [{\citenamefont {Tenorio}\ \emph {et~al.}(2022)\citenamefont {Tenorio}, \citenamefont {Rangel},\ and\ \citenamefont {Menezes}}]{tenorio2022adaptive}%
  \BibitemOpen
  \bibfield  {author} {\bibinfo {author} {\bibfnamefont {M.}~\bibnamefont {Tenorio}}, \bibinfo {author} {\bibfnamefont {E.}~\bibnamefont {Rangel}},\ and\ \bibinfo {author} {\bibfnamefont {J.}~\bibnamefont {Menezes}},\ }\bibfield  {title} {\bibinfo {title} {Adaptive movement strategy in rock-paper-scissors models},\ }\href@noop {} {\bibfield  {journal} {\bibinfo  {journal} {Chaos, Solitons \& Fractals}\ }\textbf {\bibinfo {volume} {162}},\ \bibinfo {pages} {112430} (\bibinfo {year} {2022})}\BibitemShut {NoStop}%
\bibitem [{\citenamefont {Justino}\ \emph {et~al.}(2025)\citenamefont {Justino}, \citenamefont {de~Oliveira},\ and\ \citenamefont {dos Santos}}]{justino2025critical}%
  \BibitemOpen
  \bibfield  {author} {\bibinfo {author} {\bibfnamefont {R.~R.}\ \bibnamefont {Justino}}, \bibinfo {author} {\bibfnamefont {B.~F.}\ \bibnamefont {de~Oliveira}},\ and\ \bibinfo {author} {\bibfnamefont {F.~E.~A.}\ \bibnamefont {dos Santos}},\ }\bibfield  {title} {\bibinfo {title} {Critical phenomena in the rock-paper-scissors model},\ }\href@noop {} {\bibfield  {journal} {\bibinfo  {journal} {Physica A: Statistical Mechanics and its Applications}\ ,\ \bibinfo {pages} {131018}} (\bibinfo {year} {2025})}\BibitemShut {NoStop}%
\bibitem [{\citenamefont {Chen}\ \emph {et~al.}(2018)\citenamefont {Chen}, \citenamefont {Rubanova}, \citenamefont {Bettencourt},\ and\ \citenamefont {Duvenaud}}]{chen2018neuralode}%
  \BibitemOpen
  \bibfield  {author} {\bibinfo {author} {\bibfnamefont {R.~T.~Q.}\ \bibnamefont {Chen}}, \bibinfo {author} {\bibfnamefont {Y.}~\bibnamefont {Rubanova}}, \bibinfo {author} {\bibfnamefont {J.}~\bibnamefont {Bettencourt}},\ and\ \bibinfo {author} {\bibfnamefont {D.}~\bibnamefont {Duvenaud}},\ }\bibfield  {title} {\bibinfo {title} {Neural ordinary differential equations},\ }\href@noop {} {\bibfield  {journal} {\bibinfo  {journal} {Advances in Neural Information Processing Systems}\ } (\bibinfo {year} {2018})}\BibitemShut {NoStop}%
\bibitem [{\citenamefont {Chen}(2018)}]{torchdiffeq}%
  \BibitemOpen
  \bibfield  {author} {\bibinfo {author} {\bibfnamefont {R.~T.~Q.}\ \bibnamefont {Chen}},\ }\href {https://github.com/rtqichen/torchdiffeq} {\bibinfo {title} {torchdiffeq}},\ \bibinfo {howpublished} {\url{https://github.com/rtqichen/torchdiffeq}} (\bibinfo {year} {2018})\BibitemShut {NoStop}%
\end{thebibliography}%

\end{document}